%% file: main.tex
\documentclass[11pt]{article}

\usepackage[numbers]{natbib}
\usepackage[utf8]{inputenc}
\usepackage{graphicx}
\usepackage{amsmath}
\usepackage{amssymb}
\usepackage{amsthm}
\usepackage{setspace}
\usepackage{xcolor}
\usepackage{float}
\usepackage[shortlabels]{enumitem}
\usepackage{hyperref}
\usepackage[normalem]{ulem}
\usepackage{float}
\usepackage{tikz}
\usetikzlibrary{positioning}
\usepackage{fullpage}
\usepackage{authblk}
\usepackage{subcaption}
\newcommand{\mc}[1]{\ensuremath{\mathcal{#1}}}
\renewcommand{\v}[1]{\ensuremath{\boldsymbol{\mathrm{#1}}}}
\renewcommand{\b}[1]{\ensuremath{\overline{#1}}}
\newcommand{\reals}{\ensuremath{\mathbb{R}}}
\renewcommand{\P}{\ensuremath{\mathbb{P}}}
\newcommand{\Q}{\ensuremath{\mathbb{Q}}}
\newcommand{\E}{\ensuremath{\mathbb{E}}}
\newcommand{\GTE}{\ensuremath{\mathsf{GTE}}}
\newcommand{\UB}{\ensuremath{\mathsf{UB}}}

\newcommand{\ADE}{\ensuremath{\mathsf{ADE}}}
\newcommand{\FPP}{\ensuremath{\mathsf{FPP}}}
\newcommand{\FNP}{\ensuremath{\mathsf{FNP}}}

\newcommand{\Var}{\ensuremath{\mathrm{Var}}}
\newcommand{\Cov}{\ensuremath{\mathrm{Cov}}}

\newcommand{\DM}{\ensuremath{\widehat{\GTE}}}

\definecolor{cobalt}{rgb}{0.0, 0.28, 0.67}

\ifdefined\DRAFT
\newcommand{\anushka}[1]{{\color{red} \textbf{[A: #1]}}}

\newcommand{\delam}[1]{{\color{red} \sout{#1}}}
\newcommand{\rj}[1]{{\color{blue} \textbf{[R: #1]}}}

\newcommand{\delrj}[1]{{\color{blue} \sout{#1}}}
\newcommand{\gw}[1]{{\color{cyan} \textbf{[G: #1]}}}

\newcommand{\delgw}[1]{{\color{cyan} \sout{#1}}}
\newcommand{\hl}[1]{{\color{magenta} \textbf{[H: #1]}}}

\newcommand{\delhl}[1]{{\color{magenta} \sout{#1}}}
\else
\newcommand{\anushka}[1]{}

\newcommand{\delam}[1]{}
\newcommand{\rj}[1]{}

\newcommand{\delrj}[1]{}
\newcommand{\gw}[1]{}

\newcommand{\delgw}[1]{}
\newcommand{\hl}[1]{}

\newcommand{\delhl}[1]{}

\fi

\newtheorem{theorem}{Theorem}[section]
\newtheorem{proposition}{Proposition}[section]
\newtheorem{lemma}{Lemma}[section]

\newtheorem{corollary}{Corollary}[section]
\newtheorem{assumption}{Assumption}[section]
\newtheorem{definition}{Definition}[section]

\theoremstyle{definition}
\newtheorem{example}{Example}[section]

\title{When Does Interference Matter?\\Decision-Making in Platform Experiments}
\author[1]{Ramesh Johari}
\author[2]{Hannah Li}
\author[1]{Anushka Murthy}
\author[3]{Gabriel Y.~Weintraub}
\affil[1]{Management Science and Engineering, Stanford University}
\affil[2]{Columbia Business School}
\affil[3]{Stanford Graduate School of Business}
\date{\today}

\begin{document}

\maketitle

\begin{abstract}
This paper investigates decision-making in A/B experiments for online platforms and marketplaces. In such settings, due to constraints on inventory, A/B experiments typically lead to biased estimators because of {\em interference} between treatment and control groups; this phenomenon has been well studied in recent literature.  By contrast, there has been relatively little discussion of the impact of interference on decision-making.  In this paper, we analyze a benchmark Markovian model of an inventory-constrained platform, where arriving customers book listings that are limited in supply; our analysis builds on a self-contained analysis of general A/B experiments for Markov chains.  We focus on the commonly used frequentist hypothesis testing approach for making launch decisions based on data from customer-randomized experiments, and we study the impact of interference on (1) false positive probability and (2) statistical power.

We obtain three main findings. First, we show that for {\em sign-consistent} treatments---i.e., those where the treatment changes booking probabilities in the same direction relative to control for all states of inventory availability---the false positive probability of a test statistic using the standard difference-in-means estimator with a corresponding na\"{i}ve variance estimator is correctly controlled. We obtain this result by characterizing the false positive probability via analysis of A/A experiments with arbitrary dependence structures, which may be of independent interest. Second, we demonstrate that for sign-consistent treatments in realistic settings, the statistical power of this naïve approach is higher than that of any similar pipeline using a debiased estimator.  Taken together, these two findings suggest that platforms may be better off not debiasing when treatments are sign-consistent: using the naïve approach, the false positive probability is controlled, and the power is strictly higher.  Third, using numerics, we investigate false positive probability and statistical power when treatments are sign-inconsistent, and we show that in principle, the performance of the naïve approach can be arbitrarily worse in such cases.

Our results have important implications for the practical deployment of debiasing strategies for A/B experiments.  In particular, they highlight the need for platforms to carefully define their objectives and understand the nature of their interventions when determining appropriate estimation and decision-making approaches.  Notably, when interventions are sign-consistent, the platform may actually be worse off by pursuing a debiased decision-making approach.
\end{abstract}

\input{Sections/intro}
\input{Sections/related}

\input{Sections/cc}

\input{Sections/fpp}

\input{Sections/power}
\input{Sections/sims}

\input{Sections/conclusion}

\section{Acknowledgments}

This work was supported in part by the National Science Foundation and Stanford Data Science.  We are grateful for helpful conversations with Mohsen Bayati, Peter Coles, Alex Deng, John Duchi, David Holtz, Guido Imbens, Samir Khan, Lihua Lei, Inessa Liskovich, Ruben Lobel, Nian Si, Johan Ugander, Stefan Wager, and Martin Wainwright, as well as comments and feedback from seminar participants at the INFORMS Annual Meeting, INFORMS Revenue Management and Pricing Conference, the MIT Conference on Digital Experimentation, the Harvard Business School Workshop on Experimentation and Evaluation in Operations, the Marketplace Innovation Workshop, the Stanford Causal Science Center Conference, and departmental seminars at Georgia Tech, Harvard, MIT, and Stanford.

\bibliographystyle{abbrv}
\bibliography{refs}

\appendix

\input{Sections/app_MC}

\input{Sections/app_proofs}

\input{Sections/app_sims}

\end{document}

%% file: Sections/intro.tex
\section{Introduction}
\label{sec:intro}

Online platforms and marketplaces routinely use randomized controlled trials, also known as {\em A/B experiments} to test changes to their market design, such as the introduction of new algorithms, new features, or pricing and fee changes.  For example, an online marketplace for lodging might use A/B experiments to test the impact of a change in the description of listings on overall bookings; or a ride-sharing marketplace might use A/B experiments to test the impact of lowering the surge price on overall rides.  In a typical approach for such experiments, customers (i.e., users looking to purchase) are randomized to {\em treatment} (i.e., the new product feature) or {\em control} (i.e., the existing product feature), using a simple i.i.d.~Bernoulli randomization scheme.  Of interest to the platform is the {\em global treatment effect} ($\GTE$): the difference in aggregate sales if the entire market is in treatment as compared to control.  Once outcomes (e.g., the number of bookings or rides) are collected after the experiment, a simple difference-in-means (DM) estimator (denoted $\DM$) is used to estimate the $\GTE$. 

A key challenge for these platforms is that this simple estimation approach suffers from {\em interference} between treatment and control units, because of the constrained inventory on each side of the platform.  For example, consider an A/B experiment in an online marketplace for lodging that randomizes arriving customers to treatment or control; this is commonly referred to as {\em user-level randomization} in industry, and such experiments are a typical approach to testing a wide range of platform design features.  Because these customers interact with the same inventory of listings, customers' booking outcomes impact the state of the market as seen by subsequent arriving customers.  This interference effect leads $\DM$ to be {\em biased} relative to the $\GTE$.  Extensive recent literature has investigated conditions characterizing the magnitude of this bias, conditions under which it is magnified, and methods for {\em debiasing} (i.e., combinations of designs and/or estimators that estimate $\GTE$ with low bias).  See Section \ref{sec:related} for related references. 

Despite this extensive attention on bias and debiasing of estimates of the treatment effect, thus far there has been limited investigation of the impact of bias on {\em decision-making}.  In the typical use case of A/B experiments, beyond estimation of the treatment effect, platforms are also making a decision about whether or not to launch the treatment change being tested to the entire marketplace.  {\em What is the impact of interference on these decisions?}  Our paper focuses on this question.

Our main contribution is an analysis and characterization of the impact of interference in a benchmark model of decision-making.  As we show, for a wide range of interventions, despite the presence of interference a platform may actually be {\em no worse off} (and possibly better off, in a sense we make precise) if they make decisions using the ``na\"{i}ve'' DM estimation approach, along with the associated classical variance estimator.  We further discuss conditions under which debiasing can be essential to make correct decisions.

We consider a frequentist decision-making process that is quite commonly used after A/B experiments in online platforms, based on frequentist hypothesis testing.  After calculating $\DM$, the platform computes an associated na\"{i}ve variance estimator assuming observations were i.i.d.~(i.e., that there is no interference or correlation between observations); we denote this variance estimator $\widehat{\Var}$.  (Note that because interference is present, in general $\widehat{\Var}$ will also be biased for the true variance of $\DM$.)  Using these quantities, the platform forms the standard {\em t-test statistic}, $T = \DM/\sqrt{\widehat{\Var}}$.  The platform then supposes that under the {\em null hypothesis} $H_0$ that $\GTE = 0$, $\hat{T}$ is approximately distributed according to a standard normal random variable; this assumption is valid if data is independent across observations (and the sample size is sufficiently large), but not necessarily in the presence of interference.  In particular, $\hat{T}$ is then compared to the tail quantiles of a standard normal random variable, and $H_0$ is rejected if $|\hat{T}|$ is sufficiently large.  Commonly, $H_0$ is rejected if $|\hat{T}| > 1.96$, corresponding to 95\% statistical significance; this practice is equivalent to rejecting $H_0$ if zero is not in the 95\% normal confidence interval around $\DM$. (Note that in practice, the platform will typically only launch the intervention if $H_0$ is rejected, {\em and} the estimate $\DM$ is in a beneficial direction for the platform, e.g., positive for metrics like revenue, negative for metrics like cost.  It can be shown that our results can extend to incorporate this additional sign constraint as well.\footnote{In particular, the decision rule that first tests whether $|\hat{T}|$ is sufficiently large, then checks that the sign of $\hat{T}$ is positive or negative, is equivalent to a {\em one-sided} hypothesis testing procedure; we briefly discuss this in Section \ref{sec:cc}.  For simplicity of presentation in this paper, we consider the {\em two-sided} hypothesis testing procedure as described above.})

Although the frequentist hypothesis testing paradigm faces many criticisms in practice, it is also widely deployed and prevalent in the decision-making practices of technology companies generally, and online platforms and marketplaces in particular.  Typically this decision-making pipeline is evaluated based on two criteria.  {\em First} is the {\em false positive probability} (or type I error rate): what is the chance of mistakenly rejecting $H_0$, when $H_0$ is true, i.e., $\GTE = 0$, and in particular, does the false positive probability match the desired control in the decision rule?  (For example, when a cutoff of $|\hat{T}| > 1.96$ is used, the false positive probability should be no more than $5\%$.)  {\em Second} is the {\em statistical power}, or the complement of the {\em false negative probability} (type II error rate): what is the chance of correctly rejecting $H_0$, under a specific alternative for which $\GTE \neq 0$?  We analyze and characterize these two quantities in the presence of interference.

Formally, to carry out our analysis, we consider a general continuous-time Markovian model of an inventory-constrained platform with finitely many listings; similar models have been considered by prior papers as well to model interference (see Section \ref{sec:related}).  Customers arrive over time according to a Poisson process, and can book a listing if one is available; our model allows customers to have heterogeneous preferences, but listings are homogeneous. (We briefly discuss listing heterogeneity in Section \ref{sec:conclusion}.)  Once booked, a listing remains occupied for some time before becoming available to book again.  Our model assumes that an arriving customer's {\em booking probability} is lower when fewer listings are available, as is the case in real-world platforms.  We consider a natural class of treatments: those that change the state-dependent booking probabilities of customers.  We suppose that the platform runs a Bernoulli customer-randomized experiment, collects data on the booking outcomes of $N$ arriving customers, and executes the decision-making pipeline above.

A challenge arises here because in general, the null hypothesis $H_0$ that $\GTE = 0$ does {\em not} completely determine the distribution of observations: informally, this is because there are many configurations of treatment and control booking probabilities that lead to the {\em same} value of $\GTE$. For example, an intervention could increase booking probabilities in some states while decreasing booking probabilities in other states, resulting in $\GTE=0$; and there are many such interventions. (In statistical terms, $H_0: \GTE = 0$ is a {\em composite} null hypothesis, rather than a {\em simple} null hypothesis.) To make progress, we focus our attention on sign-consistent treatments: these are treatments where in all states, the average booking probability of customers moves in the same direction (i.e., higher or lower) relative to control. 
Informally, sign-consistency is a reasonable assumption for  those treatments that are not state-dependent; for example, a platform might provide more information about cancellation policies or payment procedures, or it might enable a new checkout flow.  The platform may be uncertain whether users find the additional friction beneficial or detrimental; but because the intervention is not state-dependent, it may be plausible to assume in advance that the sign of the average treatment effect across customers is not state-dependent either. Sign-consistency can also be a reasonable assumption for some state-dependent settings, where one can reasonably expect that the change in booking probability will be the same in all states (e.g., if a platform raises fees, then it is reasonable to assume that regardless of the state, this will lower booking probabilities).
Prior literature has also studied the bias in estimation of $\GTE$ from platform experiments when treatments are sign-consistent; see, e.g., \cite{holtz2020reducing, johari2022experimental, li2022interference, bright2022reducing, dhaouadi2023price}.

Crucially, we show that if the treatment is known to be sign-consistent {\em and} $H_0$ holds, then all booking probabilities in treatment and control are {\em identical}. In other words, the assumption of sign-consistent treatments, together with $\GTE = 0$, determines that the experimental observations correspond to an {\em A/A experiment}.  A/A experiments are tests where the same version of a product or feature is tested against itself (so that the joint distribution of observations in both the global treatment and global control conditions is the same), and are routinely used in practice to validate testing infrastructure \cite{geteppoWhatTest,optimizelyTestWhat,statsigConductTest}).

Our {\em first} main contribution is to show that {\em when treatments are sign-consistent, the false positive probability of the decision-making pipeline above is correctly controlled} (Section \ref{sec:fpp}). We obtain this insight by considering a more general setting of A/A experiments with arbitrary (even non-Markovian) dependence structure between observations.  We show via a probabilistic exchangeability argument in this more general setting that the estimator $\widehat{\Var}$ is {\em unbiased} for the true variance of $\GTE$, despite no estimation of covariance between observations.  %
An associated central limit theorem is also given in a general setting of Markovian system dynamics (see Theorem \ref{thm:aa_clt}).  Application of these results to our inventory-constrained platform yields the desired control of false positive probability under sign-consistent treatments.

Our {\em second} main contribution is to show that {\em when treatments are sign-consistent, the statistical power of the decision-making pipeline above is {\em significantly} higher than that achieved by any similar pipeline using a debiased estimator}, in the realistic setting of large state spaces (Section \ref{sec:power}).  We show this by imagining that the platform has access to an alternative estimator $\hat{\theta}$ which is unbiased for the $\GTE$, and alongside $\hat{\theta}$ is able to exactly compute the variance $\Var(\hat{\theta})$.  We suppose the platform could form a test statistic $\hat{U} = \hat{\theta}/\sqrt{\Var(\hat{\theta})}$.  Our key finding is that when the state space is large, the magnitude of $\hat{T}$ is exponentially larger than that of $\hat{U}$.  This requires two steps.  First, using stochastic monotonicity arguments applied to the underlying experiment Markov chain, we show that both in finite samples and asymptotically, $\DM$ is {\em larger} than $\GTE$ (Theorems \ref{thm: pos bias DM} and \ref{thm: pos bias limit}); in finite samples in particular, this is a result that has not been shown previously.  On the other hand, leveraging a Cram\'{e}r-Rao lower bound on variance of any unbiased estimator given in \cite{farias2022markovian}, we show that $\Var({\hat{\theta}})$ generically grows {\em exponentially} with the size of the state space in any capacity-constrained platform experiment of the type we study in this paper, while $\widehat{\Var}$ remains bounded.  (In \cite{farias2022markovian}, a specific example with this property is constructed; our result shows the exponential growth is generic and unavoidable.)  Together, these results suggest that using $\hat{T}$ yields significantly {\em higher} power than using the unbiased estimation strategy $\hat{U}$ as the state space grows, as we demonstrate via numerical examples.  

These two contributions together suggest the surprising finding that when treatments are sign-consistent, and the platform uses the decision-making pipeline above, then the platform is likely strictly better off {\em not} debiasing.  Our {\em third} main contribution is to investigate the robustness of this finding, by studying via numerics the consequences when treatments are sign-inconsistent (Section \ref{sec:sims}).  We consider a natural class of treatments: those that might {\em increase} booking probabilities when many listings are available, but {\em lower} booking probabilities when few listings are available.  For example, a ride-sharing platform may want to test changing prices in a state-dependent manner, lowering prices relative to control when many drivers are available, but raising prices relative to control when few drivers are available.  We show via numeric example that with sign-inconsistent treatments, in principle, the performance of the na\"{i}ve decision-making pipeline above using $\hat{T}$ can be arbitrarily worse than a debiased strategy using $\hat{U}$, both in terms of false positive probability and in terms of statistical power.

Taken together, our findings have important implications for the deployment of debiasing strategies in practice.  Many of the debiasing methods suggested in the literature are nontrivial, and from a practical standpoint, there can be significant organizational friction in adopting these alternatives.  Our work suggests that understanding the nature of the intervention is important to determining whether the additional effort in debiasing is worthwhile; and indeed, for sign-consistent interventions and customer-randomized experiments, it may be strictly preferable {\em not} to debias.  Of course, platforms have many other goals as well in A/B experimentation; for example, often the precise estimate of the treatment effect is of interest (e.g., when evaluating the benefits of an intervention against the cost of deployment), in which case debiasing is essential to obtain an accurate estimate of the true $\GTE$.  Broadly, our work argues that platforms should carefully define their objectives and inferential goals in determining an appropriate approach to estimation and decision-making.

The remainder of the paper is organized as follows.  In Section \ref{sec:related}, we present related work.  In Section \ref{sec:cc}, we introduce our benchmark inventory-constrained platform model, as well as the decision-making pipeline outlined above, and introduce the concept of sign-consistent interventions.  In Section \ref{sec:fpp}, we present our general results for A/A experiments (Section \ref{subsec:aa_general}), and use it to characterize false positive probability under sign-consistent treatments, both in finite samples and asymptotically (Section \ref{sec:aa_monotone}).  In Section \ref{sec:power}, we study statistical power, again in finite samples (Section \ref{sec: finite sample power}) and asymptotically (Section \ref{sec: asymptotic power}).  Finally, in Section \ref{sec:sims}, we present numerics results investigating false positive probability and statistical power when treatments are sign-inconsistent.  We conclude in Section \ref{sec:conclusion}.

We collect together supplementary material in several appendices.  Appendix \ref{app:MC} may be of independent interest to other researchers working on experimentation in Markovian settings, where we present together (in a self-contained manner) central limit theorems and associated analysis for both A/B and A/A experiments when the underlying treatment and control systems are general Markov chains.  (We use these results in our analysis of our inventory-constrained platform model.)  Appendix \ref{app:proofs} contains proofs of several results in the main paper.  Appendix \ref{app:sims} contains additional numerics results.

%% file: Sections/related.tex
\section{Related work}
\label{sec:related}

In this section we discuss three related streams of work: (1) Interference in experiments, particularly in networks and markets; (2) the use of Markov chain models to study experimental design and estimation; and (3) the practice of A/B experimentation, and particularly making decisions from A/B experiments.

\paragraph{Interference in experiments.}  A rich literature in causal inference broadly, and more recently in the study of networks  and markets, has considered interference between treatment and control groups in experiments.  We refer the reader to \cite{Hudgens2008,rosenbaum2007interference,imbens2015causal,hu2022average} for broader discussion of interference.  In the literature on social networks, a range of papers have studied the role of interference on bias of estimation, as well as approaches to obtain unbiased estimates of direct and indirect treatment effects; see, e.g., \cite{aronow2017estimating,Manski2013,ugander2013graph,eckles2017design} for early influential work in this area.  We note that \cite{athey2018exact} provides exact p-values in a randomization inference framework for network experiments, which correctly control false positive probability.

More recently, extensive attention has also been devoted to interference in marketplace and platform experiments; see, e.g., \cite{blake2014marketplace,wager2021experimenting,johari2022experimental,li2022interference,bajari2023experimental,munro2021treatment,bright2022reducing,shirani2024causal,han2023detecting} as examples of this line of work.  As discussed in Section \ref{sec:intro}, this prior work primarily focuses on the presence of {\em bias} in the use of ``naive'' estimators of $\GTE$, such as the difference-in-means estimator, and often investigates designs and/or estimators to reduce that estimation bias.  Many of the papers that study Markov chain models in the context of experimentation are also specifically motivated by similar questions in marketplace experimentation, as we discuss below.  In contrast to these works, our emphasis is on understanding false positive probability and statistical power when the difference-in-means estimator (and its associated naive variance estimator) are used for decision-making, in spite of interference and the resulting bias.

We note that a number of papers have specifically considered interference in the context of {\em price} experimentation, including \cite{blake2014marketplace,wager2021experimenting,li2023experimenting,dhaouadi2023price,roemheld2024interference}.  One interesting issue that can arise there is that optimization based on the estimated effect of a price change on revenue or profit can lead the decision-maker astray, as (due to interference) the estimator can have the wrong sign compared to the true $\GTE$ (see \cite{dhaouadi2023price,roemheld2024interference}).  We consider bookings as our primary outcome metric in our paper, so our findings should be viewed as distinct from and complementary to this line of work.

\paragraph{Markov chain models of experiments.}  A number of recent papers have employed Markov chain models as a structural representation of treatment and control in experimental settings; see, e.g., \cite{glynn2020adaptive,hu2022switchback,johari2022experimental,farias2022markovian,farias2023correcting,li2023experimenting,boutilier2024randomized,hays2025double}, as well as related work on off-policy evaluation for Markov decision processes \cite{liao2021off,kallus2020double,hu2023off,mehrabi2024off}. 

Of these papers, the most related to our own work are \cite{johari2022experimental,farias2022markovian,li2023experimenting}, all of which consider a queueing model to capture interference due to limited inventory, as arises in marketplaces.  In \cite{johari2022experimental}, the paper focuses on analysis of estimation bias in a mean-field regime, as well as an alternative two-sided randomized design and associated estimator to mitigate bias; a similar design was contemporaneously proposed and studied in \cite{bajari2023experimental}.  In \cite{li2023experimenting}, the paper exploits knowledge of the queueing model to construct low-variance estimators in the presence of congestion, in the context of pricing experiments that impact arrival rates.  In \cite{farias2022markovian}, the paper studies an estimator based on an (estimated) difference in {\em $Q$-values}, to reduce bias in Markov chain experiments in general, and marketplace experiments in particular.  We leverage a Cram\'{e}r-Rao lower bound from \cite{farias2022markovian} in our analysis of asymptotic statistical power in Section \ref{sec: asymptotic power}.

These papers on Markov chain models for experiments do not consider the impact of interference on the resulting decisions that are made based on the experiment.   One recent exception is \cite{boutilier2024randomized}; they study a setting with interference arising due to capacity constrained interventions, and (using a queueing-theoretic approach) show that the statistical power to detect a positive effect is not monotonic in the number of subjects recruited, as is the case in standard experiments without interference.  We also note that more recently (subsequent to our work in this paper), \cite{ni2025decision} studies optimal decision-making after switchback experiments.

\paragraph{The practice of A/B experimentation.}  Finally, our paper is connected to a wide literature on the use of A/B experiments in industry to make decisions about features, algorithms, and products.  For a broad overview of the relevant considerations in this space, we refer the reader to \cite{kohavi2020trustworthy}.  In general, this literature considers experiments where data are independent observations, and so does not consider the role of interference.

A number of different lines of work consider potential challenges in the decision-making pipeline that follows A/B experiments.  For example, a range of papers (see, for example, \cite{Johari2017,johari2022always,howard2021time}) study methods to ameliorate the inflation of false positive probability that arises when decision-makers continuously monitor experiments.  
A recent paper notes that because most experiments do not succeed in practice, false positives may be more common than practitioners realize \cite{kohavi2024false}.

Several papers highlight the fact that if companies evaluate experiments based on the returns generated by the ``winning" variations, and if there are opportunity costs to experimentation, then running many, shorter experiments is ideal to find the potentially big winners (see, e.g., \cite{feit2019test,schmit2019optimal,azevedo2020b}).  More generally, we note that if one is interested in simply picking the best alternative among many possible features, products, or algorithm designs, then the extensive literature on {\em multi-armed bandit} algorithms (see, e.g., \cite{lattimore2020bandit} for a recent textbook treatment) provides appropriate methodology.  In our paper, by contrast, the emphasis is specifically on the impact of interference for decisions made from A/B experiments that follow the paradigm of randomized controlled trials.

%% file: Sections/cc.tex
\section{Experiments and decision-making in inventory-constrained platforms}
\label{sec:cc}

In this section, we introduce a benchmark model for interference, experimentation, estimation, and decision-making in an inventory-constrained platform.  Using this model, we articulate the key questions studied in this paper: If a platform follows common practice in running experiments by using a hypothesis test for the null hypothesis that the treatment effect is zero, what is the probability of false positives (type I error rate) and false negatives (type II error rate, or the complement of statistical power)? 

In Section \ref{sec:cc_model}, we present the details of the stochastic model we study.  In Section \ref{sec:cc_expts}, we describe experiments and, estimation.  In Section \ref{sec:cc_clt}, we provide asymptotic descriptions of the estimators: a central limit theorem for the difference-in-means estimator, and an associated limit for the variance estimator).  In Section \ref{sec:cc_ht} we formalize the null hypothesis that the global treatment effect is zero, and describe the standard t-test statistic for this hypothesis.
Finally, in Section \ref{sec:cc_error}, we also specialize our setting to sign-consistent interventions, and define the false positive probability and statistical power.

\subsection{A model of an inventory-constrained platform}
\label{sec:cc_model}

We consider a platform where {\em customers} (the demand side) arrive over time, and can choose to book from a finite supply of {\em listings}; for example, such a model is a reasonable abstraction of a two-sided marketplace.  At a high level, we model such a platform as a Markovian birth-death queueing system.  In our model customers arrive over time, and can choose to book a listing if one is available when they arrive.  If they book, then the listing is made unavailable for a period of time, before being made available again for booking.  Similar models have previously been considered in the context of experiments (see, e.g., \cite{johari2022experimental, farias2022markovian, li2023experimenting}).

The formal details of our model are as follows.  Throughout we use \textbf{boldface} to denote vectors and matrices.

\paragraph{Time.}  The system evolves over an infinite continuous time horizon $t \in [0, \infty)$. 

\paragraph{Listings.} The system consists of $K$ homogeneous listings.

\paragraph{State description.}  At each time $t$, each listing can be either {\em available} or {\em booked} (i.e., occupied by a customer who previously booked it).  We let $X_t \in \{0, \dots, K\}$ denote the number of booked listings at time $t$. 

\paragraph{Customers.} Customers arrive sequentially to the platform and can book a listing if one is available, i.e., if $X_t < K$ for a customer who arrives at time $t$.  We give arriving customers sequential indices $1, 2, \ldots$.  Each customer $i$ has a type $\gamma_i \in \Gamma$, where $\Gamma$ is a finite nonempty set that represents customer heterogeneity. Customers of type $\gamma$ arrive according to a Poisson process with rate $\lambda_\gamma > 0$.  

We assume that when a customer of type $\gamma$ arrives to the platform with $k$ booked listings, the customer will book with probability $p_\gamma(k)$.  We let $Y_i \in \{0,1\}$ denote the booking outcome of customer $i$ (where $1$ denotes that they book a listing, and $0$ denotes that they do not make a booking). 

Throughout the paper, we make the natural assumption that the booking probabilities are strictly decreasing in the number of booked listings, as formalized below; notably, this assumption is satisfied by related models studied in prior literature, e.g., \cite{johari2022experimental,farias2022markovian,bright2022reducing,li2023experimenting}.  For completeness we also define $p_\gamma(K) = 0$ for all $\gamma$: no listings can be booked if none are available.  (We assume the inequalities are strict in the following definition largely for technical simplicity of the remainder of our presentation.) 

\begin{assumption}
\label{as:booking_prob_monotone}
For any $\gamma \in \Gamma$, there holds $p_\gamma(0) > p_\gamma(1) > \cdots > p_\gamma(K-1) > p_\gamma(K) = 0$.
\end{assumption}

In our subsequent development it will be convenient to abstract away from customer heterogeneity, by defining $\lambda = \sum_\gamma \lambda_\gamma$, and:
\[ p(k) = \frac{1}{\lambda} \sum_\gamma \lambda_\gamma p_\gamma(k). \]
Note that $p(k)$ is the probability a listing is booked conditional on %
a customer arrival, if $k$ listings are booked.  Further, note that as long as each $\v{p}_\gamma$ satisfies Assumption \ref{as:booking_prob_monotone}, then $\v{p}$ satisfies this assumption as well.  

\paragraph{Listing holding times.}  We assume that when $X_t = k$, i.e., $k$ listings are booked, the time until at least one of those listings becomes available is exponential with parameter $\tau(k)$.  

We make the following natural monotonicity assumption on the holding time parameters $\tau(k)$.  Qualitatively, this assumption guarantees that as the available inventory of listings becomes more scarce, the rate at which listings become available again can only increase.  Note that this allows for a wide range of specifications.  For example, if $\tau(k) = k \tau$ for a fixed constant $\tau$, then booked listings have independent exponential holding times with mean $1/\tau$.  
On the other hand, if $\tau(k) = \tau$ for all $k$, then the service system operates as a single server queueing system with finite buffer capacity $K$.

\begin{assumption}
\label{as:tau_monotone}
For $k = 1, \ldots, K-1$, there holds $0 < \tau(k) \leq \tau(k+1)$.
\end{assumption}

\paragraph{Steady state.} With the preceding assumptions, the state $X_t$ is a continuous time Markov chain, with generator $\v{Q}$ defined as follows:
\begin{equation}
\label{eq:generator}
\v{Q} = \begin{bmatrix}
  -\lambda p(0)   & \lambda p(0)  & 0 & \dots & 0\\
  \tau(1) & -\tau(1)- \lambda p(1)  & \lambda p(1)  & \dots & 0\\
  \vdots & \vdots & \vdots& \vdots & \vdots \\
  0 & \dots & 0 & \tau(K) & -\tau(K)
\end{bmatrix}.
\end{equation}

Under our assumptions, we note that the Markov chain above is a birth-death chain that is irreducible on a finite state space, with a unique steady state distribution $\v{\pi} = (\pi(0), \ldots, \pi(K))$ defined by $\v{\pi} \v{Q} = 0$ and $\sum_{k=1}^K \pi(k) = 1$.  Note that $\v{\pi}$ depends on the parameters of the system through the generator $\v{Q}$, in particular, the aggregate arrival rate $\lambda$, the average booking probabilities $p(k)$ for each $k$, and the holding time parameters $\tau(k)$ for each $k$.

\paragraph{Steady state average booking probability.}  A key statistic of interest to us is the {\em steady state average booking probability} of the service system, denoted $\rho$.  This is the probability an arriving customer books a listing in steady state, and because Poisson arrivals see time averages (PASTA), $\rho$ can be obtained as: 
\begin{equation}
\label{eq:ssabp}
\rho = \sum_{k=1}^K \pi(k) p(k).
\end{equation}

\subsection{Experiments and estimation}
\label{sec:cc_expts}

In this section, we model treatments that change the system dynamics by altering customer booking behavior.  We also define a canonical experiment design and difference-in-means estimator for the resulting treatment effect, as well as an associated variance estimator.  The formal details of our approach are as follows; similar models have been considered by prior work as well (see, e.g., \cite{johari2022experimental, farias2022markovian, li2023experimenting}).

\paragraph{Binary treatment.}  We consider interventions that change the booking probability of a customer of type $\gamma$ in state $k$. %
Formally, we denote treatment by $1$ and control by $0$, and consider an {\em expanded} type space for customers: for each type $\gamma \in \Gamma$, we let $(\gamma, 1)$ denote a treatment customer of type $\gamma$, and $(\gamma, 0)$ denote a control customer of type $\gamma$.  For treatment status $z \in (0,1)$, we let $p_{\gamma,z}(k)$ denote the probability a type $(\gamma, z)$ customer books in state $k$.  We again emphasize that we assume these booking probabilities satisfy Assumption \ref{as:booking_prob_monotone}.

As before, we average over types $\gamma$ and obtain:
\[ p_z(k) = \frac{1}{\lambda} \sum_\gamma \lambda_\gamma p_{\gamma,z}(k). \]
We refer to $p_1(k)$ (resp., $p_0(k)$) as the treatment (resp., control) {\em booking probability} in state $k$.  %

\paragraph{Bernoulli customer randomization.} We assume the platform runs an experiment on the first $N$ customers to arrive; these are also commonly referred to as {\em user-level} randomized experiments in industry.  We assume a parameter $a \in [0,1]$ such that each arriving customer is randomized independently to treatment with probability $a$, i.e., for each customer $i$, their treatment status $Z_i$ is an independent Bernoulli$(a)$ random variable. For notational convenience, we define $N_1 = \sum_i Z_i$, $N_0 = \sum_i (1 - Z_i)$.  

\paragraph{Experiment system dynamics.}  In an experiment, the system dynamics again evolve as a Markovian birth-death queueing system as before, but with arrival rates that are mediated by assignment to treatment or control.  Formally, define:
\begin{align}
q_a(k) &= (1-a) p_0(k) + a p_1(k); \label{eq:expt_booking}
\end{align}
this is the probability that an arriving customer books in an experiment with treatment allocation $a$, when the state is $k$.  Then the generator for the system dynamics in an experiment is again given by \eqref{eq:generator}, but with $\lambda p(k)$ replaced by $\lambda q_a(k)$.

We let $\v{\pi}_a = (\pi_a(0), \ldots, \pi_a(K))$ denote the steady state distribution of this Markov chain.  (Note, though, that the experiment will only last for a random finite time, since it involves a sample size of $N$ customers.)  

\paragraph{Estimand: The global treatment effect.}  With these definitions, the {\em global treatment} condition is the special case where $a = 1$, and the {\em global control} condition is the special case where $a = 0$.  We can then define steady state average booking probabilities: for $z = 0,1$, we have $\rho_z = \sum_k \pi_z(k) p_z(k)$. %

We define the estimand as the steady state {\em global treatment effect} (or global average treatment effect):
\begin{equation}
\label{eq:gte_cc}
 \GTE= \rho_1 - \rho_0. 
\end{equation}

\paragraph{Difference-in-means (DM) estimator.}  A common practice is to employ a {\em difference-in-means} (DM) estimator to formulate a test statistic.  We denote this estimator by $\DM_N$.  Formally:
\begin{equation}
\label{eq:DM}
\DM_N = \b{Y}(1) - \b{Y}(0), 
\end{equation}
where:
\[ \b{Y}(1) = \frac{\sum_i Z_i Y_i}{N_1}; \quad \b{Y}(0) =  \frac{\sum_i (1 - Z_i) Y_i}{N_0}. \]
(Recall $N_1$ and $N_0$ are the number of treated and control units, respectively.)  All summations are over all $N$ units.  

Note that $\DM_N$ is not well defined if either $N_1 = 0 $ or $N_0 = 0$; we will use conditioning to avoid this event in our analysis.  

\paragraph{The naive variance estimator and bias.}  Associated to the DM estimator is a ``naive'' variance estimator, defined as follows:
\begin{equation}
\label{eq:var_hat}
 \widehat{\Var}_N = \frac{1}{N_1(N_1 - 1)} \sum_i Z_i( Y_i  - \b{Y}(1))^2 + \frac{1}{N_0(N_0-1)} \sum_i (1 - Z_i) (Y_i - \b{Y}(0))^2.
\end{equation}
Note that this estimator assumes that observations are independent draws across treatment and control groups, and i.i.d.~draws within treatment and control groups; if these assumptions held, the variance estimator would be unbiased for the true variance of the DM estimator.  

Again, note that $\widehat{\Var}_N$ is not well defined if either $N_1 \leq 1$ or $N_0 \leq 1$; again, we will use conditioning to avoid this event in our analysis.  

\paragraph{Interference and bias.}  We note that the Bernoulli customer-randomized experiment together with the DM estimator and the naive variance estimator suffers from {\em interference}, i.e., violation of the {\em stable unit treatment value} assumption (SUTVA) \cite{imbens2015causal}.  This assumption requires that the outcome of one unit (in this case, the booking outcome of a customer) does not depend on the treatment assignment of other units (i.e., other customers).  However, because inventory is constrained in this service system, there is interference over time between customers: if a customer books a listing, then that listing may be unavailable for a subsequent customer.

In general, this interference effect will imply that both the estimator $\DM_N$ will be biased (as has been extensively discussed in the literature, cf.~Section \ref{sec:related}), and the variance estimator $\widehat{\Var}_N$ will be biased as well.

\subsection{A central limit theorem}
\label{sec:cc_clt}

In this section we state a central limit theorem (CLT) for $\DM_N$ and also provide an asymptotic characterization of $\widehat{\Var}_N$.  We prove this result in a more general setting, considering A/B experiments between two arbitrary Markov chains on a finite state space; see Appendix \ref{app:MC_cc} for details.  

We require the following definitions.   Fix a treatment allocation parameter $a$, and let $\v{\pi}_a$ be the steady state distribution as above.  We define the {\em average direct effect} ($\ADE_a$) as follows \cite{hu2022average}:
\begin{equation}
 \ADE_a = \sum_k \pi_a(k) p_1(k) - \sum_k \pi_a(k) p_0(k). \label{eq:cc_ade}
 \end{equation}
This is the difference in treatment and control booking probabilities, but when the distribution over states is given by the steady state distribution from the experiment.  Since, in general, the steady state distribution $\v{\pi}_a$ neither matches global treatment $\v{\pi}_1$ nor global control $\v{\pi}_0$, in general $\ADE_a \neq \GTE$.

Next, define the following quantities for $z,z' \in \{0,1\}$:
\begin{align}
V(z) &= \E_{\v{\pi}_a} \left[ \left( Y_1 - \sum_k \pi_a(k) p_z(k) \right)^2 \Bigg| Z_1 = z \right] \notag \\
&= \sum_k \pi_a(k) \left( p_z(k) \left( 1 - \sum_k \pi_a(k) p_z(k) \right)^2 + (1 - p_z(k)) \left( \sum_k \pi_a(k) p_z(k) \right)^2 \right); \notag \\
&= \Var_{\v{\pi}_a}(Y_1 | Z_1 = z); \label{eq:V_cc}\\
C_j(z,z') &= \E_{\v{\pi}_a} \left[ \left(Y_1 - \sum_k \pi_a(k) p_z(k)\right) \left(Y_j - \sum_k \pi_a(k) p_{z'}(k)\right) \Bigg| Z_1 = z, Z_j = z'\right]. \notag \\
&= \Cov_{\v{\pi}_a}(Y_1, Y_j | Z_1 = z, Z_j = z'). \label{eq:Cj_cc}
\end{align}
In these expressions, the subscript $\v{\pi}_a$ on the expectations indicates that the chain is initialized in the steady state distribution $\v{\pi}_a$ just prior to the arrival of the first customer.  These quantities are the conditional variance and covariance of rewards, respectively, given the treatment assignments.

The following central limit theorem comes from an application of Theorem \ref{thm:ab_clt} in Appendix \ref{app:MC} to this setting; see the discussion in Appendix \ref{app:MC_cc} for details.

\begin{theorem}
\label{thm:ab_clt_cc}
Suppose $0 < a < 1$.  Then regardless of the initial distribution, $\DM_N \to^p \ADE_a$ as $N \to \infty$, and $\DM_N$ obeys the following central limit theorem as $N \to \infty$:  
\begin{equation}
\label{eq:dm_clt_cc}
\sqrt{N} ( \DM_N - \ADE_a )  \Rightarrow \mc{N}\left(0, \tilde{\sigma}_a^2\right),
\end{equation}
where:
\begin{equation}
\tilde{\sigma}_a^2 = \left(\frac{1}{1-a}\right) V(0) + \left( \frac{1}{a} \right) V(1) + 2 \sum_{j > 1} C_j(0,0) + C_j(1,1) - C_j(0,1) - C_j(1,0), \label{eq:tilde_sigma_cc}
\end{equation} 
with $\tilde{\sigma}_a^2 > 0$.
\end{theorem}

Note that since the mean of $\DM_N$ tends to $\ADE_a$ as $N \to \infty$, in general $\DM_N$ is biased as an estimator of $\GTE$.  Further, 
note that the variance expression in the limit has an intuitive explanation: the true variance of the DM estimator includes positive contributions from covariances {\em within} the same treatment groups, and negative contributions from covariances {\em across} treatment groups.  Also note that if $C_j(0,0) = C_j(1,1) = C_j(0,1) = C_j(1,0)$ for each $j$, then the covariance terms will cancel; this observation plays an important role in our analysis of false positive probability, cf.~Section \ref{sec:aa_monotone}.

It is straightforward using the ergodic theorem for Markov chains to study the variance estimator $\widehat{\Var}_N$ in the limit as $N \to \infty$.  In particular, we can establish the following result as an instantiation of Theorem \ref{thm:var_est} in Appendix \ref{app:MC}; see the discussion in Appendix \ref{app:MC_cc}.

\begin{theorem}
\label{thm:var_est_cc}
Suppose $0 < a < 1$.  Then $\widehat{\Var}_N$ satisfies
\begin{equation}
N \widehat{\Var}_N \to^p \left(\frac{1}{a}\right)V(1) + \left(\frac{1}{1-a}\right) V(0).
\end{equation}
\end{theorem}

Note that in comparison to the true variance $\tilde{\sigma}_a^2$, the estimator $\widehat{\Var}_N$ misses all the covariance terms.  Whether this is an overestimate or underestimate of the true variance depends on whether the within-group covariances are stronger are weaker than the across-group covariances, cf.~\eqref{eq:tilde_sigma_cc}.

\subsection{Hypothesis testing}
\label{sec:cc_ht}

In this section, we present the canonical approach to frequentist hypothesis testing and decision making using the estimators $\DM_N$ and $\widehat{\Var}_N$.
Of particular interest for the platform is the {\em null hypothesis} $H_0$ that the GTE is zero:
\begin{equation}
H_0 : \GTE = 0. \label{eq:null}
\end{equation}
Informally, a typical approach to making decisions involves computation of a t-test statistic using $\DM_N$ and $\widehat{\Var}_N$, and rejecting the null hypothesis if this statistic exceeds a threshold.  In this section we formalize this process.

\paragraph{Test statistic.}  Given the DM estimator and associated variance estimator, the platform forms the following test statistic:
\begin{equation}
\label{eq:tstat}
 \hat{T}_N = \frac{\DM_N}{\sqrt{\widehat{\Var}_N}}. 
\end{equation}
If observations were normally distributed and independent, and identically distributed within groups, then this would be the standard $t$-test statistic for testing the null hypothesis of zero treatment effect.  Of course these assumptions do not hold in our setting in general.

\paragraph{Decision rule.}  Nevertheless, common practice involves comparing the test statistic $\hat{T}_N$ to Student's $t$ distribution to determine whether sufficient evidence exists to reject the null hypothesis $H_0$.  As we will be primarily interested in large sample behavior as $N$ grows, we compare the test statistic to a reference standard normal distribution to determine whether to reject $H_0$ (i.e., we consider an asymptotic $z$-test).  Formally, given a desired false positive probability $\alpha$, we assume that the platform uses the following decision rule:
\begin{equation}
\label{eq:decision}
\text{Reject } H_0 \quad \Leftrightarrow \quad |\hat{T}_N| > \Phi_{\alpha/2},
\end{equation}
where $\Phi_q$ is the upper quantile of a standard normal distribution, i.e., the unique value such that $P(W \geq \Phi_q) = q$ for a standard normal random variable $W$.

The decision rule we have described is a {\em two-sided} hypothesis testing procedure: the null hypothesis is $H_0 = 0$, and thus the test involves comparing only the {\em magnitude} of $\hat{T}_N$ to a cutoff.  In practice, it is often the case that decision-makers are interested in testing not only the magnitude of $\hat{T}_N$, but also the sign; in particular, a common policy when deciding whether to launch a feature is to {\em first} check whether $|\hat{T}_N|$ is larger than a cutoff (e.g., $1.96$ for a 95\% statistical significance level), and {\em second}, only launch if in addition the sign of $\hat{T}_N$ (or equivalently, $\widehat{\GTE}_N$) is in a beneficial direction for the platform (e.g., positive for metrics like bookings).  Such a decision rule is equivalent to a {\em one-sided} hypothesis test, of the composite null $\tilde{H}_0: \GTE \leq 0$.  Analysis of this procedure adds technical complexity, since in principle the distribution of $\hat{T}_N$ can be asymmetric.  Nevertheless, we expect that all our results extend to this setting; however, for simplicity of technical presentation, we adopt the two-sided hypothesis test decision rule described above.

\subsection{Sign-consistent interventions and error rates}
\label{sec:cc_error}

In this section, we formalize the false positive probability and statistical power of the decision rule in \eqref{eq:decision}.  Again, if observations were normally distributed and independent, and identically distributed within groups, then the decision rule \eqref{eq:decision} has the property that (asymptotically, as the number of observations $N \to \infty$ the false positive probability converges to $\alpha$.  In other words, the probability that the null hypothesis is rejected when $H_0$ is true approaches $\alpha$.  We have no such guarantee {\em a priori} in our setting due to the interdependence of observations within and across groups. 

Before proceeding with a formal specification, a challenge arises in the inventory-constrained platform model, because the null hypothesis $H_0$ is {\em underspecified}: there are {\em many} specifications of treatment and control that lead to $\GTE =0$.  (In statistical terms, $H_0: \GTE = 0$ is a {\em composite} null hypothesis, rather than a {\em simple} null hypothesis.)  Formally, this is simply because $\GTE = 0$ is only a one-dimensional constraint, while the number of parameters specifying the treatment and control systems is much higher dimensional.  Without further assumptions, only assuming $H_0$ holds is insufficient to specify the data distribution, and thus prevents us from specifying the false positive probability.  

To make progress, we restrict attention to {\em sign-consistent interventions}; informally, these are interventions where the booking probability either rises in every state, or falls in every state.  We have the following definition.\footnote{We note that the definition here is a type of {\em convex cone} constraint on the booking probabilities; see \cite{wei2019geometry} for more on the structure of hypothesis testing under such constraints.}
\begin{definition}[Sign-consistent interventions]
\label{def:monotone}
Given treatment booking probabilities $\v{p}_1$ and control booking probabilities $\v{p}_0$, we say the treatment is a {\em positive} (resp., {\em negative}) intervention if $p_1(k) \geq p_0(k)$ (resp., $p_1(k) \leq p_0(k)$) for all states $k$.  A \em{sign-consistent treatment} is one that is either negative or positive (or both, in case treatment is identical to control).

We say the treatment is {\em strictly positive} (resp., {\em strictly negative}) if the corresponding inequality in the definition is strict for at least one state $k$.  A \em{strictly sign-consistent} treatment is one that is either strictly positive or strictly negative.
\end{definition}

Informally, sign-consistency is a reasonable assumption for those treatments that are not state-dependent.  Inventory-constrained platforms test many features that fall in this category.  To take just a few representative examples, a platform may choose to provide more information about some aspect of their service (e.g., cancellation policies, insurance, payment procedures, etc.); it might enable a new checkout flow (e.g., adding a mobile payment service as an option); or it might manipulate the user interface to add additional calls-to-action (e.g., buttons to sign up for notifications).  In all of these cases, the platform may be uncertain whether users find the additional friction beneficial or detrimental; but because the intervention is not state-dependent, it may be plausible to assume in advance that the sign of $p_1(k) - p_0(k)$ is not state-dependent either.  Monotonicity can also be a reasonable assumption for some state-dependent settings, where one can reasonably expect that the direction of the change in booking probability will be the same in all states.  For example, if a platform raises fees, then it is reasonable to assume that regardless of the state, this should lower booking probabilities in each state.  (Of course, in the latter case, the platform would expect the treatment can only yield $\GTE \leq 0$; and so an experiment would primarily be relevant to judge whether the estimated effect is statistically significant.)  We also note that monotonicity is an assumption on the {\em average} booking probabilities $\v{p}_0, \v{p}_1$ across customer types $\gamma$; in particular, there may be customer types whose sign does not match the sign of $p_1(k) - p_0(k)$ in all states $k$, and yet the overall intervention can still be sign-consistent, cf.~Definition \ref{def:monotone}.  A number of previous papers have considered and motivated sign-consistent interventions in platform experiments as well; see, e.g., \cite{holtz2020reducing, johari2022experimental, li2022interference, bright2022reducing, dhaouadi2023price}.

The following proposition shows that if the treatment is sign-consistent, and $H_0$ is true, then treatment must be identical to control.  See Appendix \ref{app:proofs} for the proof.  The proof uses an application of stochastic comparison of the equilibrium distributions of the treatment and control chains; we develop the necessary results for this approach later in our analysis of the finite sample bias of $\DM$, in Section \ref{sec: finite sample power}.  Informally, we show that if the treatment is strictly positive (resp., strictly negative), then $\v{\pi}_1$ (resp., $\v{\pi}_0$) strictly first order stochastically dominates $\v{\pi}_0$ (resp., $\v{\pi}_1$), so $H_0$ cannot hold.  (The proof of this result leverages the fact that our Markov chain model is stochastically monotone, cf.~\cite{keilson1977monotone}.)

\begin{proposition}
\label{pr:null_monotone}
Suppose that the treatment is sign-consistent (cf.~Definition \ref{def:monotone}) and that $H_0$ is true, i.e., $\GTE = 0$.  Then $p_1(k) = p_0(k)$ for all $k$.
\end{proposition}

Thus if interventions are sign-consistent, then the {\em false positive probability} ($\FPP$) of the decision rule can be defined as follows:
\begin{equation}
\label{eq:FPP} 
\FPP_N = \P(\hat{T}_N > \Phi_{\alpha/2} | H_0) = \P(|\hat{T}_N| > \Phi_{\alpha/2} | \v{p}_1 = \v{p}_0).
\end{equation}

For a specific alternative with $\v{p}_0 \neq \v{p}_1$, we analogously define the {\em false negative probability} ($\FNP$):
\begin{equation}
\label{eq:FNP} 
\FNP_N(\v{p}_0, \v{p}_1) = 1 - \P(|\hat{T}_N| > \Phi_{\alpha/2} | \v{p}_0, \v{p}_1).
\end{equation}
The {\em statistical power} is then $1 - \FNP_N(\v{p}_0, \v{p}_1) = \P(|\hat{T}_N| > \Phi_{\alpha/2} | \v{p}_0, \v{p}_1)$.

In the remainder of our paper, we study the false positive probability and statistical power of the decision rule \eqref{eq:decision}.  In Section \ref{sec:fpp}, we study the false positive probability, and in Section \ref{sec:power}, we study the statistical power.

%% file: Sections/fpp.tex
\section{False positive probability under sign-consistent interventions}
\label{sec:fpp}

In this section we study the false positive probability of the decision rule \eqref{eq:decision}.  From Proposition \ref{pr:null_monotone}, recall that if the treatment is sign-consistent and $H_0$ is true, i.e., $\GTE =0$, then all booking probabilities in treatment and control are {\em identical}.  In other words, in this case the experimental observations correspond to an {\em A/A experiment}, i.e., tests where the same version of a product or feature is tested against itself (so that the joint distribution of observations in both the global treatment and global control conditions is the same).

Companies routinely use A/A experiments to validate the experimental setup, and ensure the accuracy of their experimentation infrastructure (see, e.g., \cite{geteppoWhatTest,optimizelyTestWhat,statsigConductTest}).  For example, an e-commerce platform might run an A/A test to ensure that user traffic is being evenly split across different servers, or a streaming service could use an A/A test to confirm that their recommendation algorithm is not unintentionally favoring one group of users over another. Such experiments are a core aspect of the validation of experimentation within any organization that adopts A/B testing.  Given Proposition \ref{pr:null_monotone}, our approach in this section is to study the false positive probability of A/A experiments.

Our first key finding is that for an A/A experiment, not only is the estimator $\DM_N$ unbiased for $\GTE = 0$, but in addition, notably, {\em the na\"{i}ve variance estimator $\widehat{\Var}_N$ is unbiased for the true variance of $\DM_N$}, even though it ignores correlations. We prove this result in a far more general setting of A/A experiments with arbitrary dependence structure between observations, which may be of independent interest to practitioners. Our inventory-constrained platform model is a specific instance of this general setting. Our approach involves combining ideas closely related to classical results in design-based inference by both Neyman \cite{neyman1923,ding2024first} and Fisher \cite{fisher1971design, wager2024causal}.

Corresponding to this result, we show that for the model described in Section \ref{sec:cc} asymptotically the test statistic $\hat{T}_N$ has a standard normal distribution, and thus in particular the false positive probability converges to $\alpha$.  In other words, despite the presence of interference, as long as the treatment is sign-consistent, the decision rule \eqref{eq:decision} has a false positive probability that is (asymptotically) controlled at level $\alpha$, as it is in a setting without interference.

\subsection{A detour: Estimation for general A/A experiments}
\label{subsec:aa_general}

To establish unbiasedness of $\DM_N$ and $\widehat{\Var}_N$ for A/A experiments, we first introduce a completely general model of A/A experiments.  In our more general model in this section, we abstract away from the system dynamics described in Section \ref{sec:cc}, and instead assume nothing more than a joint distribution over the $N$ observations $Y_1, \ldots, Y_N$ that corresponds to the baseline system.  Our formal specification is as follows.

\paragraph{Sample and joint distribution.}  As before, $Y_1, \ldots, Y_N$ denote $N$ observations, constituting the sample in the experiment; we refer to each index $i$ as a {\em unit}.  We let $\Q$ denote the joint probability distribution of $Y_1, \ldots, Y_N$.  We interpret $\Q$ as the distribution of the observations under an A/A test.

A key feature of our definition is that we make {\em no assumptions} about the structure of $\Q$, allowing for general dependence structure between observations.  This is a novel aspect of our analysis, as it informs the analysis of A/A experiments even in settings with interference. Note that in the model of Section \ref{sec:cc}, the distribution $\Q$ corresponds to the joint distribution of booking outcomes induced by the Markov chain specification for a given global treatment condition.

For simplicity we assume observations are bounded, i.e., there exists $M$ such that $\sup_i |Y_i| \leq M$.  (In Section \ref{sec:cc}, $M = 1$.)

\paragraph{Bernoulli treatment assignment.}  As before we let $Z_1, \ldots, Z_N$ represent i.i.d.~Bernoulli($a$) random variables corresponding to treatment assignment; in an A/A experiment, these are {\em independent} of $Y_1, \ldots, Y_N$.  We assume that $0 < a < 1$, and as before define $N_1 = \sum_i Z_i$, $N_0 = \sum_i (1 - Z_i)$.

\paragraph{Null hypothesis.}  Because we focus on false positive probabilities of A/A experiments, informally our null hypothesis is that the treatment and control data generating processes are identical.  We formalize this as the following null hypothesis:
\begin{equation}
\label{eq:aa_H0}
 \tilde{H}_0 : (Y_1, \ldots, Y_N)\text{ is independent of } (Z_1, \ldots, Z_N).    
\end{equation} 
We refer to any experiment where $\tilde{H}_0$ holds as an {\em A/A experiment}.  For the duration of this subsection, we assume that $\tilde{H}_0$ holds.

$\tilde{H}_0$ is a strong null hypothesis, as it asserts that regardless of the treatment assignment, the distribution of the observations is {\em unchanged}.  It is analogous to Fisher's sharp null hypothesis in randomization or design-based inference; see, e.g., \cite{fisher1971design,imbens2015causal,wager2024causal}.  In particular, suppose we adopt the potential outcomes framework, and view each unit as having two (random) potential outcomes $Y_i(0), Y_i(1)$, with $Y_i = Y_i(Z_i)$ being the observed outcome.  Then \eqref{eq:aa_H0} can be alternatively defined as the hypothesis that $Y_i(0) = Y_i(1)$ for all $i$, which is exactly Fisher's sharp null hypothesis in randomization inference in the absence of interference.

Note under $\tilde{H}_0$ in the inventory-constrained platform model of Section \ref{sec:cc}, we trivially have $\GTE = 0$.  Indeed, we expect that any reasonable estimand of interest should be zero in an A/A experiment.

\paragraph{Estimation.} The difference-in-means estimator $\DM_N$ and the naive variance estimator $\widehat{\Var}_N$
are the same as \eqref{eq:DM} and \eqref{eq:var_hat}.

\paragraph{Main result: Finite sample expectations of $\DM_N$ and $\widehat{\Var}_N$.}  We have the following theorem.

\begin{theorem}
\label{thm:aa_expectations}
In an A/A experiment, there holds:
\begin{equation}
\E[\DM_N | N_1 > 0, N_0 > 0] = 0; \quad \text{and} \quad \E[\widehat{\Var}_N | N_1 > 1, N_0 > 1] = \Var(\DM_N | N_1 > 1, N_0 > 1).
\end{equation}
\end{theorem}

The first part of the theorem shows that $\DM_N$ has zero expectation, which is unsurprising.  The second part of the theorem is more surprising, and shows that despite the generality of the model, the naive variance estimator $\widehat{\Var}_N$ is {\em unbiased} for the true variance of $\DM_N$.  

We next establish this result below in two steps.  First, we show in Proposition \ref{pr:reduction_to_exchangeable} that we can assume without loss of generality that the observations $Y_1, \ldots, Y_N$ are {\em stochastically exchangeable} under $\Q$.  Next, we show in Proposition \ref{pr:gte_and_var_aa_exchangeable} that the result of Theorem \ref{thm:aa_expectations} holds when observations are exchangeable. Combining these propositions yields the theorem.

\paragraph{Step 1: Reduction to stochastically exchangeable observations.} We start by establishing  that we can analyze general A/A experiments by assuming that the distribution of observations is {\em exchangeable} (in the probabilistic sense).  We recall that a vector $\v{X} = (X_1, \ldots, X_N)$ of random variables is exchangeable if $\v{X}$ has the same joint distribution as any permutation of the elements of $\v{X}$.  We emphasize that in our model, observations are generally not exchangeable due to capacity constraints, but our reduction establishes that, without loss of generality, it suffices to consider a system where observations are exchangeable.

To prove the reduction, let $\sigma$ be an independent permutation of $\{1,\ldots,N\}$ chosen uniformly at random.  Define $\tilde{Y}_i = Y_{\sigma(i)}$.   Observe that the vector $(\tilde{Y}_1, \ldots, \tilde{Y}_N)$ is trivially exchangeable, because $\sigma$ was uniformly chosen.

We let $\tilde{\Q}$ denote the joint distribution of $(\tilde{Y}_1, \ldots, \tilde{Y}_N)$. Define:
\[ \widetilde{\GTE}_N = \tilde{\b{Y}}(1) - \tilde{\b{Y}}(0), \]
where:
\[ \tilde{\b{Y}}(1) = \frac{\sum_i Z_i \tilde{Y}_i}{N_1}; \quad \tilde{\b{Y}}(0) =  \frac{\sum_i (1 - Z_i) \tilde{Y}_i}{N_0}. \]
(Recall that $N_1 = \sum_i Z_i$ and $N_0 = \sum_i (1 - Z_i)$.)  Analogously, define the variance estimator:
\[ \widetilde{\Var}_N = \frac{1}{N_1(N_1 - 1)} \sum_i Z_i( \tilde{Y}_i  - \tilde{\b{Y}}(1))^2 + 
        \frac{1}{N_0(N_0-1)} \sum_i (1 - Z_i) (\tilde{Y}_i - \tilde{\b{Y}}(0))^2.\]
(As before, $\DM_N$ and $\widetilde{\GTE}_N$ are only well defined if $N_1, N_0 > 0$; and $\widehat{\Var}_N$ and $\widetilde{\Var}_N$ are only well defined if $N_1, N_0 > 1$.)

We have the following result; the proof is in Appendix \ref{app:proofs}.

\begin{proposition}
\label{pr:reduction_to_exchangeable}
Conditional on $N_1 > 1$ and $N_0 > 1$, the pair $(\DM_N, \widehat{\Var}_N)$ and the pair $(\widetilde{\GTE}_N, \widetilde{\Var}_N)$ have the same joint distribution.
\end{proposition}

Proposition \ref{pr:reduction_to_exchangeable} shows that we can assume (probabilistic) exchangeability of the observations without loss of generality when studying the bias of $\DM_N$ and $\widehat{\Var}_N$.  We note that the concept of probabilistic exchangeability is also leveraged in Bayesian inference \cite{bernardo2009bayesian}. A related concept plays a role in causal inference \cite{robins1992identifiability,hernan2006estimating}, where (for data from experiments), ``exchangeability'' is a definition that guarantees that potential outcomes are independent of treatment assignment; note that even the original outcomes $Y_1, \ldots, Y_N$ in our experiment satisfy this criterion.  (There have been works \cite{lindley1981role,saarela2023role} aiming to unify these definitions.)  In all these lines of work, exchangeability is an assumption imposed on the causal model, in contrast to our result in Proposition \ref{pr:reduction_to_exchangeable}.

\paragraph{Step 2: Expected values of $\DM_N$ and $\widehat{\Var}_N$ in A/A experiments with exchangeable observations.} Proposition \ref{pr:reduction_to_exchangeable} shows that in analyzing the distribution of $\DM_N$ and $\widehat{\Var}_N$, we can assume without loss of generality that the observations in the A/A experiment are exchangeable.  For the remainder of the section, therefore, we make the following assumption.

\begin{assumption}
\label{as:exchangeable}
The observations $Y_1, \ldots, Y_N$ are exchangeable under the measure $\Q$.     
\end{assumption}

We refer to an A/A experiment with observations satisfying Assumption \ref{as:exchangeable} as an {\em A/A experiment with exchangeable observations}.   
To fix notation, for an A/A experiment with exchangeable observations, we let $\mu = \E[Y_i]$ and $V = \Var(Y_i)$ for each $i$, and $C = \Cov(Y_i, Y_j)$ for $i \neq j$, with all expectations taken under $\Q$.  The following proposition establishes unbiasedness of $\DM_N$ and $\widehat{\Var}_N$; see Appendix \ref{app:proofs} for the proof.

\begin{proposition}
\label{pr:gte_and_var_aa_exchangeable}
Consider an A/A experiment with exchangeable observations, i.e., Assumption \ref{as:exchangeable} holds.  For all $\v{Z}$ with $N_0 > 0$ and $N_1 > 0$, there holds $\E[\DM_N | \v{Z}] = 0$.  In particular , $\E[\DM_N | N_0 > 0, N_1 > 0] = 0$.  Further:
\[ \E[\widehat{\Var}_N | \v{Z}] = \Var(\DM_N | \v{Z})  = \left( \frac{1}{N_1} + \frac{1}{N_0} \right) (V - C). \]
In particular:
\[ \E[\widehat{\Var}_N | N_0 > 1, N_1 > 1] =
    \Var(\DM_N | N_0 > 1, N_1 > 1)  =  
    \E \left[\frac{1}{N_1} + \frac{1}{N_0} \Bigg| N_0 > 1, N_1 > 1 \right] (V - C). \]
\end{proposition}

Proposition \ref{pr:gte_and_var_aa_exchangeable} reveals that even if observations are correlated, the commonly used estimators $\DM_N$ and $\widehat{\Var}_N$ are both {\em unbiased} for their respective targets in A/A experiments.  The first result is straightforward under Assumption \ref{as:exchangeable}, since $\E[Y_i] = \mu$ for all $i$.  The second result reveals that conditional on the realization of $\v{Z}$, the exact variance of the DM estimator is $(1/N_1 + 1/N_0) (V - C)$.  Although the estimator $\widehat{\Var}_N$ does not include any estimation of covariance terms, the covariance between each observation and the group sample mean ($\b{Y}(1)$ and $\b{Y}(0)$ for treatment and control, respectively) exactly ``corrects'' for covariance in the estimator.

Taken together, Propositions \ref{pr:reduction_to_exchangeable} and \ref{pr:gte_and_var_aa_exchangeable} establish  Theorem \ref{thm:aa_expectations}.
Our results in these propositions are related to classical results on design-based inference by both Neyman \cite{neyman1923, ding2024first} and Fisher \cite{fisher1971design, wager2024causal}; indeed, our result can be considered a unification of ideas from their contributions.  In particular, our use of exchangeability is similar to Fisher's development of permutation testing, while our analysis of $\widehat{\Var}_N$ in the exchangeable setting is related to Neyman's result in a setting where potential outcomes are constant \cite{neyman1923}.  We discuss this relationship in more detail in Appendix \ref{app:fpp_proofs}.

We conclude by noting that statisticians have considered many variants of two-sample t-test statistics when observations might be correlated; see, e.g., \cite{rosner2006fundamentals}, Chapter 8.  A leading example is the {\em paired} two-sample t-test, where treatment and control observations are collected on the same unit; other examples include settings with clustering or temporal correlation.  All these examples lead to changes to the t-test statistic to deal with correlation.  By contrast, in our case, the results in Propositions \ref{pr:reduction_to_exchangeable} and \ref{pr:gte_and_var_aa_exchangeable} are strongly suggestive that the practitioner running an A/A experiment can use {\em the usual t-test statistic}, without any adjustment needed for correlation.  We study one such example next: inventory-constrained platforms with sign-consistent interventions.

\subsection{False positive probability in inventory-constrained platforms with sign-consistent interventions}
\label{sec:aa_monotone}

We now return to the inventory-constrained platform model of Section \ref{sec:cc}, and in particular characterize the behavior of the false positive probability $\FPP_N$ defined in \eqref{eq:FPP}.  

\paragraph{Finite sample analysis.}  We start by recalling that if the null hypothesis $H_0 : \GTE = 0$ holds {\em and} the treatment is sign-consistent, then by Proposition \ref{pr:null_monotone} we know that $\v{p}_0 = \v{p}_1$, i.e., we are in the setting of an A/A experiment: both treatment and control Markov chains are identical, and of course this is a special case of the general model of A/A experiments considered in the previous section.  Therefore, by Theorem \ref{thm:aa_expectations}, we conclude that $\E[\DM_N | H_0, N_0 > 0, N_1 > 0] = \GTE = 0$, and $\E[\widehat{\Var} | H_0, N_0 > 1, N_1 > 1] = \Var(\DM | H_0, N_0 > 1, N_1 > 1)$, i.e., both are unbiased for their target estimands. %

\paragraph{Asymptotic false positive probability.}  Given the previous result, we should expect that asymptotically, the t-test statistic $\hat{T}_N = \DM_N/\sqrt{\widehat{\Var}_N}$ should be approximately a standard normal random variable, as long as $\DM_N$ obeys a central limit theorem.  Further, if this is the case, then the false positive probability should approach the target level $\alpha$ for the decision rule \eqref{eq:decision} as $N \to \infty$.  Indeed, the following theorem establishes this result.  The theorem is an application of Theorem \ref{thm:aa_clt} for a general setting of Markovian system dynamics; we refer the reader to Appendix \ref{app:MC_aa} and \ref{app:MC_cc} for details.

\begin{theorem}
\label{thm:aa_clt_cc}
Suppose $0 < a < 1$.  Suppose also that $H_0$ holds, i.e., $\GTE = 0$, and the treatment is sign-consistent (cf.~Definition \ref{def:monotone}).  By Proposition \ref{pr:null_monotone}, we have $\v{p}_0 = \v{p}_1$; thus define $V = V(0) = V(1)$, cf.~\eqref{eq:V_cc}.

Then $\DM_N$ obeys the following central limit theorem as $N \to \infty$:  
\begin{equation}
\sqrt{N} \DM_N \Rightarrow \mc{N}\left(0, \left(\frac{1}{a} + \frac{1}{1-a}\right) V\right)
\end{equation}
In addition, $\widehat{\Var}_N$ satisfies
\begin{equation}
N \widehat{\Var}_N \to^p \left(\frac{1}{a} + \frac{1}{1-a}\right) V.
\end{equation}
In particular, the t-test statistic $\hat{T}_N$ converges in distribution to a standard normal random variable as $N \to \infty$:
\[ \hat{T}_N = \frac{\DM_N}{\sqrt{\widehat{\Var}_N}} \Rightarrow \mc{N}(0,1).\]
Thus as $N \to \infty$, there holds $\FPP_N \to \alpha$.
\end{theorem}

Theorem \ref{thm:aa_clt_cc} illustrates that despite the presence of interference, if we restrict attention to sign-consistent treatments, then the false positive probability is correctly controlled at level $\alpha$ by the decision rule \eqref{eq:decision}.  Note that because the test statistic is asymptotically standard normal, we can conclude that if the treatment is sign-consistent, then frequentist p-values (e.g., as are commonly displayed in experimentation dashboards in industry) are (asymptotically) correct under $H_0$ as well.  Taken together, the preceding results suggest that when treatments are sign-consistent, there are no benefits to false positive probability to be gained by implementing debiasing techniques for either $\DM_N$ or $\widehat{\Var}_N$.  

We note that in comparison with Proposition \ref{pr:gte_and_var_aa_exchangeable}, there is no covariance term in the limit.  Informally, this is because the analog of $C$ in Proposition \ref{pr:gte_and_var_aa_exchangeable} in this setting is the {\em average} covariance between the first $N$ observations.  This average covariance decays to zero as $N \to \infty$ in our stationary Markov chain setting, so we are only left with the variance.  We can also see this another way, by comparing the preceding result with Theorem \ref{thm:ab_clt_cc}, and in particular \eqref{eq:tilde_sigma_cc}.  In contrast to those results, the covariance terms exactly cancel each other out in Theorem \ref{thm:aa_clt_cc}; informally, this happens because the ``within-group" covariances $C_j(0,0) + C_j(1,1)$ cancel the ``across-group" covariances $C_j(0,1) + C_j(1,0)$ for each $j$, cf.~\eqref{eq:Cj_cc}.

 We conclude by noting that Theorem \ref{thm:aa_clt_cc} is similar in spirit to a longstanding literature on Fisherian randomization inference with ``studentized" statistics (such as the $t$-test statistic we consider). Here, the goal is typically to provide studentized test statistics that can correctly control asymptotic type I error under weaker null hypotheses than Fisher's sharp null (\cite{chungromano2013,romano1990,wu2021randomization}); by contrast, our result is in the setting where $H_0$ is equivalent to Fisher's sharp null.  Given these results, we conjecture that similar asymptotic control of the false positive probability could be obtained in any setting where the tail of the variance estimator $\widehat{\Var}$ is well controlled.

%% file: Sections/power.tex
\section{Statistical power under sign-consistent interventions}
\label{sec:power}

Although $\DM_N$ and $\widehat{\Var}_N$ are unbiased under $H_0$ when treatments are sign-consistent (i.e., when $\v{p}_1 = \v{p}_0$), both estimators will be biased in general if $\v{p}_0 \neq \v{p}_1$.  
An important implication of the preceding section is that if treatments are sign-consistent, then the implementation of a debiased estimation and hypothesis testing strategy can only be beneficial if it increases the statistical power when $\v{p}_0 \neq \v{p}_1$.  In this section, we focus on the inventory-constrained platform model with sign-consistent treatments, and we study the asymptotic statistical power of the test statistic \eqref{eq:tstat}.   

In particular, we compare the estimator in \eqref{eq:tstat} to the test statistic obtained if the decision maker uses an unbiased estimator for $\GTE$ instead (with an associated unbiased variance estimator).  In Section \ref{sec: finite sample power}, we gain intuition for our asymptotic analysis by starting with a finite sample analysis: we show that if treatments are sign-consistent then $|\E[\DM_N]|$ will be larger than $|\GTE|$, and in particular positive (resp., negative) treatments result in positive (resp., negative) $\GTE$ bias.  

This finding hints that if the variance of any unbiased estimator for $\GTE$ is no smaller than $\widehat{\Var}$, then the power of the naive test statistic $\hat{T}_N$ in \eqref{eq:tstat} should be larger as well. In Section \ref{sec: asymptotic power}, we investigate this behavior.  We simulate a lower bound on the variance of an unbiased estimator provided by \cite{farias2022markovian}, and demonstrate via numerics that it is typically larger than $\widehat{\Var}_N$. As a consequence, we find in numerics that the naive test statistic $\hat{T}_N$ has higher power.  Taken together, these findings demonstrate that the interaction between the bias and variance of na\"{i}ve estimation can be more favorable for controlling false negatives than unbiased estimation.

\subsection{Finite sample analysis}
\label{sec: finite sample power}

We first consider the behavior of $\E[\widehat{\GTE}_N]$ for a fixed $N$, when the treatment is sign-consistent.  In particular, we establish that
if the treatment is strictly positive (resp., strictly negative) then $\E[\widehat{\GTE}_N] > \GTE$ (resp., $\E[\widehat{\GTE}_N] < \GTE$).  Our technical approach involves establishing stochastic dominance relations between the state distribution each arriving customer encounters, in global treatment, global control, and the experimental systems, respectively.  

\paragraph{Distribution seen by each customer.} Let $\v{Z} = (Z_1, \dots, Z_N) \in \{0,1\}^N$ denote the realized vector of treatments assigned to each customer.  For each customer $i=1,\dots,N$, let $A_i$ denote the arrival time of the $i$th customer. For each realization of $\v{Z}$, define:
\begin{equation}
    \label{eq: ith customer distribution}
    \nu_i^{\v{Z}}(k) := \liminf_{\epsilon \rightarrow 0} \P\left(X_{A_i-\epsilon}=k|\v{Z}\right) \text{ for } k = 0,\dots,K.
\end{equation}
In other words, $\v{\nu}^{\v{Z}}_i$ is the state distribution seen by the $i$'th customer when they arrive to the system, if the sequence of treatment assignments is $\v{Z}$.  Note that $\v{\nu}^{\v{Z}}_i$ only depends on the treatment assignments $Z_1, \ldots, Z_{i-1}$ prior to customer $i$.  Finally, we note that the probability above is over randomness in the entire state process $\{X_t\}$ (including the arrival times $\{A_j\}$ over all customers $j$), but conditional on the exact treatment assignment vector $\v{Z}$.
We let $\v{1}$ (resp., $\v{0}$) denote the vector of $N$ ones (resp., $N$ zeros); note that this vector is $N$-dimensional, though we suppress this dependence on $N$ to simplify notation.  Observe that $\v{\nu}_i^{\v{1}}$ (resp., $\v{\nu}_i^{\v{0}}$) denotes the distribution as seen by the $i$th customer when the system is in global treatment, i.e., $a = 1$ (resp., global control, i.e., $a = 0$).

\paragraph{Stochastic dominance.}  Our analysis requires stochastic comparison of the distributions $\v{\nu}_i^{\v{Z}}$ for different values of $\v{Z}$.  We have the following definition.

\begin{definition}
\label{def:stochdom}
If $\v{\nu}=(\nu(0),\dots,\nu(K))$ and $\v{\mu}=(\mu(0),\dots,\mu(K))$ are probability distributions on $(0, \ldots, K)$, then we say $\v{\nu}$ \emph{stochastically dominates} $\v{\mu}$ and write $\v{\nu} \succ^d \v{\mu}$ if for all $k \geq 0$ we have
$$\sum_{k'=k}^K \nu(k') \geq \sum_{k'=k}^K \mu(k'),$$
with strict inequality for at least one value of $k$.
\end{definition}

\paragraph{Bias of $\DM_N$.}  Our finite sample analysis of the bias of $\DM_N$ proceeds by establishing stochastic dominance relationships among $\v{\nu}_i^{\v{Z}}$, depending on the treatment assignment $\v{Z}$. Our key result (Proposition \ref{prop:distributiondom} in Appendix \ref{app:proofs}) establishes that for a strictly positive treatment, and for any treatment assignment $\v{Z}$ with $N_1, N_0 > 0$, for every arriving customer $i$ one of the following three cases holds:
\begin{enumerate}
    \item[(1)] $\v{\nu}_i^{\v{Z}} =\v{\nu}_i^{\v{1}}$; 
    \item[(2)] $\v{\nu}_i^{\v{Z}} = \v{\nu}_i^{\v{0}}$; or
    \item[(3)] $\v{\nu}_i^{\v{0}} \prec^d \v{\nu}_i^{\v{Z}} \prec^d \v{\nu}_i^{\v{1}}$.
\end{enumerate}
Cases (1) and (2) correspond to the case where every arriving customer prior to customer $i$ was in treatment or control, respectively; in this case, the state distribution seen by customer $i$ under $\v{Z}$ corresponds to what they would have seen under global treatment or global control, respectively.  Case (3) is the interesting case: if neither (1) or (2) hold, then the state distribution for the experimental treatment assignment lives strictly ``between" the global control and global treatment conditions (in the sense of stochastic dominance).  (If the treatment is strictly negative, these inequalities are reversed.)

Our proof approach leverages an elementary comparison argument together with stochastic monotonicity (cf.~\cite{keilson1977monotone}) of an appropriate underlying pure death chain.  Informally, stochastic monotonicity guarantees us that stochastic dominance relationships are preserved over time.   We use this approach to establish that under a strictly positive treatment, whenever an arriving customer is allocated to treatment, this ``increases" the state distribution (in the sense of stochastic dominance); and that increase is preserved over time.  Corollary \ref{cor:invariant_start_dist} in Appendix \ref{app:proofs} extends this argument to the steady state limit, to allow comparisons to the steady state distributions $\v{\pi}_0$ and $\v{\pi}_1$.

Once we have established the preceding insight, it is straightforward to establish that for strictly positive interventions, the DM estimator has a positive bias.  This is a finite sample analog of similar ``mean field" or fluid limit results seen in other platform settings; see, e.g., \cite{johari2022experimental,li2022interference,bright2022reducing}.  (We establish a similar asymptotic result in our setting in the next section.)  

Before stating this theorem, we impose the simplifying, benign assumption that the first arriving customer sees the system in global control; since the distribution on arrival of the first customer does not depend on any treatment assignments, we simply write $\v{\nu}_1$ for this distribution (without superscript $\v{Z}$).  (In fact our result holds for any initial distribution that lies between $\v{\pi}_0$ and $\v{\pi}_1$, in the sense of stochastic dominance.)
\begin{assumption}
     \label{as: first customer distr}
For the first customer, $\v{\nu}_1 = \v{\pi}_0$.
\end{assumption}

We have the following theorem.   See Appendix \ref{app:proofs} for the proof.
\begin{theorem}
    \label{thm: pos bias DM}
Suppose the treatment is strictly positive, cf.~Definition \ref{def:monotone}, and Assumption \ref{as: first customer distr} holds.  For all $\v{Z}$ with $N_0>0$ and $N_1>0$, there holds $\E[\widehat{\GTE}_N|\v{Z}]> \GTE > 0$. In particular, $\E[\widehat{\GTE}_N|N_0>0,N_1>0]>\GTE>0$.  If the treatment is strictly negative, these inequalities are reversed.  Thus for any strictly sign-consistent intervention, there holds $|\E[\widehat{\GTE}_N|N_0>0,N_1>0]|>|\GTE|$.
\end{theorem}
We remark that Assumption \ref{as: first customer distr} is not essential.  It is in fact possible to show the same result for any distribution $\v{\nu}_1$ such that $\v{\pi}_0 \prec^d \v{\nu}_1 \prec^d \v{\pi}_1$, as well as $\v{\nu}_1 = \v{\pi}_1$; we omit the details.

\subsection{Asymptotic analysis}
\label{sec: asymptotic power}

In this section, we consider the asymptotic behavior of statistical power when the treatment is sign-consistent.

\paragraph{Asymptotic bias of $\DM_N$.}  We start by noting that Theorem \ref{thm: pos bias DM} shows that the DM estimator systematically overestimates $\GTE$ in magnitude when the treatment is strictly sign-consistent.  Unsurprisingly, a similar result is true asymptotically as well, as we show in the next theorem; see Appendix \ref{app:proofs} for the proof.  Recall from Theorem \ref{thm:ab_clt_cc} that $\ADE_a$ is the approximate asymptotic mean of $\DM_N$. 
\begin{theorem}
    \label{thm: pos bias limit}
    If the treatment is strictly positive (cf.~Definition \ref{def:monotone}), then $\ADE_a>\GTE>0$. If the treatment is strictly negative,  then $\ADE_a<\GTE<0$.  Thus for any strictly sign-consistent intervention, there holds $|\ADE_a| > |\GTE| > 0$.
\end{theorem}

\paragraph{Comparison to decision-making using unbiased estimators.} Building from Theorem \ref{thm: pos bias limit}, we compare and contrast two approaches to decision-making when $\GTE \neq 0$.  First, we consider the test statistic \eqref{eq:tstat}, formed using the DM estimator $\DM_N$ and the associated variance estimator $\widehat{\Var}_N$, and applying the decision rule \eqref{eq:decision}.

Second, we consider a decision-making procedure using a hypothetical unbiased estimator $\hat{\theta}_N$ for $\GTE$; for example, \cite{farias2022markovian} constructs an unbiased estimator using an off-policy evaluation (OPE) approach based on least-squares temporal-differences, or LSTD, learning.  Such an estimator $\hat{\theta}_N$ is a real-valued functional of the observations $Y_1, \ldots, Y_N$ generated from a Bernoulli customer-randomized experiment (cf.~Section \ref{sec:cc}).  We say that $\hat{\theta}_N$ is an {\em unbiased estimator} if for all feasible values of system parameters (i.e., $K$, $\lambda$, $\v{\tau}$, $\v{p}_1$, $\v{p}_0$) and the design parameter $a$, there holds $\E[\hat{\theta}_N] = \GTE$.  For the remainder of this section, we assume that $\hat{\theta}_N$ is an unbiased estimator, and we have access to the associated (true) variance $\Var(\hat{\theta}_N)$. 

Suppose we form the following test statistic:
\begin{equation}
\label{eq:ub_tstat}
 \hat{U}_N = \frac{\hat{\theta}_N}{\sqrt{\Var(\hat{\theta}_N)}}. 
\end{equation}
and apply the same decision rule as \eqref{eq:decision}, but with the test statistic $\hat{U}_N$ instead.  How does this method compare to the decision rule using $\hat{T}_N$?  Which method yields higher power depends on the relative behavior of the numerators and denominators of $\hat{T}_N$ and $\hat{U}_N$, respectively, as $N \to \infty$.  We note that since Theorem \ref{thm: pos bias limit} implies $|\ADE_a| > |\GTE|$, we expect the numerator of $\hat{T}_N$ is asymptotically larger than the numerator of $\hat{U}_N$.  Therefore our emphasis is on comparing the denominators of $\hat{T}_N$ and $\hat{U}_N$, respectively, i.e., comparing $\widehat{\Var_N}$ with $\Var(\hat{\theta}_N)$.

\paragraph{Bounding the variance of an unbiased estimator.}  To study $\Var(\hat{\theta}_N)$, we leverage Theorem 3 of \cite{farias2022markovian} which uses the multivariate Cram\'{e}r-Rao bound to give a lower bound on the variance of any unbiased estimator.\footnote{The result presented in \cite{farias2022markovian} is for $a = 1/2$, but it can be verified that the expression given there generalizes to any $0 < a< 1$.}  In particular, for any system and design parameters $K$, $\lambda$, $\v{\tau}$, $\v{p}_1$, $\v{p}_0$, and $a$, their theorem provides a quantity $\sigma^2_{\UB}$ such that for any $N$ and any unbiased estimator $\hat{\theta}_N$, there holds $N \Var(\hat{\theta}_N) \geq \sigma^2_{\UB}$.  This bound is also tight: the unbiased OPE LSTD estimator of \cite{farias2022markovian} achieves the lower bound; the bound is also met asymptotically as $N \to \infty$ by the consistent nonparametric maximum likelihood estimator proposed by \cite{glynn2020adaptive}.\footnote{The estimator of \cite{glynn2020adaptive} directly (unbiasedly) estimates the transition matrices and reward functions of treatment and control; this is the source of the variance estimates in the Cram\'{e}r-Rao bound.  In \cite{glynn2020adaptive}, these primitives are then used to estimate $\GTE$; the latter estimator may be biased, as it involves inversion of the estimated transition matrices to compute steady state distributions.}
Note that $\sigma^2_{\UB}$ depends on the system and design parameters (i.e., $K$, $\lambda$, $\v{\tau}$, $\v{p}_1$, $\v{p}_0$, and $a$), but {\em not} the sample size $N$.  We suppress the parameter dependence of $\sigma^2_{\UB}$ for notational simplicity.   %
In general, $\sigma^2_{\UB}$ can be quite large: in \cite{farias2022markovian}, it is shown that there is an example class of Markov chains where $\sigma^2_{\UB}$ scales exponentially in the size of the state space.  For the sake of brevity, we leave the formal presentation of Theorem 3 of \cite{farias2022markovian} and its mapping to our setting to Appendix \ref{app:CR}.  

In \cite{farias2022markovian}, an example is given of a sequence of Markov chains (together with associated reward functions) with state space size $K \to \infty$, for which the Cram\'{e}r-Rao lower bound increases exponentially with $K$ (see Theorem 4 of that paper).  Note that, by contrast, $N \widehat{\Var}_N$ is asymptotically bounded above by $(1/a + 1/(1-a))(1/4)$, since $V(0)$ and $V(1)$ in the Theorem \ref{thm:var_est} are each no larger than $1/4$.  In the worst case, this pair of facts suggests that unbiased estimators can have exponentially worse scaling with $K$ than the naive variance estimator.

We ask: is exponential growth of $\sigma^2_{\UB}$ idiosyncratic, or unavoidable?  We now formally show that this finding is in fact generic, for any capacity-constrained platform experiment of the type we consider in this paper, with a sign-consistent treatment: the exponential increase in variance for any unbiased estimator is in fact unavoidable for the practitioner.  Establishing such a result generically requires construction of an appropriate sequence of capacity-constrained platforms with sign-consistent treatments.  In the following theorem, we consider a sequence of treatment and control systems that approach appropriate ``mean-field'' limits, respectively, as $K \to \infty$.  We use this limiting system to show that for any such sequence, regardless of parameters, the Cram\'{e}r-Rao lower bound $\sigma^2_{\UB}$ grows exponentially in the size of the state space.  The proof of the theorem can be found in Appendix \ref{app:MC_cc}. 

\begin{theorem}
\label{thm:cr_expo_growth}
Fix $\b{\lambda} > 0$, and continuous functions $\b{p}_0, \b{p}_1, \b{\tau}: [0,1] \to \mathbb{R}$ such that $0 < \b{p}_0(s) < \b{p}_1(s) < 1$ for all $s < 1$; $\b{p}_0, \b{p}_1$ are strictly decreasing; $\b{\tau}(s) > 0$ for all $s > 0$; and $\b{\tau}$ is nondecreasing.  Further, assume that $\b{p}_z(1) = 0$, and $\b{\tau}_z(0) = 0$.

For each $K = 1, 2, \ldots$, define the following quantities:
\begin{enumerate}
    \item $\lambda^{(K)} = K \b{\lambda}$;
    \item $\tau^{(K)}(x) = K \b{\tau}(x/K)$, for $x = 0, \ldots, K$;
    \item $p_z^{(K)}(x) = \b{p}_z(x/K)$, for $z = 0,1$, and $x = 0, \ldots, K$.
\end{enumerate}
Then for each $K$, this defines a pair of control and treatment systems with strictly sign-consistent (positive) treatment.  Let $\v{\pi}_z^{(K)}$, $z = 0,1$ be the stationary distribution of the global control and global treatment Markov chains, respectively; and let $\v{\pi}_a^{(K)}$ be the stationary distribution of the experiment chain, with treatment probability $a$ (where $0 < a < 1$).

Let $\sigma^2_{\UB}(K)$ denote the Cram\'{e}r-Rao lower bound for the $K$'the system.  Then there exists constants $C > 0$ and $\gamma > 1$ (which depend on the system parameters $\b{\lambda}$, $\b{p}_0, \b{p}_1$, and $\b{\tau}$) such that for all $K$:
\[ \sigma^2_{\UB}(K) \geq C \gamma^K.\]
\end{theorem}

Informally, the proof takes advantage of the fact that as $K \to \infty$, the steady state distributions of the global control, global treatment, and experiment Markov chains concentrate on ``probable" states that are well separated from each other.  Unbiased estimation of $\GTE$ requires sufficient visits during the experiment to the ``probable" states under both global treatment and global control.  However, using the detailed balance conditions, we show that these probable states in the global control (resp., global treatment) chain are exponentially less probable (in $K$) in steady state under the experiment chain than in the global control (resp., global treatment) chain; we take advantage of this fact, together with additional technical arguments, to establish the theorem.  Note that the same theorem can be adapted to a strictly negative treatment, by changing the role of $\v{p}_0$ and $\v{p}_1$; we omit the details. 

Taken together, Theorems \ref{thm: pos bias DM}, \ref{thm: pos bias limit}, and \ref{thm:cr_expo_growth} indicate that if treatments are strictly sign-consistent, then the power of the test statistic using the DM estimator and the associated naive variance estimator will be much {\em higher} than the power of any test statistic using an unbiased estimator. We note that in Theorem \ref{thm:cr_expo_growth}, we focus on unbiased estimators that leverage data from a user-level Bernoulli randomized experiment.  In principle, one could consider instead a combination of a different design {\em and} estimator that are unbiased.  We conjecture that such an approach cannot avoid the exponential growth of variance in Theorem~\ref{thm:cr_expo_growth}.\footnote{Some evidence in this direction is that the nonparametric maximum likelihood estimation approach in \cite{glynn2020adaptive} can be applied even under more general designs than the CR experiment, and in the limit as $N \to \infty$, continues to achieve the same variance lower bound as in Theorem~\ref{thm:cr_expo_growth}.  A complete technical resolution of this issue is beyond the scope of the current paper.}

\paragraph{Numerics.}  We now briefly present numerical results to both illustrate the variance behavior of the Cram\'{e}r-Rao bound as the number of listings grows, and more importantly, to illustrate numerically the fact that statistical power is higher when using the naive decision-making approach via $\hat{T}_N$.  Our theoretical results are highly suggestive of a significant statistical power advantage for $\hat{T}_N$ over the unbiased estimation approach, and indeed our numerical results will illustrate this to be the case.

We start by noting that for any unbiased estimator $\hat{\theta}_N$, since $N \Var(\hat{\theta}_N) \geq \sigma^2_{\UB}$, we conclude that:
\begin{equation}
 |\hat{U}_N| \leq \frac{\sqrt{N} |\hat{\theta}_N|}{\sigma_{\UB}}.  \label{eq:ub_tstat_upper_bound}
 \end{equation}
We leverage this bound in our numerics to study the performance of unbiased estimation.  In particular, for the remainder of the section we assume that $\hat{\theta}_N$ is an optimal unbiased estimator, so it obeys the following central limit theorem:
\begin{equation}
\frac{\sqrt{N}(\hat{\theta}_N-\GTE)}{\sigma_{\UB}} \Rightarrow \mathcal{N}(0,1) \label{eq:unbiased_clt}
\end{equation}
as $N \to \infty$.  For general unbiased estimators, the limiting variance will be larger than $1$, since the true variance of $\hat{\theta}_N$ is lower bounded by $\sigma^2_{\UB}/N$.  Using this limit, it is straightforward to check that for any treatment such that $\v{p}_1 \neq \v{p}_0$, the power of the decision rule \eqref{eq:decision} when \eqref{eq:unbiased_clt} holds is asymptotically approximately:
\begin{equation}
\P\left(W > \Phi_{\alpha/2}-\frac{\sqrt{N}\GTE}{\sigma_{\UB}}\right) + \P\left(W < -\Phi_{\alpha/2}-\frac{\sqrt{N}\GTE}{\sigma_{\UB}}\right), \label{eq:ub_power}
\end{equation}
where $W$ is a standard normal random variable and $\Phi_{\alpha/2}$ is the upper $\alpha/2$-quantile of the standard normal distribution.  We use this calculation in our numerics below.

We consider a Bernoulli customer-randomized experiment with $K$-dependent arrival rate $\lambda^{(K)}=K\bar{\lambda}$, constant $K$-dependent death rates $\tau^{(K)}(k)= K\bar{\tau}(k/K)$ for all $k$, and $a = 0.5$. We set $\bar{\lambda}=1.5$ and $\bar{\tau}(x)=x$, so that $\tau^{(K)}(k) = k$; for further details on the numerics we refer the reader to Section \ref{app:sims} of the Appendix. We construct a $K$-dependent positive treatment using the logit booking probability model of \cite{johari2022experimental}. In particular, we suppose that a control customer receives value $v_0=0.5$ from booking a listing, and that a treatment customer receives value $v_1=v_0+\delta$ for booking a listing. All customers have outside option $\bar{\epsilon}$ for leaving the platform without booking a listing. If the customer arrives to the platform with $k$ listings booked, then the control and treatment booking probabilities respectively are given by:
\begin{align}
    p_0^{(K)}(k)&= \frac{(K-k)v_0}{\bar{\epsilon}+(K-k)v_0};      \label{eq: booking probs1}\\
    p_1^{(K)}(k)&= \frac{(K-k)v_1}{\bar{\epsilon}+(K-k)v_1}=\frac{(K-k)(v_0+\delta)}{\bar{\epsilon}+(K-k)(v_0+\delta)}.     \label{eq: booking probs2}
\end{align}
For the numerics in this section we fix $\bar{\epsilon}=1$ and $\delta=0.05$, which leads to control steady state booking probability $\rho_0 = 0.332$, treatment steady state booking probability $\rho_1 = 0.354$, and thus a $\GTE$ of $\rho_1 - \rho_0 = 0.022$. 
\begin{figure}[ht!]
  \centering
    \includegraphics[width=0.65\textwidth]{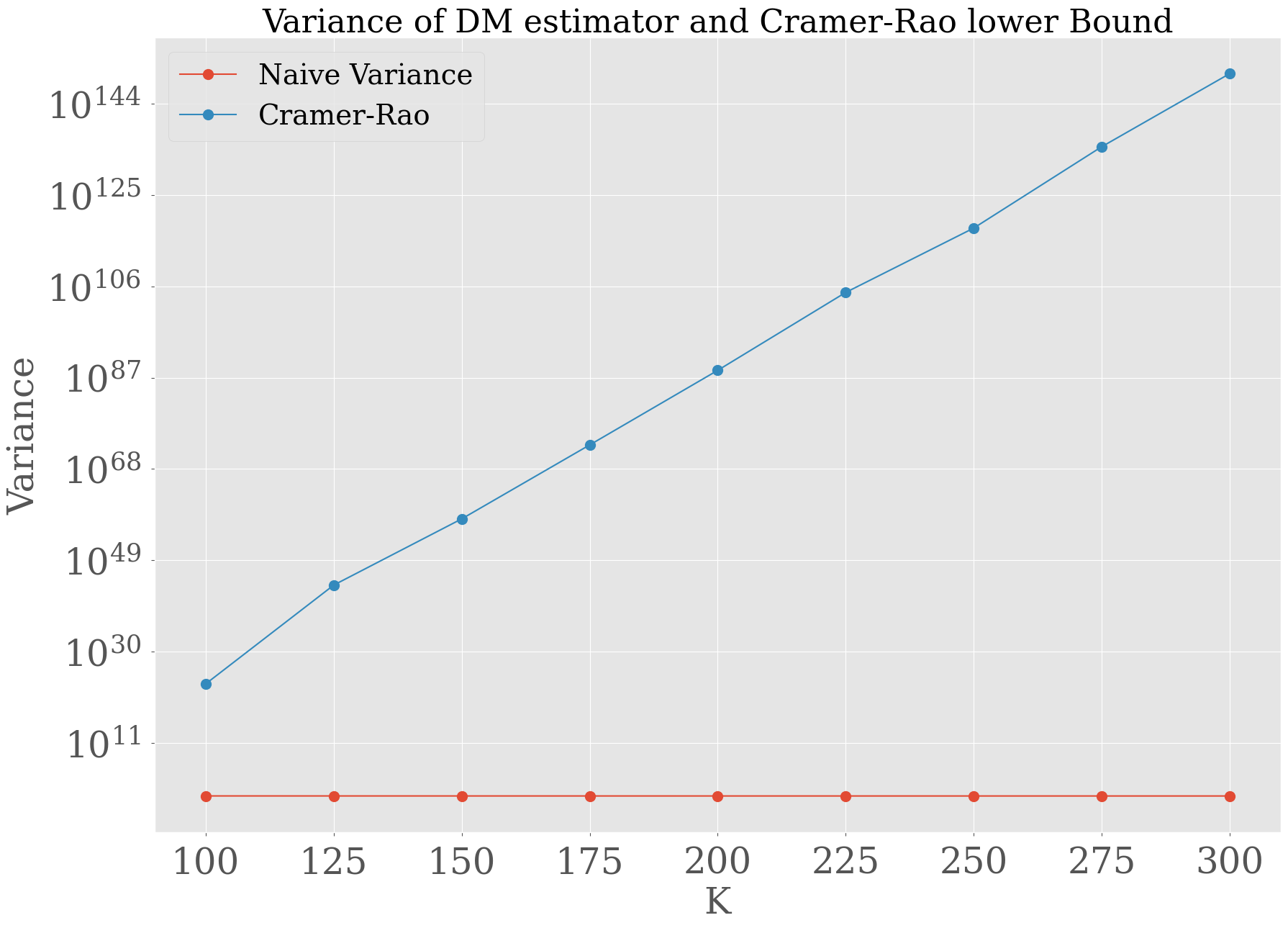}
    
    \caption{Variance behaviors for Bernoulli customer-randomized experiments as $K$ grows. We see that the Cram\'{e}r-Rao lower bound $\sigma^2_{\UB}$ (and thus the variance of any unbiased estimator) remains larger than the scaled limit of $\widehat{\Var}_N$. Parameters are $\bar{\lambda}=1.5$, $\bar{\tau}(x) = x$, $a = 0.5$ with treatment and control booking probabilities as specified in \eqref{eq: booking probs1}-\eqref{eq: booking probs2}. }
    \label{fig:variance plot K N}
\end{figure}

For the first set of numerics, we let $K$ vary from $100$ to $300$; in Appendix \ref{sec:sims var} we vary the other model parameters.  For each value of $K$, we compute the scaled limit of the naive variance estimator $\widehat{\Var}_N$ given in Theorem \ref{thm:var_est}, as well as the lower bound $\sigma^2_{\UB}$ given in Theorem \ref{thm: farias} in Appendix \ref{app:CR}.  Figure \ref{fig:variance plot K N} shows that as the number of listings grows, the bound $\sigma^2_{\UB}$ grows exponentially with $K$, consistent with Theorem \ref{thm:cr_expo_growth}; by contrast, the scaled limit of the naive variance estimator, given by $(1/a)V(1) + 1/(1-a)V(0)$ (cf.~Theorem \ref{thm:var_est}) remains bounded. In Appendix \ref{sec:sims var}, we also demonstrate that $\sigma^2_{\UB}$ remains larger than the scaled limit of the naive variance estimator even as we vary the other model parameters.

\begin{figure}[ht!]
  \centering
    \includegraphics[width=0.65\textwidth]{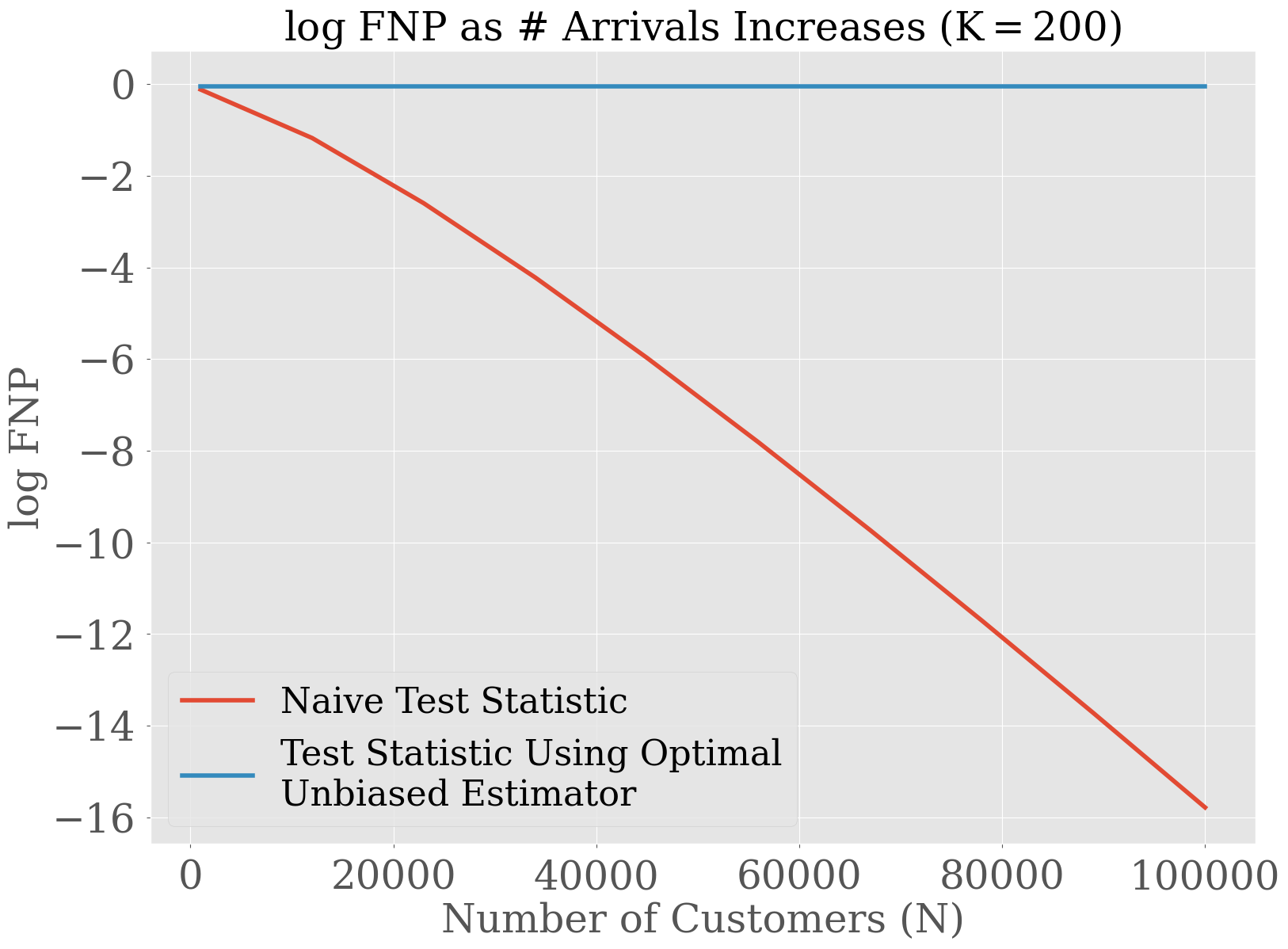}
    \caption{Log of false negative probability (FNP) in Bernoulli randomized user-level experiments with $\bar\lambda=1.5$, $\bar\tau(x) = x$ , $a = 0.5$, $K=200$, and $N$ growing. Treatment and control booking probabilities are specified as in \eqref{eq: booking probs1}-\eqref{eq: booking probs2} The log FNP of the test statistic $\hat{T}_N$ using the DM estimator $\DM$ and variance estimator $\widehat{\Var}$ decays faster than the log FNP of an unbiased test statistic (cf.~\eqref{eq:ub_power}) using an unbiased estimator that obeys \eqref{eq:unbiased_clt}.}
    \label{fig:power plot K N}
\end{figure}

For the second set of numerics, we fix $K=200$ and let $N$ vary from $10^3$ to $10^5$; in Appendix \ref{sec: sims power} we vary other model parameters. At each data point, we compute the false negative probability ($\FNP$) of the naive test statistic $\hat{T}_N$ at significance level $\alpha$ using the CLT from Theorem \ref{thm:ab_clt_cc}. We compare this to the $\FNP$ , obtained from the power calculation in \eqref{eq:ub_power}, for an unbiased estimator $\hat{\theta}_N$ that obeys \eqref{eq:unbiased_clt}.
Figure \ref{fig:power plot K N} shows the result: the $\FNP$ of $\hat{T}_N$ decays faster than the $\FNP$ for any unbiased test statistic, i.e., we obtain higher statistical power using a decision rule with $\hat{T}_N$.  This matches with the findings of Theorems \ref{thm: pos bias limit} and \ref{thm:cr_expo_growth}, which together suggest that the power of the naive test statistic $\hat{T}_N$ should be larger than $\hat{U}_N$.

As expected, these results confirm that the power of the naive test statistic $\hat{T}_N$ is {\em higher} than any test statistic based on an unbiased estimate of the treatment effect, with an associated unbiased variance estimator.  When combined with Theorem \ref{thm:aa_clt_cc}, we find that the decision maker is better off {\em not} developing a debiased estimator in the presence of interference, when interventions are sign-consistent: they earn the desired control over their false positive probability, and only stand to gain statistical power as a result.

We conclude by noting that in our analysis of statistical power, we have compared against an {\em unbiased} estimator.  In practice, perhaps a platform might be interested in partial debiasing, but not go to the extreme of a fully unbiased estimator, to avoid paying the variance penalty implied by Theorem \ref{thm:cr_expo_growth}.  For example, the recently introduced difference-in-Q's (DQ) estimator \cite{farias2022markovian} was shown to reduce variance significantly relative to unbiased estimators (e.g., the LSTD OPE estimator), with lower bias than the DM estimator.  However, in general, when interventions are sign-consistent, the DM estimator will continue to be larger in magnitude than such estimators (because it is more biased), and typically exhibits lower variance; indeed, this is see in the simulation results in Figure 4 of \cite{farias2022markovian}.  Thus we conjecture that even for such partially debiased estimation approaches, the naive decision-making approach will enjoy the same advantages described in the preceding paragraph.

%% file: Sections/sims.tex
\section{Sign-inconsistent treatments}
\label{sec:sims}

Our analysis thus far has shown that if treatments are sign-consistent, then even without debiasing, a platform is able to control false positive probability (Section ~\ref{sec:fpp}) while obtaining higher power than any debiased estimation approach (Section ~\ref{sec:power}).  In this section, we use numerics to investigate the behavior of false positive probability and statistical power when interventions are {\em sign-inconsistent.}

In general, if interventions are sign-inconsistent, it is possible for both quantities to behave arbitrarily worse using the na\"{i}ve decision-making approach based on the difference-in-means estimator, compared to the decision rule associated to an optimal unbiased estimator.  In particular, there exist examples where the false positive probability can become arbitrarily close to 1, instead of being controlled at the desired pre-specified level in the decision rule \eqref{eq:decision}; and there also exist examples where the statistical power remains bounded away from 1, regardless of how many samples are collected---in contrast to the performance of a decision rule that uses an unbiased estimator.

To illustrate these possibilities, we consider a natural class of interventions that are sign-inconsistent, where booking probabilities are {\em increased} in lower states (i.e., when many listings are available), and are {\em decreased} in larger states (i.e., when few listings are available).  These are natural interventions to consider from an operations standpoint.  For example, a ridesharing platform may be interested in understanding the impact on rides if prices are lowered relative to the status quo when ample driver supply is available, but raised relative to the status quo when driver supply is relatively tightly constrained.

Concretely, we construct two extreme examples of this form, to illustrate the potential worst-case consequences for false positive probability and statistical power, respectively.  In the first example, the true $\GTE = 0$, but $\ADE_a \neq 0$; in the second example, the true $\GTE \neq 0$, but $\ADE_a = 0$.  In both cases, sign-inconsistency of the treatment means that $\widehat{\DM}_N$ is asymptotically {\em biased} (since $\ADE_a \neq \GTE$); this bias is the source of the poor control of false positive probability (in the first example) and low statistical power (in the second example).

\begin{example}
\label{ex:ADE<0,GTE=0}
($\ADE_a < 0$ but $\GTE = 0$.) In the first example, we set $\bar\lambda = 1$, $\bar\tau = 1$, and $K = 100$.  We set control booking probabilities $\v{p}_0$ as follows:
\[ p_0(k)= 0.5 \text{ if } k < K, \text{ } p_0(K)=0. \]
(We note in passing that booking probabilities are not strictly decreasing, as in Assumption \ref{as:booking_prob_monotone}, but this is not essential; similar examples can be constructed even if booking probabilities are required to be strictly decreasing.)  We set treatment booking probabilities as follows:
\[ p_1(k)= \bar{p} \text{ if } k = 0,1; \text{ } p_1(k)=0.1 \text{ if } 2\leq k <K, \text{ } p_1(K)=0, \]
where $\bar{p} > 0.5$; we show how to choose $\bar{p}$ to ensure that $\GTE =0$.  In other words, when the system is nearly empty, the treatment raises booking probabilities relative to control.  Otherwise, the treatment lowers booking probabilities relative to control. 

To construct $\bar{p}$, we first note that it is straightforward to use the detailed-balance equations to check that a sufficient condition for $\GTE=0$ is:\footnote{Under detailed-balance with $\bar{\lambda}=\bar{\tau}=1,p_0(k)=0.5$, the global control equilibrium distribution is given by $\pi_0(k)=\frac{0.5^k}{1+0.5+\dots+0.5^K}$, so the global control booking probability is $\rho_0=\frac{0.5+0.5^2+\dots+0.5^K}{1+0.5+\dots+0.5^K}$. Similarly, the global treatment booking probability is given by $\rho_1=\frac{\bar{p}+\bar{p}^2+\bar{p}^20.1+\dots+\bar{p}^20.1^{K-2}}{1+\bar{p}+\bar{p}^2+\bar{p}^20.1+\dots+\bar{p}^20.1^{K-2}}$. We observe that $\rho_1=\rho_0$ if $1+\bar{p}+\bar{p}^2+\bar{p}^20.1+\dots+\bar{p}^20.1^{K-2}=1+0.5+\dots+0.5^K$.}
\[1+\bar{p}+\bar{p}^2\frac{1-0.1^{K-1}}{1-0.1}=\frac{1-0.5^{K+1}}{1-0.5}.\]
Letting $A_K=\frac{1-0.1^{K-1}}{0.9}$ and $B_K=\frac{1-0.5^{K+1}}{0.5}$, we see that $\bar{p}$ solves 
\[A_K \bar{p}^2+\bar{p}-(B_K-1)=0,\]
so in particular we can set
\[\bar{p}= \frac{-1 +\sqrt{1+4A_K(B_K-1)}}{2A_K}.\]
For $K=100$, we have $\bar{p} \approx 0.6$, and with these parameters and $a=1/2$ we have $\GTE=0$ yet $\ADE_a=-0.009<0$.\footnote{The detailed-balance equations with $\bar{\lambda}=\bar{\tau}=1,a=1/2$ imply that $\pi_a(k)=\frac{q_a(0)\dots q_a(k-1)}{1+q_a(0)+\dots+q_a(0)\dots q_a(K-1)}$. This implies $\ADE_a=\frac{(p_1(0)-p_0(0))+\sum_{k=1}^{K-1}(p_1(k)-p_0(k))q_a(0)\dots q_a(k-1)}{1+q_a(0)+\dots+q_a(0)\dots q_a(K-1)}$. For $p_0(k)=0.5,p_1(0)=p_1(1)=\bar{p}$ and $p_1(k)=0.1$ for $k>1$, we obtain $\ADE_a=-0.009<0$.}

We simulate a Bernoulli randomized experiment with $a = 1/2$ and $N \in [100,5000]$, with the remaining parameters specified in the previous paragraph. For each fixed set of parameters (ie. fixed $N$), we run $5\times 10^4$ trajectories of the experiment, and at each trajectory we form the test statistic \eqref{eq:tstat} using $\DM$ and $\widehat{\Var}$. We obtain the false positive probability for the naive test statistic by taking the average number of times $H_0$ is rejected under significance level $\alpha=0.05$ over the $5\times 10^4$ trajectories.

As argued above, the treatments constructed in Example \ref{ex:ADE<0,GTE=0} produce $\ADE_a = -0.009 < 0$ and $\GTE = 0$.  As a consequence, this is a treatment where the null hypothesis $H_0$ is satisfied.  (Note that $\GTE = 0$ despite the fact that the treatment and control chains are distinct; this is only possible because the treatment is sign-inconsistent, cf.~the discussion in Section ~\ref{sec:cc_error}.)  %
However, because $\ADE_a < 0$, the magnitude of the mean of the test statistic $\hat{T}_N$ (cf.~\eqref{eq:tstat}) grows without bound.  As a consequence, the false positive probability of the test statistic $\hat{T}_N$ will increase towards $1$ as $N \to \infty$, as suggested by Figure ~\ref{fig:nonmon fpp}.   Of course, a test statistic $\hat{U}_N$ using an unbiased estimator for $\GTE$ along with its true variance (cf.~\eqref{eq:ub_tstat}) should control the false positive probability correctly, as long as an appropriate central limit theorem holds.

\begin{figure}
    \centering
    \includegraphics[width=0.65\linewidth]{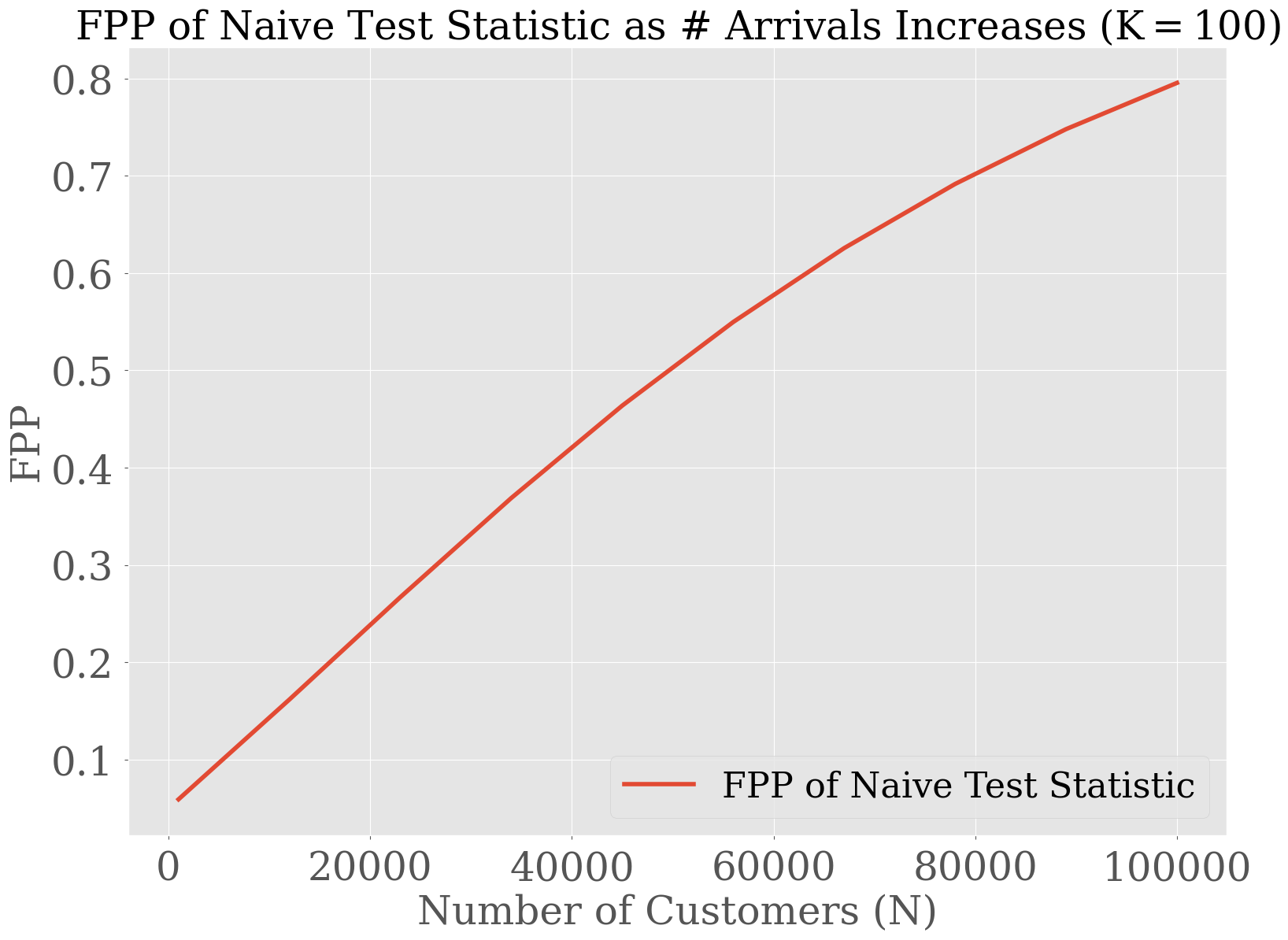}
    \caption{False positive probability of naive test statistic under a Bernoulli customer-randomized experiment with $a = 1/2$ and $N \in [10^3,10^5]$, with $K$, $\lambda$, $\tau(k)$, $\v{p}_0$, and $\v{p}_1$ specified in Example ~\ref{ex:ADE<0,GTE=0}. As $N$ grows, the false positive probability grows and will approach $1$. }
    \label{fig:nonmon fpp}
\end{figure}
\end{example}

\begin{example}
\label{ex:ADE=0}
($\ADE_a\approx 0$ but $\GTE > 0$.)
In the second example, we set $\bar\lambda = 1$, $\bar\tau = 1$, and $K = 30$.  We set control booking probabilities $\v{p}_0$ as follows:
\[ p_0(k)= 0.5 \text{ if } k < K;\ p_0(K)=0. \]
(Again note that booking probabilities are not strictly decreasing, as in Assumption \ref{as:booking_prob_monotone}, but again, this is not essential.)  We set treatment booking probabilities as follows:
\[ p_1(k)= 0.62 \text{ if } k = 0,1;\ p_1(k)=0.0745 \text{ if } 2\leq k <K; p_1(K)=0 . \]
Again, when the system is nearly empty, the treatment raises booking probabilities relative to control, and otherwise, the treatment lowers booking probabilities relative to control.  

We consider a Bernoulli randomized user-level experiment with $a = 1/2$ and $N \in [10^3,10^5]$, with the remaining parameters specified in the previous paragraph. For each fixed set of parameters (ie. fixed $N$) and significance level $\alpha=0.05$, we numerically compute the false positive probability of the test statistic \eqref{eq:tstat} using $\DM$ and $\widehat{\Var}$.

It is straightforward to verify using the detailed-balance conditions for the global treatment, global control, and experiment birth-death chains (cf.~\eqref{eq:generator}) that in this case $\GTE = 0.0087$.  As a consequence, this is a treatment where the null hypothesis $H_0$ is {\em not} satisfied.  On the other hand, $\GTE \gg \ADE_a \approx 0$; in fact, $\ADE_a = -7.0 \times 10^{-6}$.  As a consequence, the power under the decision rule \eqref{eq:decision} is significantly smaller than a test statistic $\hat{U}_N$ using an unbiased estimator for $\GTE$ that satisfies \eqref{eq:unbiased_clt}.  This effect is illustrated in Figure ~\ref{fig:nonmon power}.

\begin{figure}
    \centering
    \includegraphics[width=0.65\linewidth]{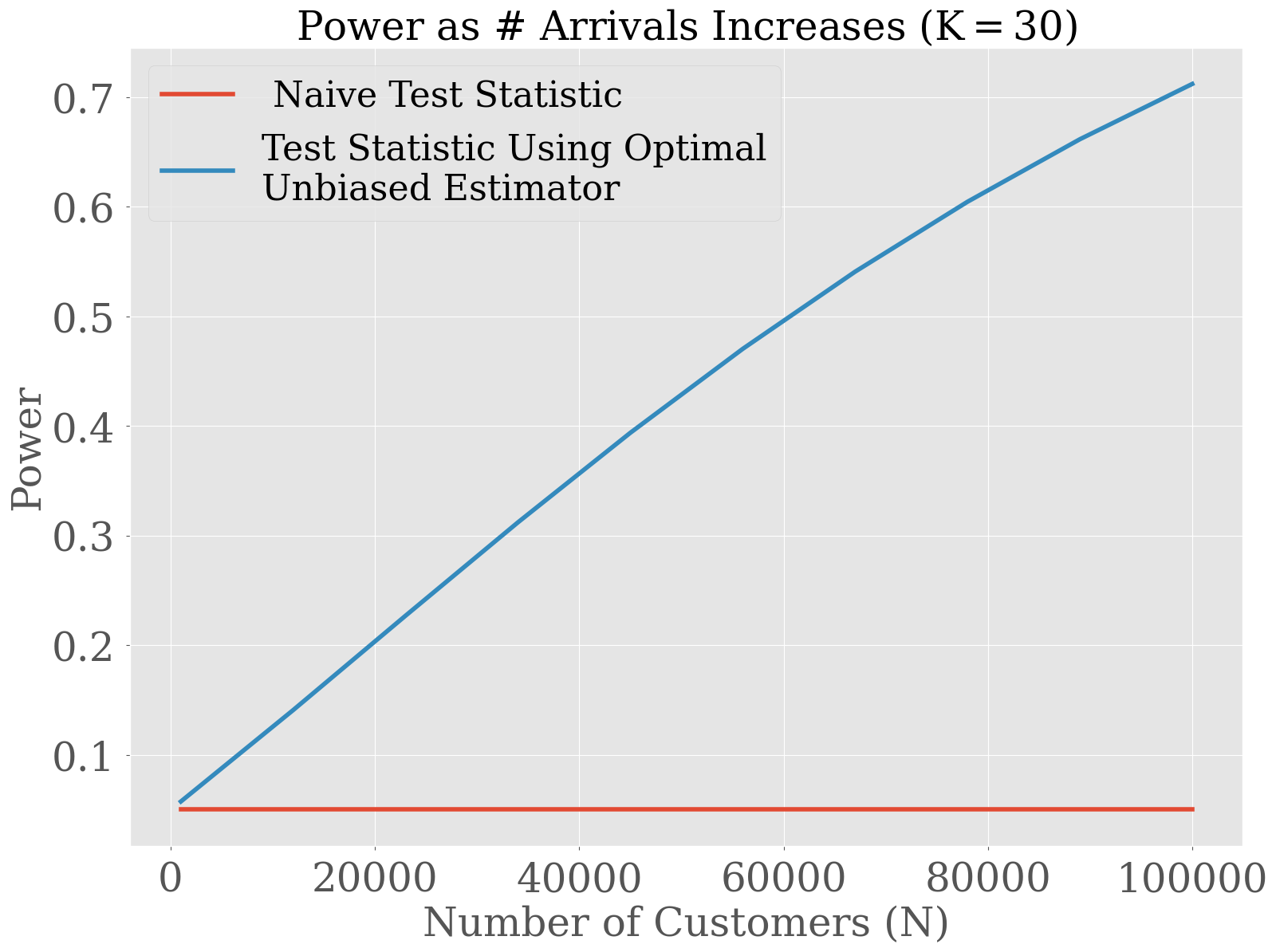}
    \caption{Power of naive test and unbiased test statistic under a Bernoulli randomized user-level experiment with $a=1/2$, $N \in [10^3,10^5]$,with $K$, $\bar\lambda$, $\bar\tau$, $\v{p}_0$, and $\v{p}_1$ specified in Example ~\ref{ex:ADE=0}. As $N$ grows, the power calculation in \eqref{eq:ub_power} (for any unbiased estimator that obeys \eqref{eq:unbiased_clt}) grows to $1$, while the power of the naive test statistic is much lower.}
    \label{fig:nonmon power}
\end{figure}
\end{example}

%% file: Sections/conclusion.tex
\section{Conclusion}
\label{sec:conclusion}

Using a benchmark Markov chain model for a two-sided platform, our results characterize the impact of interference on the false positive probability and statistical power when the experimenter uses na\"{i}ve estimation based on a t-test statistic.  We obtain the surprising finding that when treatments are sign-consistent in a Bernoulli randomized user-level experiment, the false positive probability is correctly controlled despite the presence of interference, and the statistical power is larger than that achieved by using an unbiased estimator.  In other words, in this setting the platform is actually better off {\em not} using a debiased estimator.  Despite these findings, as the numerics in Section \ref{sec:sims} suggest, if the treatment is sign-inconsistent then in the worst case, a debiased estimation method can offer significant improvements to control of false positive probability, and/or much higher statistical power, over the na\"{i}ve approach.  We highlight some significant remaining open directions below.

\paragraph{Listing heterogeneity.}  In our model in Section \ref{sec:cc}, customers are heterogeneous but listings are homogeneous.  It is natural to consider an extension of our model to a setting where listings are heterogeneous as well.  In such a setting, the appropriate Markov process representation requires a {\em multidimensional} state space: we must track the number of booked listings of {\em each} type.  Such a model is more complex to analyze, and in particular, the stochastic monotonicity structure we leverage (cf.~\cite{keilson1977monotone}) can no longer be applied.  One possible path forward might be to impose a customer choice model that governs substitution patterns {\em between} different types of listings.  For example, the model in \cite{johari2022experimental} assumes that arriving customers make booking choices according to a multinomial logit (MNL), where each customer type $\gamma$ has a valuation $v_\gamma(\delta)$ for listings of type $\delta$.  On one hand, if interventions are sign-consistent in this setting (e.g., if $v_\gamma(\delta)$ increases for all pairs $\gamma, \delta$), we might expect similar results as we found in Sections \ref{sec:fpp} and \ref{sec:power} in this paper using stochastic monotonicity arguments, at least in the mean field limit model of \cite{johari2022experimental}.  On the other hand, when treatment is sign-inconsistent across listing types, then we might see similar behavior as some of the examples in Section \ref{sec:sims}, e.g., if treatment increases the preference of customers for some types of listings while decreasing their interest in other types of listings.

\paragraph{Alternative experimental designs.} Our paper considers user-level Bernoulli randomized experiments, also referred to as customer-randomized (CR) experiments; it is natural to also consider whether similar results hold for listing-randomized (LR) experiments \cite{johari2022experimental}, where all arriving customers see a mix of treatment and control listings.  Given the generality of our results on A/A tests in Section \ref{subsec:aa_general}, it is reasonable to expect that with an appropriate definition of sign-consistent treatments for listing-randomized experiments, similar results as this paper could be obtained.  Investigation of the decision-making consequences of LR experiments, as well as other designs such as switchback experiments, remains an important direction for future work. (As one promising example, the recent paper \cite{ni2025decision} studies  decision-making from switchback experiments via robust optimization techniques.

\paragraph{Nonstationarity.}  In our model of an inventory-constrained platform, we have assumed that dynamics are {\em stationary}: the transition dynamics of the continuous time Markov process remain constant over time.  However, we conjecture that some of our insights can be extended to a setting where the environment may be nonstationary, with appropriate changes to the modeling framework.  As one example, suppose that (due to background nonstationarity) booking probabilities may depend on the customer; i.e., $p_{i,z}(k)$ is the booking probability of customer $i$ in treatment condition $z = 0,1$, if they arrive when $k$ listings are booked.  The estimand now needs to be defined {\em only within the experiment horizon}; in particular, let $\GTE_N$ be the difference in expected bookings in global treatment and global control, among the first $N$ arriving customers.  In this setting, our key observation is that the resulting platform dynamics still yield a stochastically monotone Markov process, cf.~\cite{keilson1977monotone}, allowing extension of our results that depend on stochastic monotonicity as well (Proposition \ref{pr:null_monotone}, and Theorems \ref{thm: pos bias DM} and \ref{thm: pos bias limit}).  Nevertheless, important issues remain challenging in this setting: notably, asymptotic control of the false positive probability as $N \to \infty$ would require additional assumptions on the nature of nonstationarity over an infinite horizon.  This remains an important direction for future investigation, particularly for practice.  (See, e.g., \cite{hu2022switchback}, \cite{li2023experimenting}, and \cite{johari2025estimation} for recent examples of papers on experimentation with explicit models of nonstationarity.)

\paragraph{Alternative objectives.}  More broadly, our paper has assumed a particular frequentist hypothesis testing decision-making pipeline, and in particular the use of this pipeline presumes the decision maker cares about false positive probability and statistical power.  In practice, there are many reasons this may not be the desired objective, in which case debiasing may be quite valuable.  The most obvious such case is when the platform cares about the actual value of $\GTE$ itself; this may be important if launching an intervention carries with it a significant cost, creating a tradeoff for the platform.  As an example, a ride-sharing platform may consider the launch of an incentive program for drivers.  Even if $\GTE > 0$, the platform may only be willing to launch if the cost of supporting such an incentive is not prohibitive relative to the true magnitude of $\GTE$.  In this case, debiasing is critical to understand the relative tradeoff between the treatment effect and the cost.  A similar situation arises if there are multiple target metrics of interest, and the decision maker faces tradeoffs between the impacts to these objectives.  For example, a platform testing an change in prices or fees may be interested not only in the change in bookings, but also the change in revenue or profit; recent work has studied the impact of naive experiments and estimation approaches on decision-making in such settings \cite{dhaouadi2023price}.

Indeed, in practice there are a wide range of potential objectives for a platform that uses A/B experimentation.  The primary lesson of our work is that the value of debiasing depends on both the desired inferential goal and decision objective, as well as the nature of the treatment itself.\footnote{In this vein, a recent paper similarly also argues that ``causal decision making'' and ``causal effect estimation'' are distinct goals, each appropriate in different settings \cite{fernandez2022causal}.} Our work serves provides a framework and essential insights for platforms to make choices regarding their inferential and decision-making pipeline in practice.

%% file: Sections/app_MC.tex
\section{Experiments on Markov chains}
\label{app:MC}

In this self-contained appendix we consider a setting where the treatment and control data generating processes are Markov chains, with rewards also potentially affected by treatment.  We provide a central limit theorem for the DM estimator in the setting of general A/B tests, and then specialize our results to A/A tests.  Throughout this appendix, $\to^p$ denotes convergence in probability; $\Rightarrow$ denotes convergence in distribution; and $\mc{N}(m,V)$ denotes a normal random variable with mean $m$ and variance $V$.  We conclude by showing that the capacity constrained platform model studied in Section \ref{sec:cc} is a special case of this setting.

\subsection{Experiments and estimation on general Markov chains}
\label{app: MC experiment}
We first present a general model of A/B tests when the data generating processes in treatment and control are Markov chains.  We prove a corresponding central limit theorem for $\DM_N$, and also characterize asymptotics of $\widehat{\Var}_N$.

\paragraph{State.}  We consider homogeneous Markov chains that evolve on a finite state space $S$.  We let $X_1, X_2, \ldots \in S$ denote the sequence of states visited by the chain.  

\paragraph{Bernoulli treatment assignment.}  As in the main text, we let $Z_1, Z_2, \ldots$ be i.i.d.~Bernoulli($a$) random variables, with $0 < a < 1$.  

\paragraph{Markov property and transition probabilities.}  At each time step $i$, we consider an order of events where, given the state $X_i$ and the treatment assignment $Z_i$, the next state $X_{i+1}$ is realized.  

Formally, suppose that $X_1 = x_1, \ldots, X_{i-1} = x_{i-1}, X_i = x$, and $Z_1 = z_1, \ldots, Z_{i-1} = z_{i-1}, Z_i = z$.  The Markov property asserts that:
\begin{multline*}
\P(X_{i+1} = x' | X_i = x, Z_i = z) = \\
\P(X_{i+1} = x'| X_1 = x_1, \ldots, X_{i-1} = x_{i-1}, X_i = x; Z_1 = z_1, \ldots, Z_{i-1} = z_{i-1}, Z_i = z).
\end{multline*}
We write $P(x,x' | z) = \P(X_{i+1} = x' | X_i = x, Z_i = z)$.  Note that this quantity is independent of $i$, so the transition probabilities are stationary in time; in other words, the Markov chain is {\em homogeneous}.  Note that since treatment assignments are i.i.d.~Bernoulli($a$), the sequence $\{X_i\}$ is {\em also} a discrete-time Markov chain; through an abuse of notation, we represent its transition matrix as $P(\cdot,\cdot|a)$, and we have:
\begin{align*}
 P(x,x'|a) &= a \P(X_{i+1} = x' | X_i = x, Z_i = 1) + (1 - a) \P(X_{i+1} = x' | X_i = x, Z_i = 0) \\
& = a P(x,x' |1) + (1-a) P(x,x'|0). 
 \end{align*}

We assume for $z = 0,1$ that the transition matrix $P(\cdot, \cdot | z)$ is irreducible; therefore for $0 < a < 1$, $P(\cdot,\cdot | a)$ is also irreducible.  Since treatment assignments are i.i.d.~Bernoulli($a$), if we let $W_i = (X_i, Z_i)$, then the Markov chain $\{W_i\}$ then has the following transition matrix $R$:
\[ R(x,z; x',z') = \P(X_{i+1} = x', Z_{i+1} = z' | X_i = x, Z_i = z) = a^{z'}(1-a)^{1-z'} P(x,x' | z). \]
It then follows that as long as $0 < a < 1$, the matrix $R$ is irreducible on the state space $S \times \{0,1\}$.

\paragraph{Invariant distribution.}  The transition matrix $R$ possesses a unique invariant distribution $\v{\nu}_a$.  This distribution can be written as follows:
\[ \nu_a(x,z) = a^z (1 - a)^{1-z} \phi_a(x), \]
where $\phi_a(x | z)$ satisfies:
\[ \phi_a(x) = \sum_{x'} \sum_{z' = 0,1} \nu_a(x', z') P(x',x | z'). \]
It is straightforward to verify that $\v{\phi}_a$ is the invariant distribution of the transition matrix $P(\cdot,\cdot|a)$.
We use the subscript $a$ to indicate the treatment fraction.  Note that in steady state, the distribution of the state under $\v{\nu}_a$ is independent of the treatment assignment; this reflects the fact that treatment assignments are i.i.d.

Throughout our analysis, we write $\E_{\v{\nu}_a}$ for expectations when the initial state $(X_1, Z_1)$ is sampled from the invariant distribution $\v{\nu}_a$.  Except in the context of such expectations, we make no assumptions otherwise on the initial state of the chain.

For completeness, for $a = 0, 1$ we define $\v{\phi}_0, \v{\phi}_1$ to be the invariant distributions of the chain with transition matrix $P(\cdot,\cdot|0)$, $P(\cdot,\cdot|1)$, respectively, i.e., for $a = 0,1$, $\phi_a(x) = \sum_{x'} \phi_a(x') P(x', x | a)$.  
These are the invariant distributions for {\em global control} ($a = 0$) and {\em global treatment} ($a = 1$).

\paragraph{Reward function.}  We fix a reward function $g(x, z)$ that depends on the state and treatment indicator, and we let $Y_i = g(X_i, Z_i)$ for each $i$.  Although in principle, rewards may be random given the state, note that we can let $g$ be the expected reward in this case, so our definition is without loss of generality.

\paragraph{Estimand.}  The estimand of interest is the steady state difference in reward rate between global treatment and global control, i.e., the {\em global treatment effect}:
\begin{equation}
 \GTE = \sum_x \phi_1(x) g(x,1) - \sum_x \phi_0(x) g(x,0). \label{eq:gte_MC}
\end{equation}

\paragraph{Difference-in-means (DM) estimator and variance estimator.}  Both $\DM_N$ and $\widehat{\Var}_N$ are defined identically to the main text, cf.~\eqref{eq:DM} and \eqref{eq:var_hat}.  As we consider only asymptotic results as $N \to \infty$ in our analysis in this section, we implicitly condition all our analysis on the (almost sure) event that $N_{1N} > 0$ and $N_{0N} > N$.

\paragraph{Central limit theorem.}  We now state a central limit theorem for the difference-in-means estimator.  To state the theorem, we require some additional notation.  First define:
\[ \ADE_a = \sum_x \phi_a(x) g(x,1) -  \sum_x \phi_a(x) g(x, 0). \]
The quantity $\ADE_a$ is the {\em average direct effect} \citep{hu2022average}.  Since, in general, this steady state distribution neither matches global treatment nor global control, in general $\ADE_a \neq \GTE$.

Next, define the following quantities for $z,z' \in \{0,1\}$:
\begin{align*}
V(z) &= \E_{\v{\nu}_a} \left[ \left(g(X_1,z) - \sum_x \phi_a(x) g(x,z)\right)^2 \Bigg| Z_1 = z\right] \\
&= \sum_x \phi_a(x ) \left( g(x,z) - \sum_{x'} \phi_a(x') g(x',z) \right)^2\\
&= \Var_{\v{\nu}_a}(Y_1 | Z_1 = z);\\
C_j(z,z') &= \E_{\v{\nu}_a} \left[ \left(g(X_1,z) - \sum_x \phi_a(x) g(x,z)\right) \left(g(X_j,z') - \sum_x \phi_a(x) g(x,z')\right) \Bigg| Z_1 = z, Z_j = z'\right] \\
&= \Cov_{\v{\nu}_a}(Y_1, Y_j | Z_1 = z, Z_j = z'),
\end{align*}
where we recall that $Y_i = g(X_i, Z_i)$.  These quantities capture the conditional variance and covariance of rewards, respectively, given treatment assignments.  Observe that although the expectations are with respect to the steady state distribution of $\{W_t\}$, i.e., $\v{\nu}_a$, in fact $V(z)$ can be reduced to an expectation against the steady state distribution of $\{X_t\}$, i.e., $\v{\phi}_a$.  This is because the expectation only involves $X_1$, which is independent of $Z_1$.  The same is true for $X_1$ in the definition of $C_j(z,z')$; however, $X_j$ in the definition of $C_j$ {\em does} depend on the initial treatment assignment $Z_1$.

\begin{theorem}
\label{thm:ab_clt}
Suppose $0 < a < 1$.  Then regardless of the initial distribution, $\DM_N \to^p \ADE_a$ as $N \to \infty$, and $\DM_N$ obeys the following central limit theorem as $N \to \infty$:  
\begin{equation}
\label{eq:dm_clt}
\sqrt{N} ( \DM_N - \ADE_a )  \Rightarrow \mc{N}\left(0, \tilde{\sigma}_a^2\right),
\end{equation}
where:
\begin{equation}
\tilde{\sigma}_a^2 = \left(\frac{1}{1-a}\right) V(0) + \left( \frac{1}{a} \right) V(1) + 2 \sum_{j > 1} C_j(0,0) + C_j(1,1) - C_j(0,1) - C_j(1,0), \label{eq:tilde_sigma}
\end{equation} 
with $\tilde{\sigma}_a^2 > 0$.
\end{theorem}

\begin{proof}
We first prove the convergence in probability of $\DM_N$.  Note that $N(1-a)/N_{0N} \to^p 1$, and $Na/N_{1N} \to^p 1$.  Thus the ergodic theorem for Markov chains establishes that:
\[ \b{Y}(0) \to^p \left(\frac{1}{1-a}\right) \E_{\v{\nu}_a}[ (1-Z_1) Y_1];\ \ \b{Y}(1) \to^p \left(\frac{1}{a}\right) \E_{\v{\nu}_a}[ Z_1 Y_1]. \]
Note that:
\[ \E_{\v{\nu}_a}[ (1 - Z_1) Y_1] = (1-a) \E_{\v{\nu}_a}[ g(X_1, 0) | Z_1 = 0] = (1-a) \sum_x \phi_a(x) g(x,0);\]
similarly, $\E_{\v{\nu}_a}[ Z_1 Y_1] = a \sum_x \phi_a(x)g(x,1)$.  We conclude that:
\begin{equation}
\b{Y}(0) \to^p \sum_x \phi_a(x ) g(x,0); \ \ \b{Y}(1) \to^p \sum_x \phi_a(x)g(x,1),\label{eq:conv_of_means}
\end{equation}
so that $\DM_N \to^p \ADE_a$ as $N \to \infty$.

For $z \in \{0,1\}$ and $x \in S$, define:
\begin{align*}
 h_1(x,z) &= \left(\frac{z}{a}\right) \left(g(x,1) - \sum_{x'} \phi_a(x') g(x',1)\right); \\
 h_0(x,z) &= \left(\frac{1-z}{1-a}\right) \left(g(x,0) - \sum_{x'} \phi_a(x') g(x',0)\right).
\end{align*}
Observe that $h_1(x,0) = h_0(x,1) = 0$.  Further, observe that:
\[ \E_{\v{\nu}_a}[h_1(W_1)] = a \E_{\v{\nu}_a}[h_1(X_1, 1) | Z_1 = 1]
 =  \E_{\v{\nu}_a}[g(X_1, 1) | Z_1 = 1] - \sum_x \phi_a(x) g(x,1) = 0; \]
 similarly, $\E_{\v{\nu}_a}[h_0(W_1)] = 0$.

For fixed, arbitrary $t_0, t_1 \in \reals$, define $h(z,x) = t_0 h_0(x,z) + t_1 h_1(x,z)$.  Since $\E_{\v{\nu}_a}[h_1(W_1)] = \E_{\v{\nu}_a}[h_0(W_1)] = 0$, we have 
$\E_{\v{\nu}_a}[h(W_1)] = 0$ as well.

The standard central limit theorem for Markov chains (see, e.g., \cite{jones_markov_2004}, Theorem 1) thus implies that:
\[ \frac{1}{\sqrt{N}} \sum_{i = 1}^N h(W_i) \Rightarrow \mc{N}(0, \gamma(t_0, t_1)^2), \]
where:
\[ \gamma(t_0, t_1)^2 = \E_{\v{\nu}_a}[h(W_1)^2] + 2 \sum_{j > 1} \E_{\v{\nu}_a}[h(W_1)h(W_j)].\]

Straightforward algebra gives the following identities:
\begin{align*}
\E_{\v{\nu}_a}[h(W_1)^2] & = t_0^2 \E_{\v{\nu}_a}[ h_0(W_1)^2] + t_1^2 \E_{\v{\nu}_a}[ h_1(W_1)^2] \\
&\ \ + 2t_0 t_1 \E_{\v{\nu}_a}[ h_0(W_1) h_1(W_1)];\\
\E_{\v{\nu}_a}[h(W_1)h(W_j)] & = t_0^2 \E_{\v{\nu}_a}[ h_0(W_1) h_0(W_j) ] + t_1^2 \E_{\v{\nu}_a}[ h_1(W_1) h_1(W_j)] \\
&\ \ + t_0t_1 \E_{\v{\nu}_a}[ h_0(W_1)h_1(W_j) + h_1(W_1)h_0(W_j)].
\end{align*}

Observe that for any realization of $W_1 = (X_1, Z_1)$, we have $h_0(W_1) h_1(W_1) = 0$, since either $Z_1 = 0$ or $1 - Z_1 = 0$.  Further, observe that by conditioning on the realization of $Z_1$, we have:
\begin{align*}
\E_{\v{\nu}_a}[ h_0(W_1)^2] &= \E_{\v{\nu}_a}\left[ \left(\frac{1-Z_1}{1-a}\right)^2 \left(g(X_1,0) - \sum_x \phi_a(x) g(x,0)\right)^2 \right] \\
&= \left( \frac{1}{1-a} \right) \E_{\v{\nu}_a} \left[ \left(g(X_1,0) - \sum_x \phi_a(x) g(x,0)\right)^2 \Bigg| Z_1 = 0\right] \\
&=  \left( \frac{1}{1-a} \right) V(0).
\end{align*}
Similarly, we have:
\begin{align*}
\E_{\v{\nu}_a}[ h_0(W_1) h_0(W_j)] &= \E_{\v{\nu}_a}\Bigg[ \left(\frac{(1-Z_1)(1-Z_j)}{(1-a)^2}\right) \\
& \qquad \qquad \cdot \left(g(X_1,0) - \sum_x \phi_a(x) g(x,0)\right)  \left(g(X_j,0) - \sum_x \phi_a(x) g(x,0)\right)\Bigg] \\
&= \E_{\v{\nu}_a} \left[ \left(g(X_1,0) - \sum_x \phi_a(x) g(x,0)\right) \left(g(X_j,0) - \sum_x \phi_a(x) g(x,0)\right) \Bigg| Z_1 = 0, Z_j = 0\right] \\
&=  C_j(0,0).
\end{align*}

Using similar arguments, we arrive at the following identities:
\begin{align*}
\E_{\v{\nu}_a}[h(W_1)^2] &= \left( \frac{t_0^2}{1-a}\right)  V(0) + \left( \frac{t_1^2}{a} \right) V(1) ;\\
\E_{\v{\nu}_a}[h(W_1)h(W_j)] &= t_0^2 C_j(0,0) + t_1^2 C_j(1,1) + t_0 t_1(C_j(0,1) + C_j(1,0)).
\end{align*}

Now let $(U_0, U_1)$ be a jointly normally distributed random pair, with mean zero, and covariance matrix:
\[ \v{\Sigma} = \left( \begin{array}{cc} 
  V(0)/(1-a) + 2 \sum_{j > 1} C_j(0,0) & \sum_{j > 1} C_j(0,1) + C_j(1,0) \\
 \sum_{j > 1} C_j(0,1) + C_j(1,0) &  V(1)/a + 2 \sum_{j > 1} C_j(1,1)
\end{array} \right). \]
We have shown that:
\[ t_0 \left( \frac{1}{\sqrt{N}} \sum_{i = 1}^N h_0(W_i)\right) + t_1 \left( \frac{1}{\sqrt{N}} \sum_{i = 1}^N h_1(W_i) \right) \Rightarrow t_0 U_0 + t_1 U_1.\]
Using the Cram\'{e}r-Wold device, we conclude that: 
\[ \left( \frac{1}{\sqrt{N}} \sum_{i = 1}^N h_0(W_i), \frac{1}{\sqrt{N}} \sum_{i = 1}^N h_1(W_i) \right) \Rightarrow \mc{N}(\v{0}, \v{\Sigma}).\]

Define:
\[ \v{\Gamma}_N = \left( \begin{array}{c} 
\frac{(N(1-a))}{N_{0N}} \cdot \frac{1}{N} \sum_{i = 1}^N   h_0(W_i)  \\
 \frac{Na}{N_{1N}} \cdot \frac{1}{N} \sum_{i = 1}^N h_1(W_i)
 \end{array} \right)
 = \left( \begin{array}{c} 
\frac{1}{N_{0N}} \sum_{i = 1}^N (1 - Z_i) \left(g(X_i, 0) - \sum_{x'} \phi_a(x') g(x',0) \right)    \\
 \frac{1}{N_{1N}} \sum_{i = 1}^N Z_i \left(g(X_i, 1) - \sum_{x'} \phi_a(x') g(x',1) \right) 
\end{array} \right). \]
Since $N(1-a)/N_{0N} \to^p 1$, and $Na/N_{1N} \to^p 1$, by Slutsky's theorem we have:
\[ \sqrt{N} \v{\Gamma}_N \Rightarrow \mc{N}(\v{0}, \v{\Sigma}).  \]

Now observe that: 
\begin{align*}
(-1, 1)^\top \v{\Gamma}_N &= \frac{1}{N_{1N}} \sum_{i = 1}^N Z_i g(X_i, 1) - \frac{1}{N_{0N}} \sum_{i = 1}^N (1 - Z_i) g(X_i, 0) \\
&\ \ - \left( \sum_{x'} \phi_a(x') g(x',1) - \sum_{x'} \phi_a(x') g(x',0)  \right) \\
&= \DM_N - \ADE_a,
\end{align*}
since $Z_i g(X_i, 1) = Z_i Y_i$ and $(1 - Z_i) g(X_i, 0) = (1- Z_i) Y_i$.  Straightforward algebra shows that $\tilde{\sigma}^2 = (-1, 1)^\top \Sigma (-1, 1)$.  Thus we have shown \eqref{eq:dm_clt}.
\end{proof}

\paragraph{Limit of variance estimator.}  The following theorem gives the scaled limit of $\widehat{\Var}$.

\begin{theorem}
\label{thm:var_est}
Suppose $0 < a < 1$.  Then $\widehat{\Var}_N$ satisfies
\begin{equation}
N \widehat{\Var}_N \to^p \left(\frac{1}{a}\right)V(1) + \left(\frac{1}{1-a}\right) V(0).
\end{equation}
\end{theorem}

Note that in comparison to the true variance $\tilde{\sigma}^2$, the estimator misses all the covariance terms.  Whether this is an overestimate or underestimate of the true variance depends on whether the within-group covariances are stronger are weaker than the across-group covariances, cf.~\eqref{eq:tilde_sigma}.

\begin{proof}
We observe using a standard calculation that:
\[ \sum_{i = 1}^N Z_i ( Y_i - \b{Y}(1))^2 = \sum_{i = 1}^N Z_i Y_i^2  - N_{1N} \b{Y}(1)^2. \]
From \eqref{eq:conv_of_means}, we know that $\b{Y}(1) \to^p \sum_x \phi_a(x) g(x,1).$  We know $N_{1N}/(Na) \to^p 1$ as $N \to \infty$, so by the ergodic theorem, we have:
\[ \frac{1}{N_{1N}} \sum_{i = 1}^N Z_i Y_i^2 \to^p \frac{1}{a} \E_{\v{\nu}_a}[ Z_1 Y_1^2]. \]
We have:
\[ \E_{\v{\nu}_a}[Z_1 Y_1^2] = a \E_{\v{\nu}_a}[ g(X_1, 1)^2| Z_1 = 1] = a \sum_x \phi_a(x) g(x,1)^2.\]
Finally, note that $N/(N_{1N} - 1) \to^p 1/a$ as $N \to \infty$ as well.  Combining all these facts, we conclude that:
\[ \frac{N}{(N_{1N} - 1)N_{1N}} \sum_{i = 1}^N Z_i ( Y_i - \b{Y}(1))^2 \to^p \frac{1}{a} V(1).\]
An analogous calculation follows for the second term of $\widehat{\Var}_N$, establishing the claimed convergence of $N \widehat{\Var}_N$.
\end{proof}

Note that the preceding result implies the t-test statistic centered at $\ADE_N$ given by  $(\DM_N - \ADE_N)/\sqrt{\widehat{\Var}_N}$ converges in distribution to a normal random variable with zero mean, and variance:
\[ \frac{ V(0)/(1-a) + V(1)/a)}{V(0)/(1-a) + V(1)/a + 2 \sum_{j > 1} C_j(0,0) + C_j(1,1) - C_j(0,1) - C_j(1,0).}\]

\subsection{A/A experiments of Markov chains}
\label{app:MC_aa}

In this section, we apply our preceding results to {\em A/A experiments} in a Markov chain setting.  These are experiments where both the treatment and control chains are identical, the treatment assignment is completely independent of the Markov chain evolution; and where the rewards do not depend on treatment assignment.  See also Section \ref{subsec:aa_general} for further discussion of A/A experiments.

Formally, we suppose that that there is a single irreducible transition matrix $P$ such that $P(x,x'|0) = P(x,x'|1) = P$; let $\v{\phi}$ be the invariant distribution of this matrix.  This makes the evolution of $\{X_i\}$ independent of the realization of the treatment assignments $\{Z_i\}$.  It follows that for all $a$, $0 \leq a \leq 1$, we have $\v{\phi}_a = \v{\phi}$; and thus:
\[ \nu_a(x,z) = a^z (1 - a)^{1-z} \phi(x). \]
We also assume that there is a single reward function $g$ such that $g(x,0) = g(x,1) = g(x)$.  Note that the resulting reward sequence $\{ Y_i \}$ with $Y_i = g(X_i)$ is thus also independent of $\{Z_i\}$.

It follows from these definitions that for an A/A test, $\GTE = \ADE_a = 0$.  We note also that with these definitions, we have:
\begin{equation}
 V = V(0) = V(1) = \sum_x \phi(x) \left( g(x) - \sum_{x'} \phi(x') g(x')\right)^2. \label{eq:V}
\end{equation}
Further, since all treatment assignments are independent of the states $(X_1, X_j)$, we have:
\begin{align*}
 C_j &= C_j(0,0) = C_j(1,1) = C_j(0,1) = C_j(1,0)\\
&= \E_{\v{\phi}}\left[ \left( g(X_1) - \sum_x \phi(x)  g(x) \right) \left( g(X_j) - \sum_x \phi(x)  g(x) \right)\right],
\end{align*}
where the notation $\E_{\v{\phi}}$ denotes that we take expectations with respect to the chain where the initial state $X_1$ is sampled from $\v{\phi}$.  Combining Theorems \ref{thm:ab_clt} and \ref{thm:var_est} with these observations, we have the following result.

\begin{theorem}
\label{thm:aa_clt}
Suppose $0 < a < 1$.  For an A/A experiment, $\DM_N$ obeys the following central limit theorem as $N \to \infty$:  
\begin{equation}
\sqrt{N} \DM_N \Rightarrow \mc{N}\left(0, \left(\frac{1}{a} + \frac{1}{1-a}\right) V\right)
\end{equation}
where $V$ is defined as in \eqref{eq:V}.  In addition, $\widehat{\Var}_N$ satisfies
\begin{equation}
N \widehat{\Var}_N \to^p \left(\frac{1}{a} + \frac{1}{1-a}\right) V.
\end{equation}

In particular, the t-test statistic converges in distribution to a standard normal random variable:
\[ \frac{\DM_N}{\sqrt{\widehat{\Var}_N}} \Rightarrow \mc{N}(0,1).\]
\end{theorem}

This result, when applied to general A/A tests in a Markov chain environment, can be used to demonstrate control of the false positive probability.  In particular, suppose for each $N$ that the null hypothesis $\tilde{H}_0$ is defined as in \eqref{eq:aa_H0}, i.e., the null hypothesis assumes that the observations $Y_1, \ldots, Y_N$ are independent of $Z_1, \ldots, Z_N$.  Further suppose the test statistic $\hat{T}_N$ is defined as in \eqref{eq:tstat}, and the decision rule \eqref{eq:decision} with the null hypothesis is $\tilde{H}_0$ is followed.

The corresponding false positive probability is the chance that $\tilde{H}_0$ is mistakenly rejected when it is true.   Formally, this is:
\begin{equation}
\label{eq:FPP_aa} 
\FPP_N = \P(\hat{T}_N > \Phi_{\alpha/2} | \tilde{H}_0),
\end{equation}
where the distribution of observations comes from the Markov chain and rewards defined by $P$ and $g$ respectively, with a fixed initial distribution.  Note that $\FPP_N$ in \eqref{eq:FPP_aa} is well defined: Given that $Y_1, \ldots, Y_N$ are generated in this way, and $Z_1, \ldots, Z_N$ are i.i.d., as well as independent of $Y_1, \ldots, Y_N$ under $\tilde{H}_0$, the distribution of $\hat{T}_N$ is completely determined.

Theorem \ref{thm:aa_clt} then implies the following corollary.
\begin{corollary}
\label{cor:aa_fpp_limit}
Suppose $0 < a < 1$.  Then for an A/A experiment, as $N \to \infty$, $\FPP_N \to \alpha$, where $\FPP_N$ is defined as in \eqref{eq:FPP_aa}.
\end{corollary}

The preceding result shows that for general A/A experiments in Markov chain environments, the false positive probability of the decision rule \eqref{eq:decision} is (asymptotically) controlled at level $\alpha$, as intended.

\subsection{Cram\'{e}r-Rao bound for A/B experiments of Markov chains}
\label{app:CR}

In this section we present a result of \cite{farias2022markovian}, which uses the multivariate Cram\'{e}r-Rao bound to provide a lower bound on the variance of any unbiased $\GTE$ estimator for an A/B experiment between two Markov chains.  The same paper provides a least squares temporal difference (LSTD)-based unbiased estimator that matches this lower bound.  \cite{glynn2020adaptive} present a nonparametric maximum likelihood estimator which also acheives the Cram\'{e}r-Rao lower bound.

As in Section \ref{sec: asymptotic power}, in the general Markov chain setting we consider here, an estimator $\hat{\theta}_N$ is a real-valued functional of the observations $Y_1, \ldots, Y_N$ generated from a Bernoulli randomized experiment.  We say that $\hat{\theta}_N$ is an {\em unbiased estimator} if for all feasible values of system parameters (i.e., $S$, $P$, $g$) and the design parameter $a$, there holds $\E[\hat{\theta}_N] = \GTE$, where $\GTE$ is defined as in \eqref{eq:gte_MC}.

\begin{theorem}[\cite{farias2022markovian}]
\label{thm: farias}
   Let $\hat{\theta}_N$ be an unbiased estimator for $\GTE$ for all feasible values of system parameters (i.e., $S$, $P$, $g$) and the design parameter $a$. Then $\Var(\hat{\theta}_N)$ satisfies 
\begin{align}
    N \Var(\hat{\theta}_N) &\geq \frac{1}{a}\sum_{x,x'}\frac{\phi_1(x)^2}{\phi_a(x)}P(x,x'|1)(v(x'|1)-v(x|1)+g(x,1) - \rho_1)^2 \notag\\
    &+\frac{1}{1-a}\sum_{x,x'}\frac{\phi_0(x)^2}{\phi_a(x)}P(x,x'|0)(v(x'|0)-v(x|0)+g(x,0)) - \rho_0)^2:= \sigma^2_{\UB}. \label{eq:cramer_rao_bound}
\end{align}
where for $z = 0,1$:
\[ \rho_z = \sum_x \phi_z(x) g(x,z) \]
and:
\begin{equation}
\label{eq:valuefn}
 v(x|z) = \sum_{t = 1}^\infty \E\left[ g(X_t, Z_t) - \rho_z\ \Bigg|\ X_1 = x, Z_t = z\ \text{for all}\ t \right].
\end{equation}
\end{theorem}
In the preceding statement, $v(x|z)$ is the {\em value function} of either the global treatment condition (if $z = 1$) or the global control condition (if $z = 0$); the quantity $v(x|z)$ measures the cumulative difference between expected rewards and the steady state reward per time step.

We remark here that it is well known that $v$ solves {\em Poisson's equation} for the reward function $g$:
\[ v(x|z) = g(x,z) - \rho_z + \sum_{x',y'} P(x,,x'|z) v(x'|z). \]
(See, e.g., \cite{asmussen2003applied}.) The solution to Poisson's equation is uniquely specified only up to an additive constant.  If we let $\v{\Phi}_z$ be the stochastic matrix with $\v{\phi}_z$ on every row, then one solution to Poisson's equation is obtained by:
\[ \v{v}_z = ( \v{I} - \v{P}_z + \v{\Phi}_z)^{-1} (\v{g} - \rho_z \v{e}), \]
where $\v{I}$ is the identity matrix; $\v{P}_z$ is the matrix with entries $P(\cdot, \cdot|z)$; $\v{e}$ is the vector of $1$'s; and $\v{v}_z$ is the vector with entries $v(\cdot|z)$.  This particular solution has the feature that $\sum_x \phi_z(x) v(x|z) = 0$, which can be shown to hold for the definition in \eqref{eq:valuefn}. The matrix $( \v{I} - \v{P}_z + \v{\Phi}_z)^{-1}$ is the {\em fundamental matrix} of the Markov chain.  We use this matrix formulation to compute the Cram\'{e}r-Rao lower bound numerically in the main text.

\subsection{Application to a capacity-constrained platform}
\label{app:MC_cc}

In this section we show that the model in Section \ref{sec:cc} is a special case of the general Markov chain setting presented here, and can be analyzed using the same limit theorems.  

To make the mapping, we must define appropriate discrete time Markov chains corresponding to global treatment and global control, and a corresponding reward function $g$.  We proceed as follows.  Following Section \ref{sec:cc}, let $X_t$ is the continuous time Markov chain of booked listings under global treatment (i.e., $a = 1$).  Let $T(i)$ be the arrival time of the $i$'th customer, and let $X_{T(i)}^-$ denote the state just {\em before} the arrival of the $i$'th customer.  Let $S_i = (X_{T(i)}^-, Y_i)$.  Note that $S_i$ is a stationary discrete-time Markov chain on the finite state space $\{0, \ldots, K\} \times \{0,1\}$: given $S_i$, the next state $X_{T(i+1)}^-$ is independent of the history prior to customer $i$, and of course given $X_{T(i+1)}^-$, the booking outcome $Y_{i+1}$ of customer $i+1$ is also independent of the history prior customer $i$.  

We let $P(k,y,k',y'| 1)$ be the transition matrix of this discrete-time Markov chain corresponding to global treatment.  Note that since Poisson arrivals see time averages (PASTA), the steady state distribution of $P(\cdot | 1)$ is:
\[ \phi_1(k,y) = \pi_1(k) p_1(k)^y ( 1 -p_1(k))^{1-y}. \]
We can similarly define $P(k,y,k',y' | 0)$ for the global control chain, with an analogous steady state distribution $\phi_0(k,y)$, and $P(k,y,k',y' | a)$ for the experimental chain, with an analogous steady state distribution $\phi_a(k,y)$.

Finally, for each state $k$, $y = 0,1$, and $z = 0,1$, define the reward function $g(k,y,z) = y$; this is the booking outcome.  Note the difference in long run average reward between the treatment and control chains is:
\[ \sum_k \sum_{y = 0,1} \phi_1(k,y) g(k,y) - \sum_k\sum_{y = 0,1} \phi_0(k,y) g(k,y) = \sum_k \pi_1(k) p_1(k) - \sum_k \pi_0(k) p_0(k) = \GTE, \]
i.e., exactly the global treatment effect defined in \eqref{eq:gte_cc}.  Thus the estimand in Section \ref{sec:cc} is exactly the estimand in \eqref{eq:gte_MC}, i.e., the difference in long run average reward between the treatment and control discrete time chains with the primitives $P(\cdot | 1)$, $P(\cdot | 0)$, and $g$ defined as in this section.

Thus Theorems \ref{thm:ab_clt} and \ref{thm:var_est} directly apply in the setting of Section \ref{sec:cc}, with the primitives defined as in this section.  Applications of these results with the specific primitives defined above yields Theorems \ref{thm:ab_clt_cc} and \ref{thm:var_est_cc} in the main text.  We remark here that it is important to carefully translate the expectations in the variances and covariances that appear in Theorem \ref{thm:ab_clt}.  These are expectations evaluated with the initial distribution of the {\em discrete-time} chain initialized to $\v{\pi}_1$ or $\v{\pi}_0$ for global treatment and global control, respectively.  In the continuous time Markov chain of Section \ref{sec:cc}, this corresponds to the requirement that the distribution of the state {\em just prior to the arrival of the first customer} is the steady state distribution $\v{\pi}_1$ or $\v{\pi}_0$.

Observe that if the null hypothesis $H_0 : \GTE = 0$ holds {\em and} the treatment is sign-consistent, then by Proposition \ref{pr:null_monotone} we know that $p_0 = p_1$, i.e., we are in the setting of an A/A experiment: both treatment and control Markov chains are identical.  Theorem \ref{thm:aa_clt_cc} in the main text then follows by applying Theorem \ref{thm:aa_clt} and Corollary \ref{cor:aa_fpp_limit} to the resulting A/A experiment.

Application of Theorem \ref{thm: farias} yields the Cram\'{e}r-Rao bound we leverage in Section \ref{sec: asymptotic power}, and in particular \eqref{eq:ub_tstat_upper_bound}.  In this setting, the Cram\'{e}r-Rao bound in \eqref{eq:cramer_rao_bound} becomes:
\begin{align*}
N \Var(\hat{\theta}_N) &\geq \left( \frac{1}{a} \right) \sum_{k,y} \frac{\phi_1(k,y)^2}{\phi_a(k,y)}
\sum_{k',y'} P(k,y,k',y'|1) \left(v(k',y'|1) - v(k,y|1) + g_1(k,y) - \rho_1\right)^2  \\
&\quad + \left( \frac{1}{1-a} \right) \sum_{k,y} \frac{\phi_0(k,y)^2}{\phi_a(k,y)}
\left( \sum_{k',y'} P(k,y,k',y'|0) \left( v(k',y'|0) -  v(k,y|0) + g_0(k,y )- \rho_0 \right)^2 \right) \\
& \triangleq \sigma^2_\UB,
\end{align*}
where $v(k,y|z)$ is the value function for $z = 0,1$ associated to the reward function $g$ defined above, with $\rho_z = \sum_k \pi_z(k) p_z(k)$.

We conclude this section with the proof of Theorem \ref{thm:cr_expo_growth}, which establishes exponential growth of $\sigma^2_{\UB}$ generically in the size of the state space.

\begin{proof}[Proof of Theorem \ref{thm:cr_expo_growth}]
We develop a coarse lower bound by only focusing on the control term (i.e., the second term) in the Cram\'{e}r-Rao lower bound; this suffices to establish the theorem.  Define $\b{p}_a(s) = (1-a) \b{p}_0(s) + a \b{p}_1(s)$.  Let $s_z^*$ be the unique solution to:
\[ \b{\lambda} \b{p}_z(s) = \b{\tau}(s) \]
for $z = 0, a$.  It is easy to check that $0 < s_0^* < s_a^*$, given the monotonicity properties of $\b{p}_0, \b{p}_1, \b{\tau}$, and the fact that treatment is sign-consistent.  Choose $\epsilon > 0$ such that $s_0^* - 3\epsilon > 0$ and $s_0^* + 7\epsilon < s_a^*$.  We assume $K$ is large enough that $\epsilon > 1/K$.

Let $k_0^+ = \lfloor (s_0^* + 3 \epsilon) K \rfloor$, and let $k_0' = \lfloor (s_0^* + 2 \epsilon) K \rfloor$.  For each $K$, consider the set of all $k > k_0'$.  By detailed balance, for all such $k$, we have:
\[ \lambda \b{p}_0\left( \frac{k-1}{K} \right) \pi_0^{(K)}(k-1) = \b{\tau}\left( \frac{k}{K} \right) \pi_0^{(K)}(k).\]
Note that for all $k > k_0'$, we have $\b{p}_0((k-1)/K) < \b{p}_0(k_0^+ + 2\epsilon - 1/K) < \b{p}_0(s_0^* + \epsilon)$ by monotonicity; and we also have $\b{\tau}(k/K) \geq \b{\tau}(s_0^* + 2\epsilon) \geq \b{\tau}(s_0^* + \epsilon)$ by monotonicity.  Let $c_1 = \lambda \b{p}_0(s_0^* + \epsilon)/\b{\tau}(s_0^*+\epsilon) < 1$.  By iterating up from $k_0'$, we conclude that for all $k \geq k_0^+$:
\[ \pi_0^{(K)}(k) < \pi_0^{(K)}(k_0') \cdot c_1^{k - k_0'} < c_1^{k_0^+ - k_0'}. \]
Since this is true for all $k \geq k_0^+$, we conclude that:
\begin{equation}
\label{eq:treatmentbound1}
\sum_{k = k_0^+}^{K} \pi_0^{(K)}(k)^2 < c_2 c_1^{2\lfloor \epsilon K \rfloor}
\end{equation}
for a constant $c_2 > 0$ and all sufficiently large $K$.  To see this, observe that there are less than $K$ terms in the first sum, and each term is bounded above by $c_1^{2(k_0^+ - k_0')}$.  Given the definitions of $k_0', k_0^+$ and the fact that $c_1 < 1$, we can choose an appropriate constant $c_2$ for the stated bound.

Similarly, let $k_0^- = \lceil (s_0^* - 3 \epsilon) K \rceil$.  We can use an analogous argument to conclude that there exist constants $c_3 < 1$ and $c_4 > 0$ such that:
\begin{equation}
\label{eq:treatmentbound2}
\sum_{k = 0}^{k_0^-} \pi_0^{(K)}(k)^2 < c_4 c_3^{2\lfloor \epsilon K \rfloor}
\end{equation}
We omit the details.  Taken together, note that we have:
\[ \sum_{k_0^- < k < k_0^+} \pi_0^{(K)}(k)^2 > 1 - c_6 c_5^{2\lfloor \epsilon K \rfloor} \]
for some constants $c_6 > 0$ and $c_5 < 1$.

With $k_a^- = \lfloor (s_a^* - 3\epsilon) K \rfloor$, analogous arguments can also be used to show the following inequality for a constant $c_7 < 1$ and all $k < k^-_a$:
\begin{equation}
\label{eq:experimentbound}
\pi_a^{(K)}(k) < c_7^{\lfloor \epsilon K \rfloor}.
\end{equation}
We omit the details.  Since for any $k < k_0^+$, we also have $k < k_a^-$, it follows that:
\[ \sum_{k_0^- < k < k_0^+} \frac{\pi_0^{(K)}(k)^2}{\pi_a^{(K)}(k)} > c_8 \gamma^K, \]
for some constants $c_8 > 0$ and $\gamma > 1$.

To conclude the proof, we lower bound the quadratic terms in the control term of $\sigma^2_{\UB}$.  Fix $K$ sufficiently large so that the preceding bounds hold, and let $P^{(K)}(\cdot|z)$ and $v^{(K)}(\cdot|z)$ be the transition probabilities and value function for the $K$'th control system as defined in the theorem.  Define $\rho_z^{(K)}$ to be the steady state booking probability in the $K$'th system:
\[ \rho_z^{(K)} = \sum_k \pi_z^{(K)}(k) \b{p}_z\left( \frac{k}{K}\right),\quad z = 0,1.\]

Throughout the remainder of the proof, we suppress the dependence on $K$ except where necessary for clarity (e.g., writing $\rho_z$ instead of $\rho_z^{(K)}$.  We have the following lemma; the proof is tedious but straightforward algebra using the definition of $P(\cdot | z)$, and deferred after the proof of this theorem.

\begin{lemma}
\label{lem:valfn_deltas}
Fix $K$.  The value function $v(\cdot|z)$ satisfies:
\begin{align}
v(K,0|z) - v(K-1,0|z) &= -\frac{\lambda \rho_z}{\tau(K)}; \label{eq:valfn_deltas1} \\
v(k+1,0|z) - v(k,0|z) &= \frac{\rho_z - p_z(k)}{p_z(k)} + \frac{\tau(k)}{\lambda p_z(k)} (v(k,0|z) - v(k-1,0|z)),\ \ \text{for}\ 0 < k < K; \label{eq:valfn_deltas2}\\
v(1,0|z) - v(0,0|z) &= \frac{\rho_z - p_z(0)}{p_z(0)}.\label{eq:valfn_deltas3}
\end{align}
Further, for any $k < K$, there holds:
\[ v(k,1|z) = 1 + v(k+1,0|z). \]
\end{lemma}

We now argue that for all $k < K$, $v(k+1,0|z) - v(k,0|z) < 0$.  The preceding lemma shows this holds at $k = K$ (trivially) and at $k = 0$.  The latter follows since $\pi_z(k) > 0$ for all $k$, and $p_z$ is strictly decreasing in $k$, so $p_z(0) > \rho_z$.  Now for all $k$ such that $\rho_z \leq p_z(k)$, we can argue inductively starting from $k = 0$ via \eqref{eq:valfn_deltas2} to conclude that $v(k+1,0|z) - v(k,0|z) < 0$ as well.  On the other hand, for all $k$ such that $\rho_z \geq p_z(k)$, we can argue inductively downwards starting from $k = K$ via \eqref{eq:valfn_deltas2} to conclude that $v(k+1,0|z) - v(k,0|z) < 0$.  Thus the result must hold for all $k$; in particular, $v(\cdot, 0|z)$ is a strictly decreasing function of the state $k$.

Thus for any $k' < k$, we have: 
\begin{align*}
v(k',1|0) - v(k,0|0) + h_0(k,0) &= v(k'+1,0|0) - v(k,0|0) + 1 - \rho_0 \\
& \geq 1 - \rho_0 > 0.
\end{align*}
So for any state $k < K$, we have:
\begin{align*}
\sum_{k',y'} P(k,0,k',y'|0) & \left( v(k',y'|0) -  v(k,0|0) + h_0(k,0) \right)^2 \\
& \geq \sum_{k' < k} P(k,0,k',1|0) \left( v(k',1|0) -  v(k,0|0) + h_0(k,0) \right)^2 \\
& \geq (1 - \rho_0)^2 \sum_{k' < k} P(k,0,k',1|0).
\end{align*}
Fix $k$ such that $k_0^- < k < k_0^+$.  We show that $\sum_{k' < k} P(k,0,k',1|0)$ is bounded below.  To see this, note that the next state is $(k',1)$ with $k' < k$ if (1) at least one departure occurs before the next arrival; and (2) the next arriving customer books.  Event (1) occurs with probability at least $\tau(k_0^-)/(\lambda + \tau(k_0^-)) > \b{\tau}(s_0^* - 3\epsilon)/(\lambda + \b{\tau}(s_0^* - 3\epsilon))  > 0$.  Conditional on event (1), event (2) occurs with probability at least $p(k_0^+) > \b{p}(s_0^* + 3\epsilon) > 0$.  Thus we have:
\[   \sum_{k' < k} P(k,0,k',1|0) >  \frac{\b{\tau}(s_0^* - 3\epsilon)}{\lambda + \b{\tau}(s_0^* - 3\epsilon)} \b{p}(k_0^+) \triangleq c_9 > 0. \]

In addition, for each $k < k_0^+$, note that $p_0(k) > p_0(k_0^+) > \b{p}_0(s_0^* + 3\epsilon)$, so that:
\[ \frac{\phi_0(k,1)^2}{\phi_a(k,1)} \geq \frac{\pi_0(k)^2}{\pi_a(k)} \cdot \b{p}_0(s_0^* + 3\epsilon)^2. \]

Putting the bounds together, we conclude that:
\begin{align*}
\sigma^2_{\UB} &\geq \sum_{k_0^- < k < k_0^+}  \frac{\phi_0(k,1)^2}{\phi_a(k,1)}
\left( \sum_{k',y'} P(k,0,k',y'|0) \left( v(k',y'|0) -  v(k,1|0) + h_0(k,1) \right)^2 \right) \\
&\geq (1 - \rho_0)^2 \b{p}_0(s_0^* + 3\epsilon)^2 c_9 \sum_{k < k_0^+} \frac{\pi_0(k)^2}{\pi_a(k)} \\
&\geq (1 - \rho_0)^2 \b{p}_0(s_0^* + 3\epsilon)^2 c_8 c_9 \gamma^K.
\end{align*}
Taking $C = (1 - \rho_0)^2 \b{p}_0(s_0^* + 3\epsilon)^2 c_8 c_9$ gives the desired result.
\end{proof}

\begin{proof}[Proof of Lemma \ref{lem:valfn_deltas}]
Fix $k, k', y'$.  The transition probability from $(k,0)$ to $(k',y')$ is zero if $k' > k$, and the following if $k' \leq k$:
\[ P(k,0,k',y'|z) = \left( \prod_{\kappa = k'+1}^{k} \frac{\tau(\kappa)}{\lambda + \tau(\kappa)} \right) \left( \frac{\lambda}{\lambda +\tau(k')}\right) ( 1 - p_z(k'))^{1 - y'} p_z(k')^{y'}, \]
where we interpret the first product as zero if $k' = k$.
From this expression, we obtain that for $k < K$:
\[ P(k+1,0,k',y'|z) = \left( \frac{\tau(k+1)}{\lambda + \tau(k+1)} \right) P(k,0,k',y'|z). \]
Since $h_z(k,0) = -\rho_z$, and $v$ satisfies Poisson's equation, we have:
\[ v(k,0|z) = -\rho_z + \sum_{k' \leq k} \sum_{y'=0,1} P(k,0,k',y'|z) v(k',y'|z). \]
Putting the previous observations together, we obtain that for $k < K$:
\begin{multline}
 v(k+1,0|z) = -\frac{\lambda \rho_z}{\lambda + \tau(k+1)}  + \frac{\tau(k+1)}{\lambda + \tau(k+1)} v(k,0|z) \\
+ \left(\frac{\lambda}{\lambda + \tau(k+1)} \right) ( (1  - p_z(k+1)) v(k+1,0|z) + p_z(k+1) v(k+1,1|z) ).
\label{eq:valfn_expansion}
\end{multline}
Now note that for $k < K$, we have $P(k,1,k',y'|z) = P(k+1,0,k',y'|z)$, since if a customer arrives to state $k$ and books a listing, then the state immediately becomes $k+1$.  Using this fact and Poisson's equation, it follows that $v(k,1|z) = 1 + v(k+1,0|z)$, for $k < K$.  Combining with the preceding expression, we obtain for $0 < k < K$:
\[  \tau(k) ( v(k,0|z) - v(k-1,0|z) )  = -\lambda \rho_z  + \lambda p_z(k) ( 1 +  v(k+1,0|z) - v(k,0|z)). \]
Rearranging terms gives the expression in the lemma for $0 < k < K$.

For $k = 0$, note that the only possible transitions from $(0,0)$ are to $(0,0)$ and $(0,1)$, and:
\begin{align*}
 v(0,0|z) &= -\rho_z + p_z(0) v(0,1|z) + (1 - p_z(0)) v(0,0|z)\\
&= -\rho_z + p_z(0) ( 1 + v(1,0|z)) + (1 - p_z(0)) v(0,0|z).
\end{align*}
Rearranging terms gives the expression in the lemma for $k = 0$.

Finally, for $k = K$, we have $p_z(k) = 0$.  Using \eqref{eq:valfn_expansion} with $k = K-1$, we obtain:
\[ v(K,0|z) = -\frac{\lambda \rho_z}{\lambda + \tau(K)}  + \frac{\tau(K)}{\lambda + \tau(K)} v(K-1,0|z) +
\left(\frac{\lambda}{\lambda + \tau(K)} \right) v(K,0|z). \]
Rearranging terms gives the expression in the lemma for $k = K$.
\end{proof}

%% file: Sections/app_proofs.tex
\section{Additional results and proofs}
\label{app:proofs}

\subsection{Section \ref{sec:cc}}
\label{app:cc_proofs}

\begin{proof}[Proof of Proposition \ref{pr:null_monotone}] 
Clearly if $\v{p}_1 = \v{p}_0$ then $\GTE =0.$ Now suppose $\v{p}_1 \neq \v{p}_0$, i.e., for at least one $k$, $p_1(k) \neq p_0(k)$. 
Under Assumption \ref{as:tau_monotone}, we either have $p_1(k) \geq p_0(k)$ for all $k$ or $p_1(k) \leq p_0(k)$ for all $k$; and since $\v{p}_1 \neq \v{p}_0$, we have strict inequality for at least one $k$.  We focus on the case where $p_1(k) \geq p_0(k)$ for all $k$, since the opposite case is symmetric. If $p_1(k) \geq p_0(k)$ for all $k$ and $\v{p}_0 \neq \v{p}_1$ (i.e., a strictly positive treatment), then $\v\pi_1 \succ^d  \v\pi_0$, cf.~Definition \ref{def:stochdom} and Corollary \ref{cor:invariant_start_dist}.

Now, we observe that
\begin{align*}
    \lambda \rho_1 = \sum_{k=0}^{K-1}\lambda p_1(k) \pi_1(k) &= \sum_{k=0}^{K-1} \tau(k+1)\pi_1(k+1)=\sum_{k=1}^{K}\tau(k)\pi_1(k)\\
    &>\sum_{k=1}^{K}\tau(k)\pi_0(k)=\sum_{k=0}^{K-1} \tau(k+1)\pi_0(k+1) = \sum_{k=0}^{K-1}\lambda p_0(k) \pi_0(k)=\lambda \rho_0.
\end{align*}
In the second equation, as well as the second to last equation, we use the detailed balance equations for the steady state distribution. In the inequality we use the fact that $\v\pi_1 \succ^d  \v\pi_0$ along with the fact that $\tau(\cdot)$ is a non-decreasing function with $\tau(0)=0$.

We conclude $\rho_1 > \rho_0$.  Since $\GTE = \rho_1 - \rho_0$, we must have $\GTE > 0$.  On the other hand, if $p_1(k) \leq p_0(k)$ for all $k$ and $\v{p}_1 \neq \v{p}_0$, then we have $\v\pi_1 \prec^d \v\pi_0$.  A symmetric argument then yields $\rho_0 > \rho_1$, and we would conclude $\GTE < 0.$ In either case, we have shown that $\v{p}_1 \neq \v{p}_0$ implies $\GTE \neq 0.$
\end{proof}

\subsection{Section \ref{sec:fpp}}
\label{app:fpp_proofs}

This appendix contains the proofs of Propositions \ref{pr:reduction_to_exchangeable} and \ref{pr:gte_and_var_aa_exchangeable}.  As noted in the text, our results in these propositions are related to classical results on design-based inference by both Neyman \cite{neyman1923, ding2024first} and Fisher \cite{fisher1971design, wager2024causal}.  Neyman considers a setting where the experiment is a completely randomized design (CRD), and potential outcomes are constant, so that the only randomness is due to the treatment assignment.  He shows in this case that the true variance of the DM estimator {\em decomposes} into terms that are unbiasedly estimable using $\widehat{\Var}$, and a last term that vanishes as long as treatment effects are constant across units.  

However, Neyman's argument does not apply directly in our setting, as outcomes remain random and dependent even conditional on treatment assignment---potentially contributing to the variance of the DM estimator.  We make progress by applying a uniform random permutation to the original outcomes, which creates exchangeable outcomes but leaves the distribution of $\DM$ and $\widehat{\Var}$ unchanged in an A/A test (cf.~Proposition \ref{pr:reduction_to_exchangeable}).  This approach is analogous to Fisher's development of permutation testing, which leverages the same fact to develop a hypothesis test of Fisher's sharp null $\tilde{H}_0$ (see, e.g., Section $11.2$ in \cite{wager2024causal}).  In our case, we use this reduction to be able to leverage an elementary argument analogous to Neyman to show $\widehat{\Var}$ is unbiased (cf.~Proposition \ref{pr:gte_and_var_aa_exchangeable}).

\begin{proof}[Proof of Proposition \ref{pr:reduction_to_exchangeable}] 

Let $\sigma^{-1}$ be the inverse permutation of $\sigma$.  Define $W_i = Z_{\sigma^{-1}(i)}$.  Observe that $(W_1, \ldots, W_N)$ is an independent collection of Bernoulli($a$) random variables.   Further, since $\sigma$ is a permutation chosen independent of $Y_1, \ldots, Y_N$, the variables $W_1, \ldots, W_N$ are also independent of $Y_1, \ldots, Y_N$.  In addition, note that $\sum_i W_i = \sum_i Z_i = N_1$, and $\sum_i (1 - W_i) = \sum_i (1 - Z_i) = N_0$.  

We now condition on the event that $N_1 > 1, N_0 > 1$.  Conditional on this event, the marginal distribution of $\v{Z}$ and $\v{W}$ agree.  Observe that:
\begin{align*}
\widetilde{\GTE}_N & = \frac{\sum_i Y_i W_i}{N_1} - \frac{\sum_i Y_i(1 - W_i)}{N_0}; \\
\widetilde{\Var}_N &= \frac{1}{N_1(N_1 - 1)} \sum_i W_i( Y_i  - \b{Y}(1))^2 + \frac{1}{N_0(N_0 - 1)} \sum_i (1 - W_i) (Y_i - \b{Y}(0))^2.
\end{align*}
Therefore conditional on the event $N_1 > 1, N_0 > 1$, the joint distribution of $(\widetilde{\GTE}_N, \widetilde{\Var}_N)$ is the same as the joint distribution of $(\DM_N, \widehat{\Var}_N)$.
\end{proof} 

\begin{proof}[Proof of Proposition \ref{pr:gte_and_var_aa_exchangeable}]
{\em Unbiasedness of $\DM$.}  We consider $\E[\b{Y}(1) | \v{Z}]$.  We have:
\[ \E \left[\frac{1}{N_1} \sum_i Z_i \E[Y_i]  \Bigg| \v{Z} \right]  = \mu. \]
The same is true for $\E[\b{Y}(0) | \v{Z}]$, so $\E[\DM | N_1 > 0, N_0 > 0] = 0$.  

{\em Unbiasedness of $\widehat{\Var}$.}  We first condition on the realization of  $\v{Z}$.  We have:
\begin{align}
\Var(\DM | \v{Z}) & = \frac{1}{N_1^2} ( N_1 V  + N_1(N_1 - 1) C) + \frac{1}{N_0^2} ( N_0 V + N_0(N_0 - 1) C) ) - \frac{1}{N_1 N_0} ( 2 N_1 N_0 C) \notag \\
& = \left( \frac{1}{N_1} + \frac{1}{N_0}\right) (V - C). \label{eq:truevar1}
\end{align}

Next we compute the expected value of the variance estimator $\widehat{\Var}$.  We again do this by first conditioning on the realization of the $\v{Z}$ vector.  Straightforward manipulation gives:
\[ \widehat{\Var} = \frac{1}{N_1(N_1-1)} \sum_i Z_i (Y_i^2 - \b{Y}(1)^2) + \frac{1}{N_0(N_0-1)} \sum_i (1 - Z_i) (Y_i^2 - \b{Y}(0)^2).\]
We consider the first term, since the analysis for the second is identical.  We have:
\begin{align*}
\E\left[ \b{Y}(1)^2 | \v{Z} \right] &= \frac{1}{N_1^2} \E\left[ \left( \sum_i Z_i Y_i \right)^2 \Bigg| \v{Z} \right] \\
&= \frac{1}{N_1^2} \E\left[ \sum_i Z_i Y_i^2 + \sum_{j \neq i} Z_i Z_j Y_i Y_j \Bigg| \v{Z} \right] \\
&= \frac{1}{N_1} ( V + \mu^2) + \frac{N_1 - 1}{N_1} ( C + \mu^2) \\
&= \mu^2 + \frac{V}{N_1} + \frac{(N_1 - 1)C}{N_1}.
\end{align*}
Thus we have:
\begin{align*}
\E \left[ \frac{1}{N_1(N_1 - 1)} \sum_i Z_i (Y_i^2 - \b{Y}(1)^2)\Bigg| \v{Z} \right] &= \frac{1}{N_1(N_1-1)} \left( N_1 (V + \mu^2) - N_1\mu^2 - V - (N_1 - 1)C \right) \\
& = \frac{1}{N_1} (V - C).
\end{align*}
Combining terms, we obtain that:
\[ \E\left[ \widehat{\Var} | \v{Z} \right] = \left( \frac{1}{N_1} + \frac{1}{N_0}\right) (V - C). \]

Finally, for notational simplicity, let $A =\{ N_1 > 1, N_0 > 1\}$.  It follows from the preceding derivation that:
\[\E\left[ \widehat{\Var} | A \right] = \E\left[ \frac{1}{N_1} + \frac{1}{N_0} \Bigg| A \right] (V - C). \]
We now compute the exact variance of the difference-in-means estimator conditional on $A$, denoted $\Var(\DM| A)$.  Note that $\Var( \E[ \DM | \v{Z}] | A ) = 0$, since $\E[\DM | \v{Z}] = 0$ for all realizations of $\v{Z}$.  Therefore,
\[ \Var(\DM | A ) = \E[ \Var(\DM | \v{Z}) | A] = \E\left[ \frac{1}{N_1} + \frac{1}{N_0} \Bigg| A \right] (V - C).\]
This completes the proof.
\end{proof}

\subsection{Section \ref{sec:power}}
\label{app:power_proofs}

The following proposition and the subsequent corollary, both used in proving Theorem \ref{thm: pos bias DM}, establish dominance relationships between the distribution of the state as seen by arriving customers (as well as the steady state distribution), depending on the treatment assignment vector.  The proofs follow their statements.

\begin{proposition}
\label{prop:distributiondom}
Suppose the treatment is strictly positive, cf.~Definition \ref{def:monotone}.  Then 
$\v{\nu}_i^{\v{0}} \prec^d \v{\nu}_i^{\v{1}}$.

Further, let $\v{Z}=(Z_1,\dots,Z_N)$ be a treatment assignment with $N_0>0$ and $N_1>0$, so that in particular, $\v{Z} \neq \v{1}$ and $\v{Z} \neq \v{0}$.  Then for every customer $i$, either:
\begin{enumerate}
    \item[(1)] $\v{\nu}_i^{\v{Z}} =\v{\nu}_i^{\v{1}}$; 
    \item[(2)] $\v{\nu}_i^{\v{Z}} = \v{\nu}_i^{\v{0}}$; or
    \item[(3)] $\v{\nu}_i^{\v{0}} \prec^d \v{\nu}_i^{\v{Z}} \prec^d \v{\nu}_i^{\v{1}}$.
\end{enumerate}
In addition, there exists at least one customer $i$ such that (3) holds for all $j \geq i$.
\end{proposition}

Recall that since the distribution on arrival of the first customer does not depend on any treatment assignments, we simply write $\v{\nu}_1$ for this distribution (without superscript $\v{Z}$), and that $\v{\pi}_0$ (resp., $\v{\pi}_1$) is the steady state distribution for the global control (resp., global treatment) Markov chain.

\begin{corollary}
\label{cor:invariant_start_dist}
Suppose the treatment is strictly positive, cf.~Definition \ref{def:monotone}.   If $\v{\nu}_1 = \v{\pi}_0$, then $\v{\nu}_i^{\v{0}} = \v{\pi}_0$ for all $i$.  Further, for all $i > 1$, $\v{\pi}_0 \prec^d \v{\nu}_i^{\v{1}} \prec^d \v{\pi}_1$.  

Similarly, if $\v{\nu}_1 = \v{\pi}_1$, then $\v{\nu}_i^{\v{1}} = \v{\pi}_1$ for all $i$.  Further, for all $i > 1$, $\v{\pi}_0 \prec^d \v{\nu}_i^{\v{0}} \prec^d \v{\pi}_1$.  
\end{corollary}

We note that if treatments are strictly negative, then all inequalities and stochastic dominance relations in Proposition \ref{prop:distributiondom} and Corollary \ref{cor:invariant_start_dist} are reversed.

\begin{proof}[Proof of Proposition \ref{prop:distributiondom}]

We prove the proposition via induction on customers, but with two preliminary steps to establish the required dominance relationships needed.

{\em Step 1: Characterize the state distribution after customer arrival.}  We start by noting that if the treatment assignment is $\v{Z}$, and the booking probabilities are $\v{q}$, then if we let $\v{\nu}_i^{\v{Z}\dag}$ denote the distribution {\em after} the arrival of customer $i$, we have:
\begin{equation}
\nu_i^{\v{Z}\dag}(k) = \left\{ \begin{array}{cl}
\nu_i^{\v{Z}}(k) ( 1- q(k)) + \nu_i^{\v{Z}}(k-1) q(k-1),&\ \  k > 0;\\
\nu_i^{\v{Z}}(0) (1 - q(0)),&\ \ k = 0.
\end{array}
\right. \label{eq:dist_after_arrival}
\end{equation}

Let $\v{\phi}, \v{\psi}$ be two distributions on $(0,\ldots,K)$ such that $\v{\phi} \succ^d \v{\psi}$, with $\phi(k) > 0, \psi(k) > 0$ for all $k$.  We compare four scenarios for customer $i$: (i) $\v{\nu}_i^{\v{Z}} = \v{\phi}$, and $Z_i = 1$; (ii) $\v{\nu}_i^{\v{Z}} = \v{\phi}$, and $Z_i = 0$; (iii) $\v{\nu}_i^{\v{Z}} = \v{\psi}$, and $Z_i = 1$; and (iv) $\v{\nu}_i^{\v{Z}} = \v{\psi}$, and $Z_i = 0$.  Let $\v{\phi}^\dag$ and $\v{\phi}^\ddag$ denote the distribution of the state {\em after} customer $i$ arrives, in cases (i) and (ii) respectively; and let $\v{\psi}^\dag$ and $\v{\psi}^\ddag$ denote the distribution of the state after customer $i$ arrives, in cases (iii) and (iv) respectively.  

We show several dominance relationships between these distributions, which we refer to as (DR1)-(DR5).  
\begin{enumerate}
\item[(DR1)] $\v{\phi}^\dag \succ^d \v{\phi}^\ddag$.  Using \eqref{eq:dist_after_arrival}, a straightforward calculation yields for each $k > 0$:
\begin{align*}
\sum_{k' \geq k} \phi^\dag(k') &= \phi(K)  + \cdots + \phi(k) +  \phi(k-1) p_1(k-1)\\
&\geq \phi(K) + \cdots + \phi(k) + \phi(k-1) p_0(k-1)\\
&= \sum_{k' \geq k} \phi^\ddag(k'),
\end{align*}
with strict inequality for any $k$ such that $p_1(k-1) > p_0(k-1)$; note that such a $k$ must exist, since the intervention is strictly positive.  

\item[(DR2)] $\v{\psi}^\dag \succ^d \v{\psi}^\ddag$.  This follows by an analogous argument to the previous step.

\item[(DR3)] $\v{\phi}^\dag \succ^d \v{\psi}^\dag$.  Using \eqref{eq:dist_after_arrival}, we obtain for each $k > 0$:
\begin{align*}
\sum_{k' \geq k} \phi^\dag(k') &= \phi(K)  + \cdots + \phi(k) +  \phi(k-1) p_1(k-1)\\
&= (\phi(K) + \cdots + \phi(k) +  \phi(k-1)) p_1(k-1) \\
&\quad + ( \phi(K)  + \cdots + \phi(k)) (1 - p_1(k-1)) \\
&\geq (\psi(K) + \cdots + \psi(k) +  \psi(k-1)) p_1(k-1)  \\
&\quad + ( \psi(K) + \cdots + \psi(k)) (1 - p_1(k-1))\\
&= \sum_{k' \geq k} \psi^\dag(k'),
\end{align*}
with strict inequality for any $k$ where $\sum_{k' \geq k} \phi(k) > \sum_{k' \geq k} \psi(k)$, as required.  

\item[(DR4)] $\v{\phi}^\ddag \succ^d \v{\psi}^\ddag$. This follows by an analogous argument to the previous step.

\item[(DR5)] $\v{\phi}^\dag \succ^d \v{\psi}^\ddag$.  This follows by combining the previous steps.
\end{enumerate}

{\em Step 2: Characterize the behavior of customer departures between customer arrivals.}  Observe that between the arrival of customer $i-1$ and the arrival of customer $i$, the state can only evolve through customer departures.  Accordingly, define the following generator $\v{D}$ for a continuous-time pure death process corresponding to the departure process:
\[ \v{D}(k,k') = \left\{ \begin{array}{cl}
\tau_k,&\ \ k' = k-1;\\
-\tau_k,&\ \ k' = k;\\
0,&\ \ \text{otherwise}.
\end{array} \right. \]
Associated to $\v{D}$, we define the $t$-step transition probabilities as follows (abusing notation):
\[ \v{D}^t(k,k') = \exp(t\v{D}) = \sum_{m \geq 0} \frac{(t\v{D})^m}{m!}. \]
It is well known that such matrices $\v{D}$ (and in fact, general birth-death chains) are stochastically monotone, in the sense that they preserve stochastic dominance between distributions; for details, see \cite{keilson1977monotone}, Example 2.3(b).\footnote{A subtlety here is that the definition of stochastic dominance used in \cite{keilson1977monotone} does not require strict inequality in at least one state $k$, as we do in Definition \ref{def:stochdom}. However, it is straightforward to show that the results in Example 2.3(b) of \cite{keilson1977monotone} generalize to the stronger form of stochastic dominance in Definition \ref{def:stochdom} for any irreducible birth-death chain of the type we consider in \eqref{eq:generator}; we omit the details.}

In particular, if $\v{\phi} \succ^d \v{\psi}$, then we obtain that for all $t$:
\begin{equation}
\v{\phi} \v{D}^t \succ^d \v{\psi}\v{D}^t. \label{eq:time_t_dominance}
\end{equation}

{\em Step 3: Complete the proof via induction.}  We prove the proposition by induction on customers.  For $i = 1$, the result holds because $\v{\nu}_1^{\v{Z}}$ does not depend on the treatment assignment $\v{Z}$, so $\v{\nu}_1^{\v{0}} = \v{\nu}_1^{\v{Z}} = \v{\nu}_1^{\v{1}}$.  Suppose the result holds for customers $1, \ldots, i-1$.

We observe that $\nu_i^{\v{Z}}$ can be informally obtained from $\nu_{i-1}^{\v{Z}}$ as follows: first, we realize the booking outcome of customer $i-1$; and second, we allow customers to depart during the intearrival time between customer $i$ and $i-1$, i.e., $A_i - A_{i-1}$.  Note that $A_i - A_{i-1}$ is an independent exponential random variable with mean $1/\lambda$, due to the Poisson arrival process.  Formally, we have:
\begin{equation}
\v{\nu}_i^{\v{Z}} = \int_0^{\infty} \lambda e^{-\lambda t}\v{\nu}_{i-1}^{\v{Z}\dag} \v{D}^t dt. \label{eq:induction}
\end{equation}

Now we use (DR1)-(DR5), \eqref{eq:time_t_dominance}, and \eqref{eq:induction} to check each case in the proposition for customer $i$, varying whether customer $i-1$ satisfied (1), (2), or (3) in the proposition, and also varying the treatment status of customer $i-1$.  
\begin{enumerate}
\item If $\v{\nu}_{i-1}^{\v{\v{Z}}} = \v{\nu}_{i-1}^{\v{1}}$ and $Z_{i-1} = 1$, then $\v{\nu}_i^{\v{Z}} =\v{\nu}_i^{\v{1}}$, so (1) holds.
\item If $\v{\nu}_{i-1}^{\v{\v{Z}}} = \v{\nu}_{i-1}^{\v{0}}$ and $Z_{i-1} = 0$, then $\v{\nu}_i^{\v{Z}} =\v{\nu}_i^{\v{0}}$, so (2) holds.
\item If $\v{\nu}_{i-1}^{\v{\v{Z}}} = \v{\nu}_{i-1}^{\v{1}}$ and $Z_{i-1} = 0$, then (3) holds by (DR2) applied for customer $i-1$, combined with \eqref{eq:time_t_dominance} and \eqref{eq:induction}. 
\item If $\v{\nu}_{i-1}^{\v{\v{Z}}} = \v{\nu}_{i-1}^{\v{0}}$ and $Z_{i-1} = 1$, a symmetric argument to the previous step shows (3) holds, using (DR1) instead.
\item If $\v{\nu}_{i-1}^{\v{0}} \prec^d \v{\nu}_{i-1}^{\v{Z}} \prec^d \v{\nu}_{i-1}^{\v{1}}$, and $Z_{i-1} = 1$, then we can show that $\v{\nu}_{i}^{\v{0}} \prec^d \v{\nu}_i^{\v{Z}}$ using (DR5), combined with \eqref{eq:time_t_dominance} and \eqref{eq:induction}.  We can show that $\v{\nu}_{i}^{\v{Z}} \prec^d \v{\nu}_i^{\v{1}}$ using (DR3), combined with \eqref{eq:time_t_dominance} and \eqref{eq:induction}.  Thus (3) holds.
\item If $\v{\nu}_{i-1}^{\v{0}} \prec^d \v{\nu}_{i-1}^{\v{Z}} \prec^d \v{\nu}_{i-1}^{\v{1}}$, and $Z_{i-1} = 0$,  
then we can show that $\v{\nu}_{i}^{\v{0}} \prec^d \v{\nu}_i^{\v{Z}}$ using (DR4), combined with \eqref{eq:time_t_dominance} and \eqref{eq:induction}.  We can show that $\v{\nu}_{i}^{\v{Z}} \prec^d \v{\nu}_i^{\v{1}}$ using (DR5), combined with \eqref{eq:time_t_dominance} and \eqref{eq:induction}.  Thus again (3) holds.
\end{enumerate} 

In all cases we have shown that one of (1), (2), or (3) holds.  Finally, since $N_0 > 1$ and $N_1 > 1$, it follows from our derivation that there must be at least one customer for whom (3) holds, and for all subsequent customers (3) holds.  This completes the proof.
\end{proof}

\begin{proof}[Proof of Corollary \ref{cor:invariant_start_dist}]
If the first customer arrives to find the queue in $\v{\pi}_0$ (i.e., $\v{\nu}_1 = \v{\pi}_0$) and $Z_i = 0$ for all $i$, then we remain in the global control system for all customers.  Therefore every subsequent customer also arrives to find the queue in $\v{\pi}_0$ by the PASTA (Poisson arrivals see time averages) property.
Thus in this case $\v{\pi}_i^{\v{0}} = \v{\pi}_0$ for all $i$.   If instead $\v{\nu}_1 = \v{\pi}_1$, then $\v{\nu}_i^{\v{1}} = \v{\pi}_1$ for all $i$.

Now suppose that $\v{\nu}_1 = \v{\pi}_0$.  We adopt the same notation and definitions as the proof of Proposition \ref{prop:distributiondom}.  Consider customer $2$.  Using (DR1), \eqref{eq:time_t_dominance}, and \eqref{eq:induction}, it follows that $\v{\pi}_0 = \v{\nu}_2^{\v{0}} \prec^d \v{\nu}_2^{\v{1}}$.  With this as the base case, now suppose the inductive hypothesis holds that $\v{\pi}_0 = \v{\nu}_j^{\v{0}} \prec^d \v{\nu}_j^{\v{1}}$ for all customers $j = 2, \ldots, i-1$.  By using (DR5) we obtain $\v{\nu}_{i-1}^{\v{1}\dag} \succ^d \v{\nu}_{i-1}^{\v{0}\ddag}$.  Combining with \eqref{eq:time_t_dominance} and \eqref{eq:induction}, we obtain $\v{\pi}_0  = \v{\nu}_i^{\v{0}} \prec^d \v{\nu}_i^{\v{1}}$.  Thus $\v{\pi}_0 \prec^d \v{\nu}_i^{\v{1}}$ for all $i$.  

To complete the proof, take limits as $i \to \infty$ in case (3) of Proposition \ref{prop:distributiondom}.  By PASTA, we have $\v{\nu}_i^{\v{0}} \to \v{\pi}_0$, and $\v{\nu}_i^{\v{1}} \to \v{\pi}_1$.  Thus $\v{\pi}_0 \prec^d \v{\pi}_1$.  Now suppose we let $\v{\psi} = \v{\pi}_0$, and $\v{\phi} = \v{\pi}_1$.  Then applying (DR3), together with \eqref{eq:time_t_dominance} and \eqref{eq:induction}, we conclude that if $\v{\nu}_1 = \v{\pi}_0$, then $\v{\nu}_2^{\v{1}} \prec^d \v{\pi}_1$.  We can continue in this way, inductively applying (DR3) together with \eqref{eq:time_t_dominance} and \eqref{eq:induction}, to conclude that if $\v{\nu}_1 = \v{\pi}_0$, then $\v{\nu}_i^{\v{1}} \prec^d \v{\pi}_1$ for all $i > 1$.  

A similar set of arguments can be used to establish that if $\v{\nu}_1 = \v{\pi}_1$, then $\v{\pi}_0 \prec^d \v{\nu}_i^{\v{0}} \prec^d \v{\pi}_1$ for all $i > 1$; we omit the details.  This completes the proof.
\end{proof}

\begin{proof}[Proof of Theorem \ref{thm: pos bias DM}]
We start by noting that given a treatment assignment vector $\v{Z}$, if customer $i$ has treatment status $Z_i = z$ and sees state distribution $\v{\nu}_i^{\v{Z}}$ on arrival, then:
\begin{equation}
\E[Y_i |  \v{Z}] = \sum_k \nu_i^{\v{Z}}(k) p_z(k). \label{eq:EY_i}
\end{equation} 

We define one additional piece of notation that will be useful for this proof.  If $\v{\phi}$, $\v{\psi}$ are two distributions on $\{0,\ldots,K\}$, then we write $\v{\phi} \succeq^d \v{\psi}$ if either $\v{\phi} \succ^d \v{\psi}$ or $\v{\phi} = \v{\psi}$.

Now let $\v{\zeta}$ denote a treatment assignment vector with $\sum_{i = 1}^N \zeta_i = N_1 > 0$, and $N_0 = N - N_1 > 0$.  Recall that $\v{\nu}_1 = \v{\pi}_0$, by Assumption \ref{as: first customer distr}.  Combining Corollary \ref{cor:invariant_start_dist} with Proposition \ref{prop:distributiondom}, we can conclude that for all $i$:
\[ \v{\pi}_0 = \v{\nu}_i^{\v{0}} \preceq^d \v{\nu}_i^{\v{\zeta}} \preceq^d \v{\nu}_i^{\v{1}} \preceq^d \v{\pi}_1.  \]
In particular, for any customer $i$, we have:
\begin{equation}
\v{\pi}_0 \preceq^d \v{\nu}_i^{\v{\zeta}} \preceq^d \v{\pi}_1.
\label{eq:ss_dom_relationships} 
\end{equation}

Now for any customer $i$ with $\zeta_i = 1$, if we take the expectation of $\v{p}_1$ in the second relation in \eqref{eq:ss_dom_relationships} and consider \eqref{eq:EY_i}, then we obtain:
\begin{equation}
 \E[Y_i | \v{Z} = \v{\zeta}] \geq \sum_k \pi_1(k) p_1(k) = \rho_1.
\label{eq:rho_bound_1}
\end{equation}
Similarly, for any customer $i$ with $\zeta_i = 0$, if we take the expectation of $\v{p}_0$ in the first relation in \eqref{eq:ss_dom_relationships} and consider \eqref{eq:EY_i}, then we obtain:
\begin{equation}
 \E[Y_i | \v{Z} = \v{\zeta}] \leq \sum_k \pi_0(k) p_0(k) = \rho_0.
\label{eq:rho_bound_2}
\end{equation}
Note that there must be a customer $i$ such that case (3) of Proposition \ref{prop:distributiondom} holds for all $j \geq i$, it is straightforward to check that for all these customers $j$, the corresponding inequality in \eqref{eq:rho_bound_1}-\eqref{eq:rho_bound_2} is strict.  

Thus we conclude that:
\begin{align*}
\E[\widehat{\GTE} | \v{Z} = \v{\zeta}] &= \frac{1}{N_1} \sum_i \zeta_i \E[Y_i | \v{Z} = \v{\zeta}]  - \frac{1}{N_0} \sum_i (1 - \zeta_i) \E[Y_i | \v{Z} = \v{\zeta}] \\
&> \rho_1 - \rho_0 = \GTE.
\end{align*}
Since this holds for every $\v{\zeta}$ with $N_1 > 0$ and $N_0 > 0$, we conclude:
\[ \E[\widehat{\GTE} | N_1 > 0, N_0 > 0]  > \GTE. \]
Finally, note that since $\v{\pi}_1 \succ^d \v{\pi}_0$ by Corollary \ref{cor:invariant_start_dist}, and the treatment is strictly positive, it follows that $\GTE > 0$.  All the assertions are reversed if the treatment is strictly negative, completing the proof.
\end{proof}

\begin{proof}[Proof of Theorem \ref{thm: pos bias limit}]
We first show that for any $0 < a < 1$ we have $\v{\pi}_0 \prec^d \v{\pi}_a\prec^d \v{\pi}_1$.  We use Proposition \ref{prop:distributiondom} together with the PASTA property (Poisson arrivals see time averages).  Let $Z_1, Z_2, \ldots$ be an i.i.d.~Bernoulli($a$) sequence.  Suppose Assumption \ref{as: first customer distr} holds.  Let $N_{1N} = \sum_{i = 1}^N Z_i$, and let $N_{0N} = N - N_{1N}$.  Let $\v{Z}_N = (Z_1, \ldots, Z_N)$.  Observe that from Proposition \ref{prop:distributiondom}, it follows that:
\[ \nu_i^{\v{0}} \prec^d \E[ \nu_i^{\v{Z}_i} | N_{0i} > 0, N_{1i} > 0] \prec^d \nu_i^{\v{1}}, \]
where the expectation is only over the randomness in $\v{Z}_i$.  Further, observe that $\E[ \nu_i^{\v{Z}_i}]$ is the state distribution just prior to the arrival of the $i$'th customer in a Bernoulli($a$) customer-randomized experiment.  Taking $N \to \infty$, noting that $P(N_{0i} = 0 \text{ or } N_{1i} = 0) \to 0$ as $N \to \infty$, and applying PASTA, we conclude that:
\[ \v{\pi}_0 \prec^d \v{\pi}_a \prec^d \v{\pi}_1. \]

Now note that Assumption \ref{as:booking_prob_monotone} implies
  \begin{align}
      \sum_{k} p_1(k) \pi_a(k) > \sum_{k} p_1(k) \pi_1(k);\\
      \sum_{k} p_0(k) \pi_0(k) > \sum_{k} p_0(k) \pi_a(k),
  \end{align}
  which allows us to conclude that
  \[\ADE_a= \sum_{k}(p_1(k)- p_0(k)) \pi_a(k)>\sum_{k} p_1(k) \pi_1(k)-\sum_{k} p_0(k) \pi_0(k)=\GTE.\]
\end{proof}

%% file: Sections/app_sims.tex
\section{Statistical power: Additional numerics}
\label{app:sims}
 In this section, we report additional numerical results in the setting of Section \ref{sec: asymptotic power}, by visualizing performance across variation in a wider range of parameters.  In Section \ref{sec:sims var}, we present additional numerical results comparing the Cram\'{e}r-Rao lower bound $\sigma^2_{\UB}$ to the true variance of $\DM$ as well as the variance estimator $\widehat{\Var}$.  In Section \ref{sec: sims power} we present present additional numerical results comparing statistical power of the decision rule \eqref{eq:decision} using the test statistic $\hat{T}_N$ (cf.~\eqref{eq:tstat}) as compared to the same rule using any unbiased test statistic $\hat{U}_N$ (cf.~\eqref{eq:ub_tstat}).
In all cases studied, we find that if the state space is sufficiently large, we observe that $\sigma^2_{\UB}$ is much larger than either $\Var(\DM)$ or $\widehat{\Var}$, and correspondingly that the power obtained is significantly higher using $\hat{T}_N$ instead of $\hat{U}_N$.

\subsection{Setup}
\label{sec:sims setup}

Throughout the numerics in this section, we consider the following values for the model parameters:
\begin{align*}
    K &=200;\\
   \v\tau^{(K)}(k) &= k\bar\tau \text{ with } \bar{\tau}=1;\\
   \lambda^{(K)} &= K \cdot\bar{\lambda}  \text{ with } \bar{\lambda}=1.5;\\
   \bar{\epsilon}=1.0;\\
    a &= 0.5;\\
    v_0 &= 0.5;\\
    v_1&= v_0+\delta \text{ with } \delta=0.05;\\
    p_0^{(K)}(k)&= \frac{(K-k)v_0}{\bar{\epsilon}+(K-k)v_0};\\
    p_1^{(K)}(k)&= \frac{(K-k)v_1}{\bar{\epsilon}+(K-k)v_1};\\
\end{align*}
The values of $\bar{\lambda},\bar{\tau},$ and $K$ were chosen to ensure that the system is large and has sufficient interference between customers. The values of $v_0$, $\delta$, and $\bar{\epsilon}$ were chosen to produce a small treatment effect; in this case we have $\GTE = 0.022$.

In all the numerics in this section, $a$ and $v_0$ are fixed. In each set of numerics, we fix all but one of the model parameters $K$, $\bar{\lambda}/\bar{\tau}$, $\bar{\epsilon}$, $\delta$, and vary the held out parameter. 

\subsection{Variance comparisons under sign-consistent treatments}
\label{sec:sims var}

In this section, compute the scaled limit of $\widehat{\Var}_N$ from Theorem \ref{thm:var_est} and the lower bound on $\sigma^2_{\UB}$ from Theorem \ref{thm: farias} in Bernoulli randomized user-level experiments with each of $K,\lambda/\tau, \bar{\epsilon}, \delta$ varying, and the other parameters fixed as described above.  In Figure \ref{fig:CR comparison}, we see that the bound on the variance of an unbiased estimator is larger than the true variance of $\DM$ regardless of how we vary the model parameters.

\begin{figure}[ht]
    \centering
    \begin{subfigure}[b]{0.7\textwidth}
        \centering
        \includegraphics[width=\textwidth]{Sections/figures/CR_vary_K.png}
        \caption{$K$ varying from $100$ to $300$, and other model parameters fixed as in Section \ref{sec:sims setup}. }
        \label{fig:plot1}
    \end{subfigure}
    \begin{subfigure}[b]{0.7\textwidth}
        \centering
        \includegraphics[width=\textwidth]{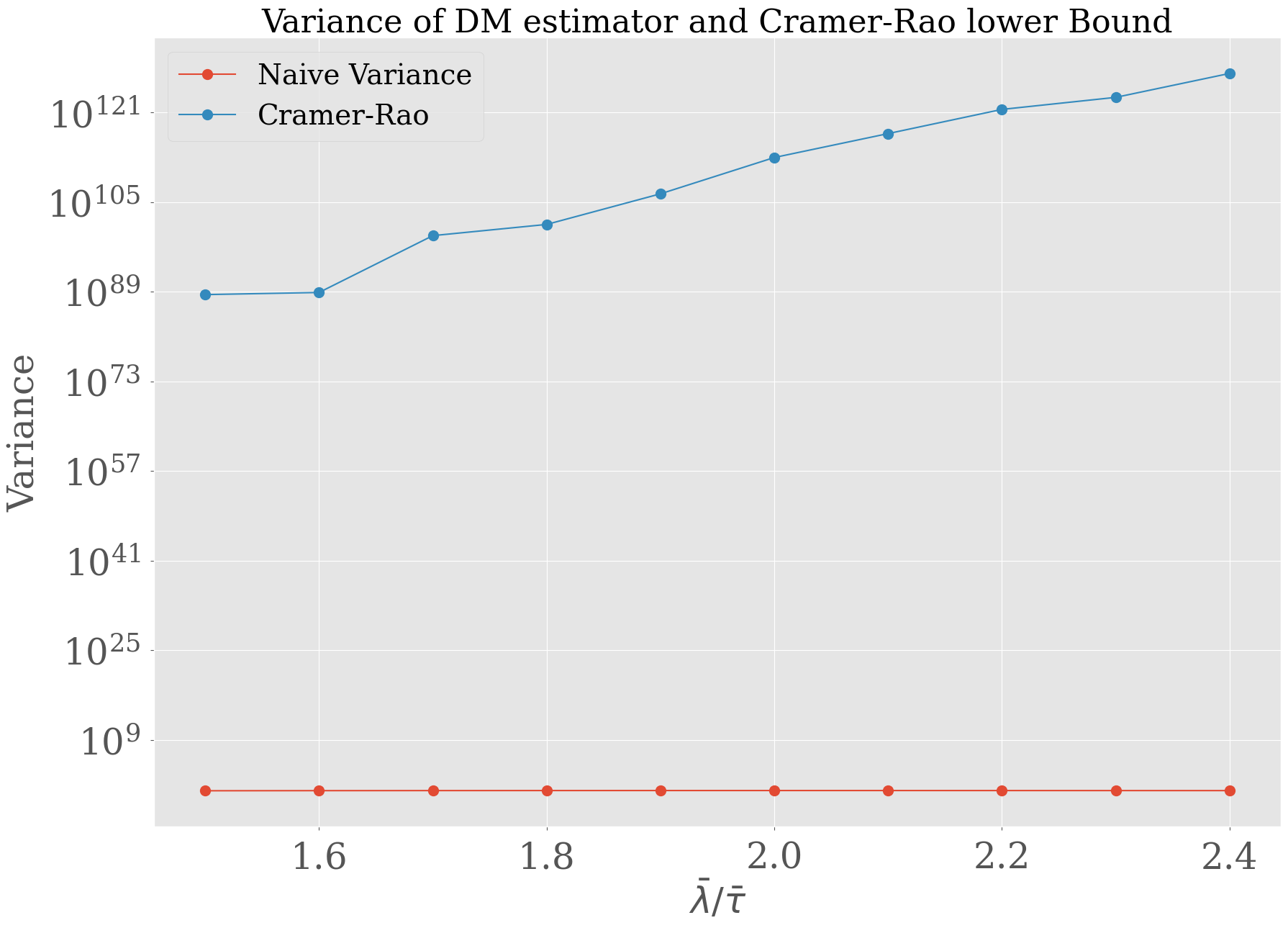}
        \caption{$\bar\lambda/\bar\tau$ varying from $1.5$ to $2.5$, and other model parameters fixed as in Section \ref{sec:sims setup}.}
        \label{fig:plot2}
    \end{subfigure}

    \caption{Asymptotic variance comparison using Bernoulli randomized user-level experiments.  The Cram\'{e}r-Rao lower bound (cf. Theorem \ref{thm: farias}) remains above the scaled limit of $\widehat{\Var}_N$  regardless of how we vary the model parameters. }
    \label{fig:CR comparison}
\end{figure}
\begin{figure}[ht]\ContinuedFloat
\centering
 \begin{subfigure}[b]{0.7\textwidth}
        \centering
        \includegraphics[width=\textwidth]{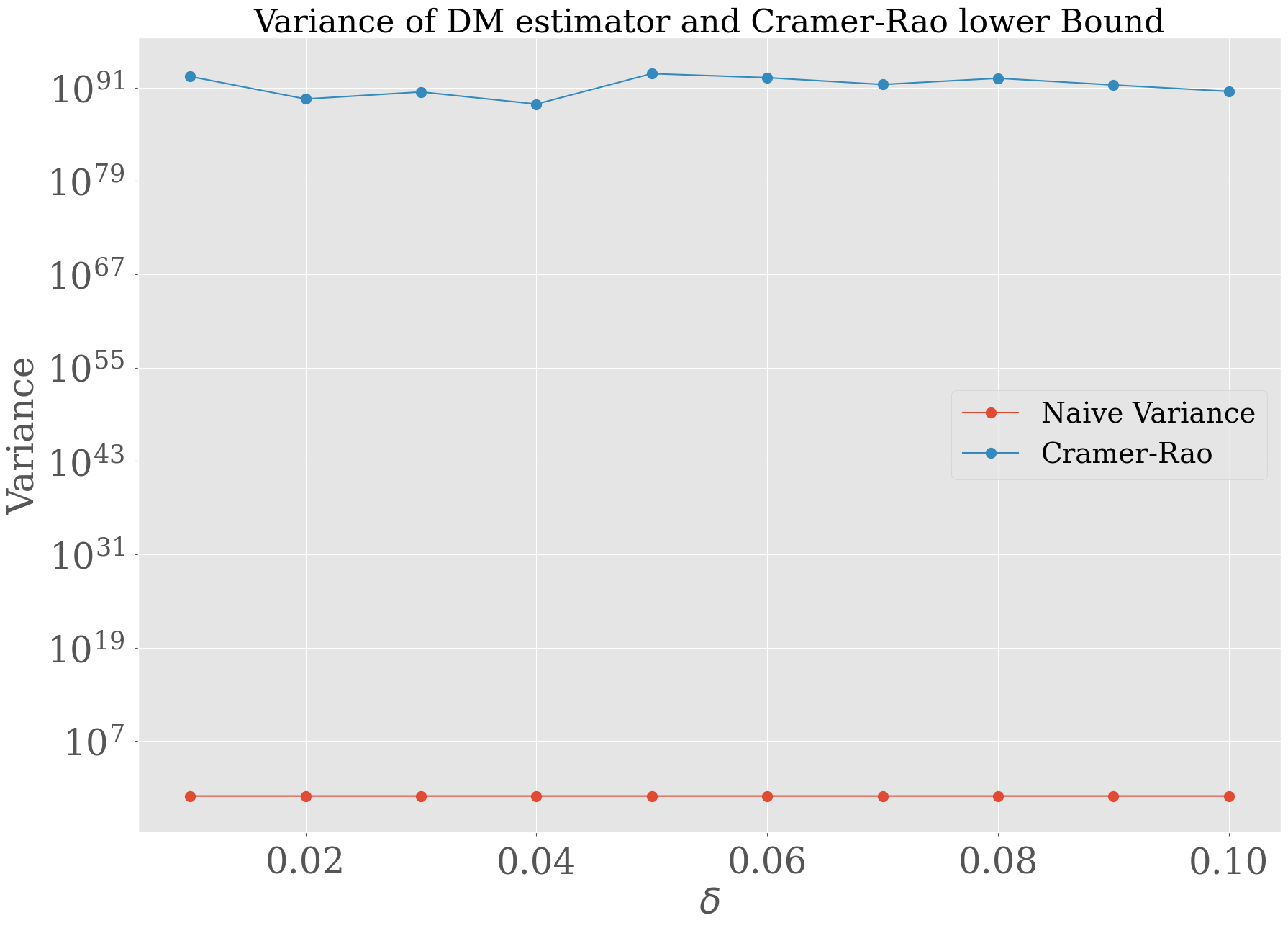}
        \caption{$\delta$ varying from $0.01$ to $0.1$, and other model parameters fixed as in Section \ref{sec:sims setup}.}
        \label{fig:plot4}
    \end{subfigure}
    \begin{subfigure}[b]{0.7\textwidth}
        \centering
        \includegraphics[width=\textwidth]{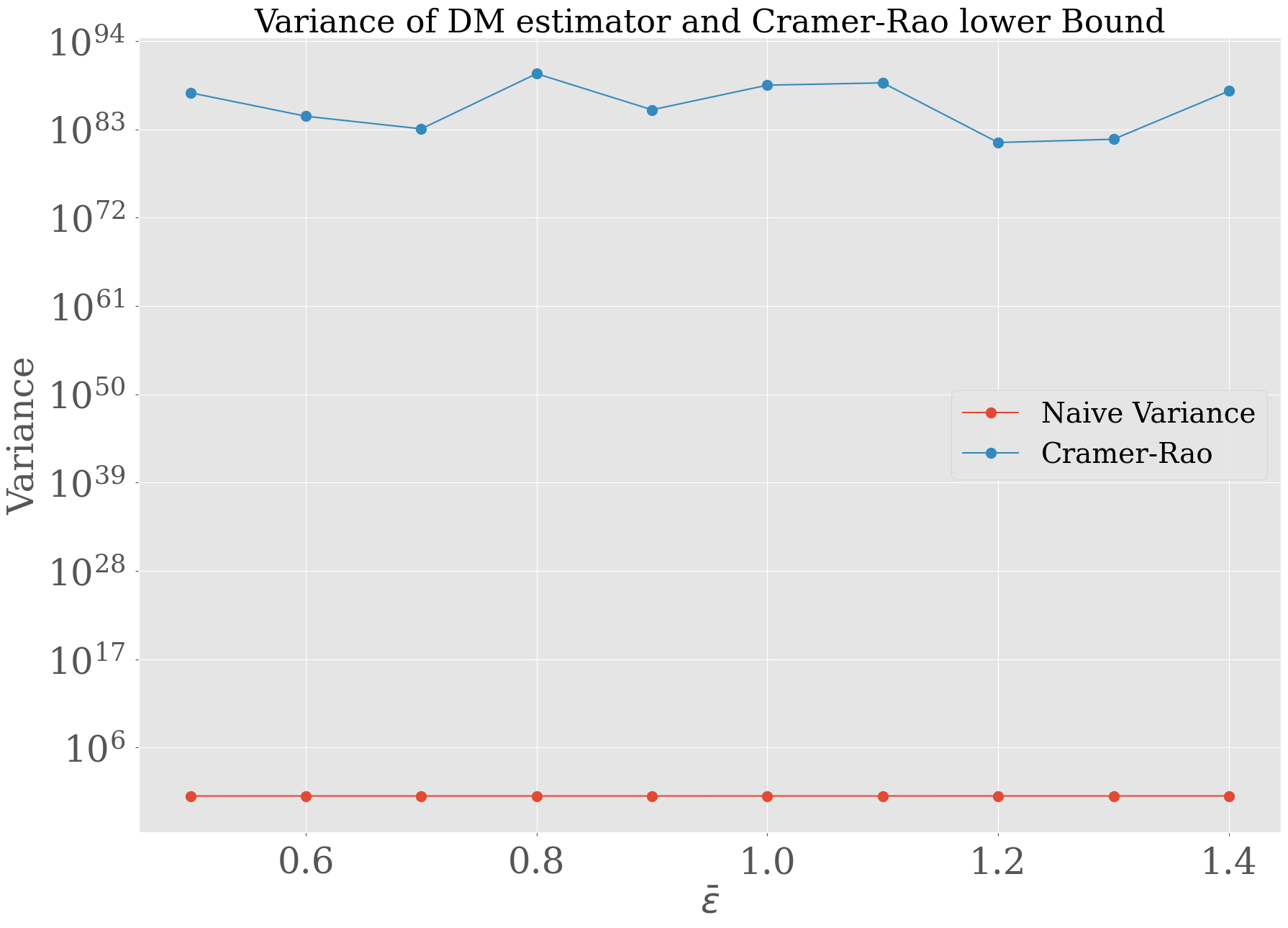}
        \caption{$\bar\epsilon$ varying from $0.5$ to $1.5$, and other model parameters fixed as in Section \ref{sec:sims setup}.}
        \label{fig:plot3}
    \end{subfigure}
    \caption{Asymptotic variance comparison using  Bernoulli randomized user-level experiments.  The Cram\'{e}r-Rao lower bound (cf. Theorem \ref{thm: farias}) remains above the scaled limit of $\widehat{\Var}_N$ regardless of how we vary the model parameters. }
 \label{fig:CR comparison2}
\end{figure}

\subsection{Power comparisons under sign-consistent treatments}
\label{sec: sims power}

In this section, we use the CLT from Theorem \ref{thm:ab_clt_cc} and the power calculation \eqref{eq:ub_power} to compute the false negative probability of the naive test statistic and any unbiased test statistic that obeys \eqref{eq:unbiased_clt}.
From Figures \ref{fig:power vary K}, \ref{fig:power vary lambda}, \ref{fig:power vary a}, and \ref{fig:power vary treatment}, we see that the decision rule using the test statistic $\hat{T}_N$ (cf.~\eqref{eq:tstat}) exhibits higher power than the same rule using any unbiased test statistic $\hat{U}_N$. In particular, the FNP of the test statistic $\hat{T}_N$ decays at a faster rate as $N$ increases in all cases.
\begin{figure}[ht!]
  \centering
    \includegraphics[width=0.49\textwidth]{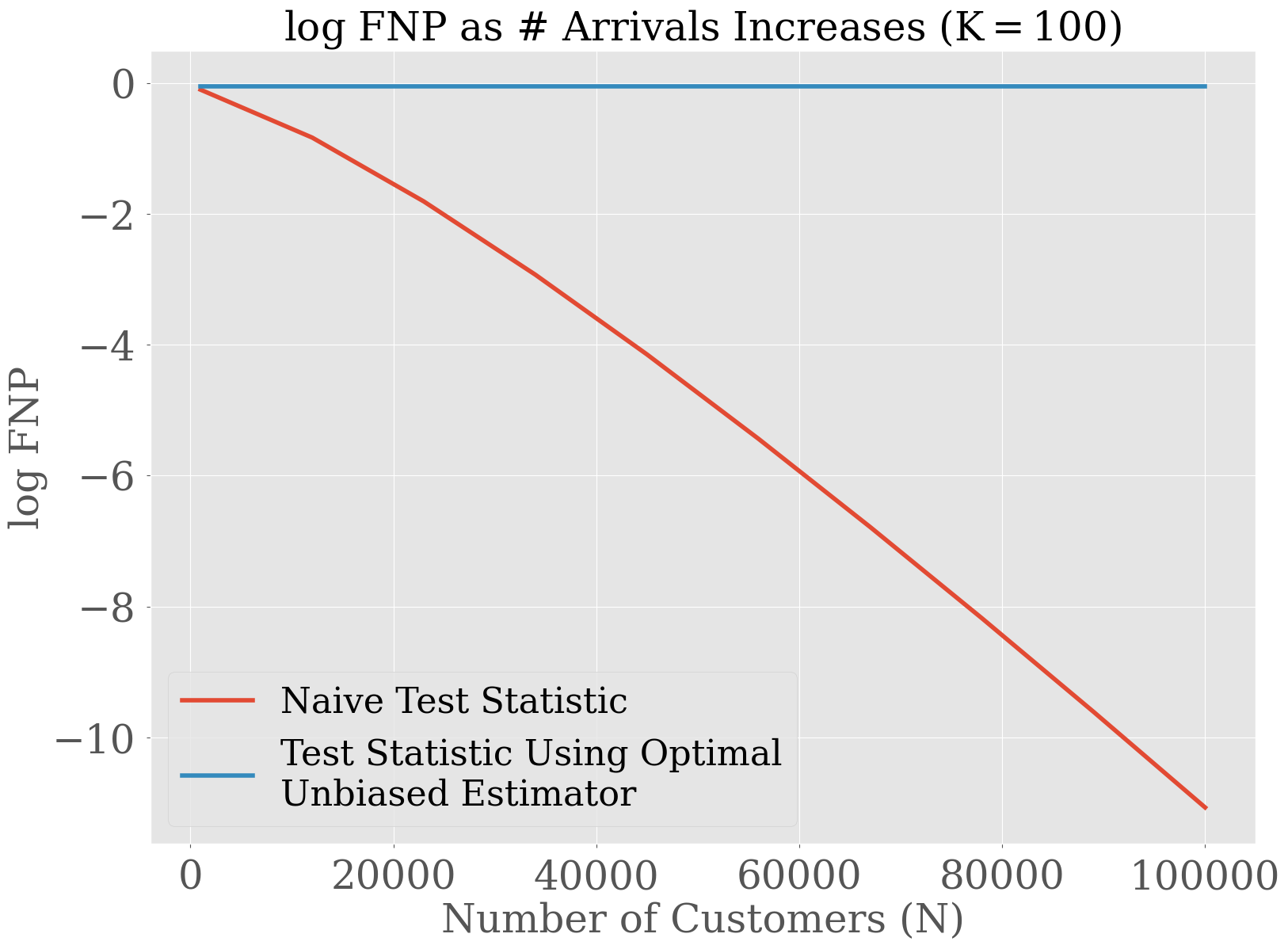}
    \includegraphics[width=0.49\textwidth]{Sections/figures/power_vary_listings/K=200.png}
    \includegraphics[width=0.49\textwidth]{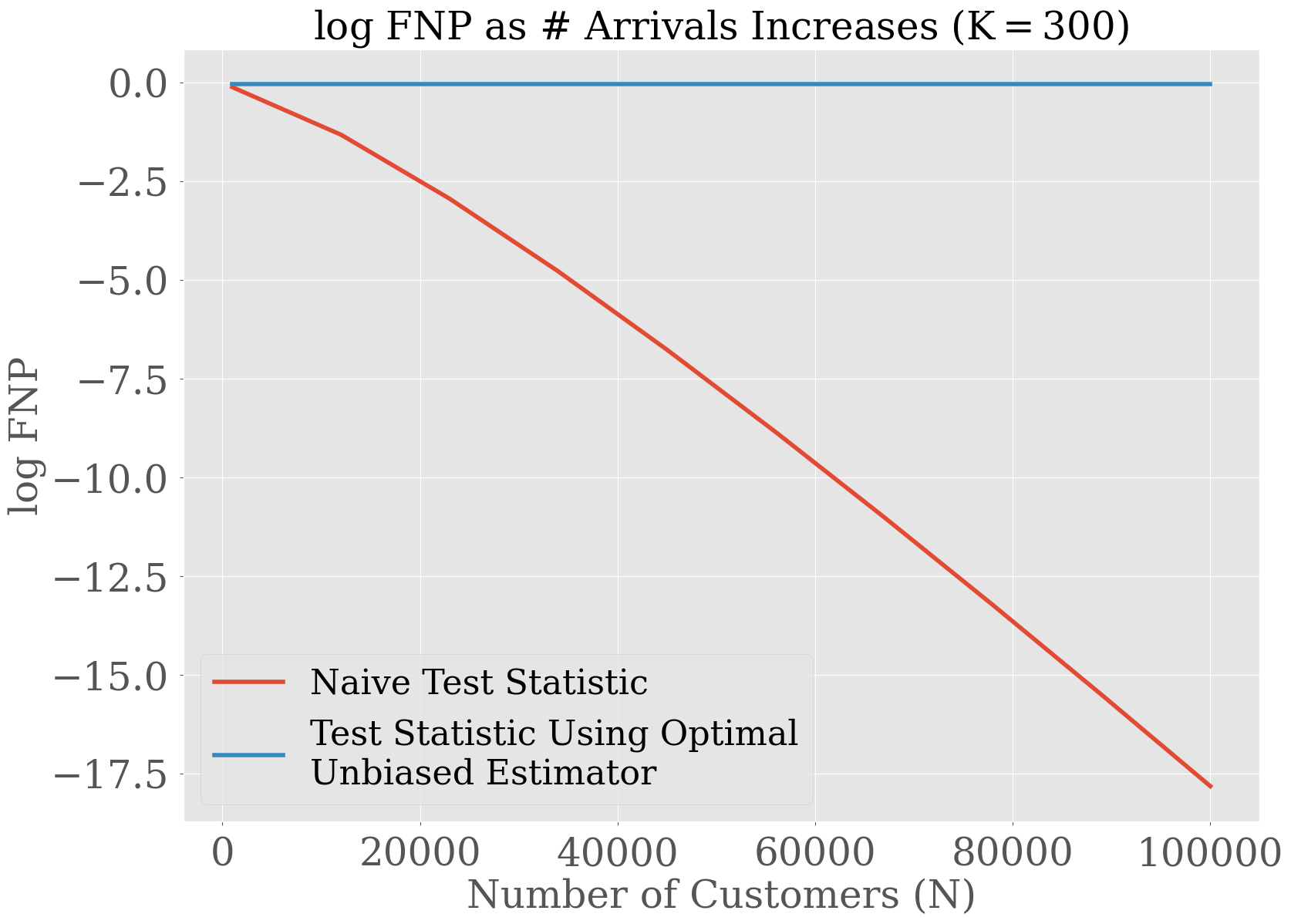}
    \includegraphics[width=0.49\textwidth]{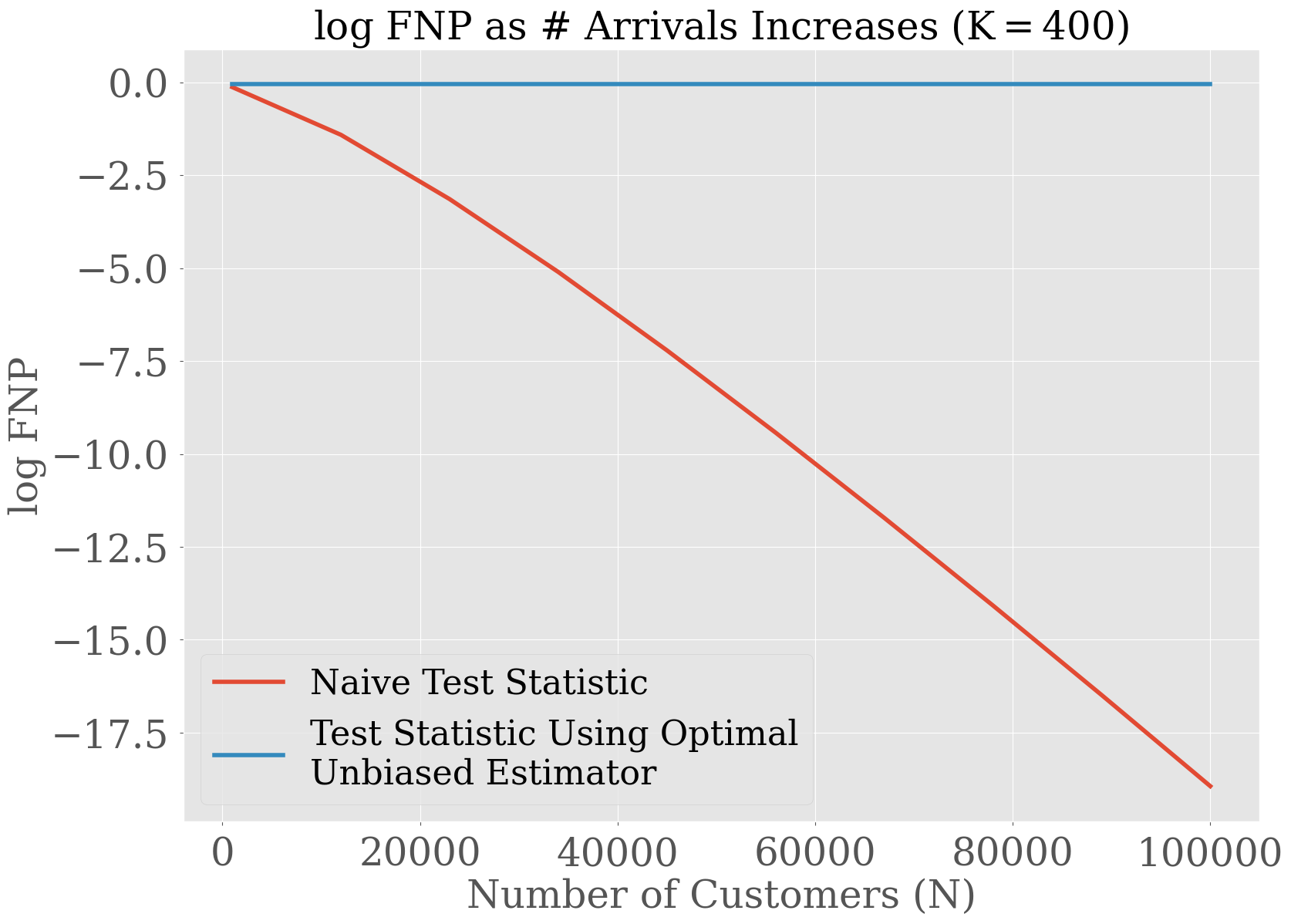}    
    \caption{  Log false negative probability of decision rule \eqref{eq:decision} using the test statistic $\hat{T}_N$ (cf.~\eqref{eq:tstat}), compared to the log false negative probability  of an unbiased test statistic (cf.~\eqref{eq:ub_power}) using an unbiased estimator that obeys \eqref{eq:unbiased_clt}. Bernoulli randomized user-level experiments with $N\in[10^3,10^5]$, $K \in\{100,200,300,400\}$, and other model parameters fixed as in Section \ref{sec:sims setup}. }
    \label{fig:power vary K}
\end{figure}

\begin{figure}[ht!]
  \centering
    \includegraphics[width=0.49\textwidth]{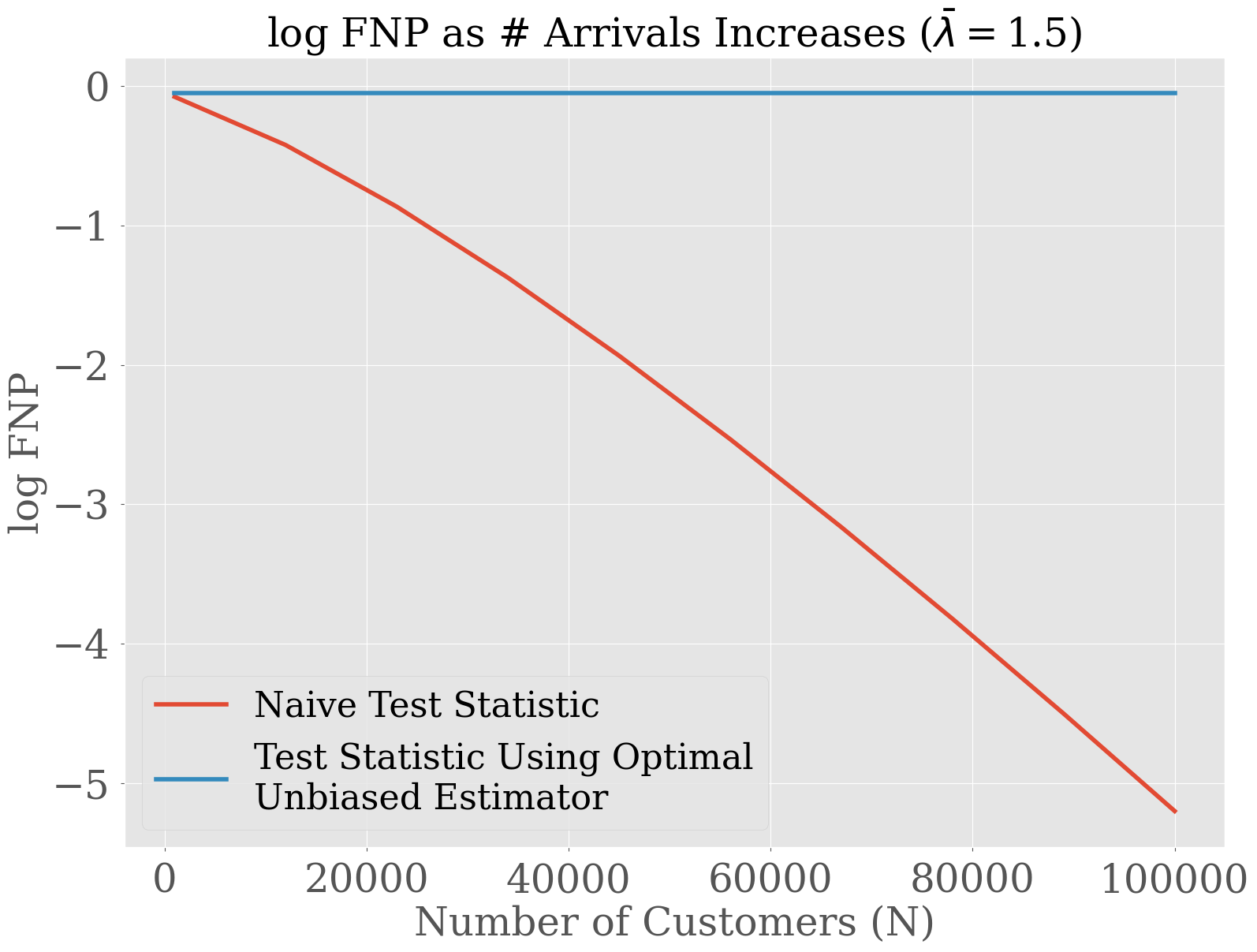}
    \includegraphics[width=0.49\textwidth]{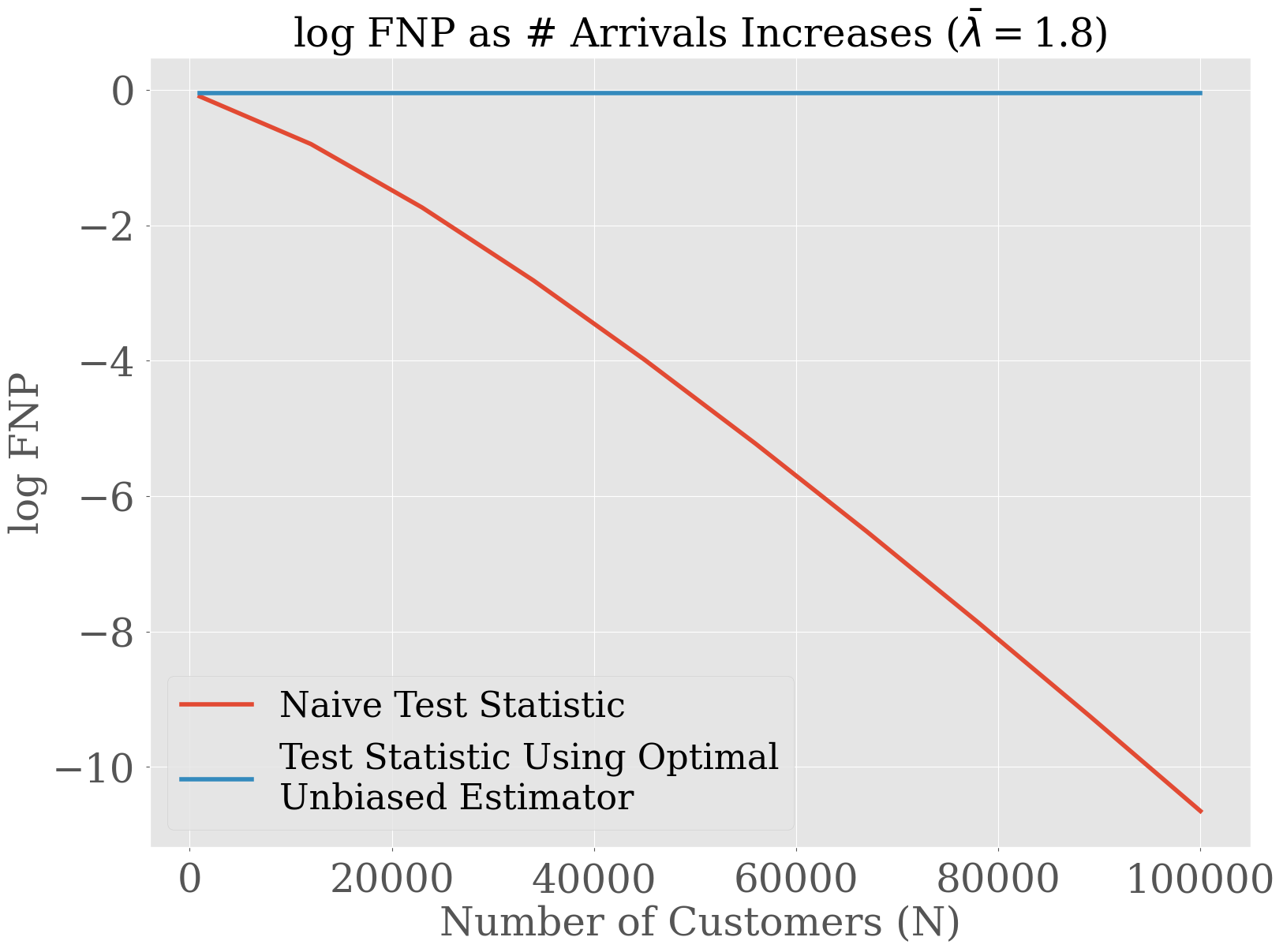}
    \includegraphics[width=0.49\textwidth]{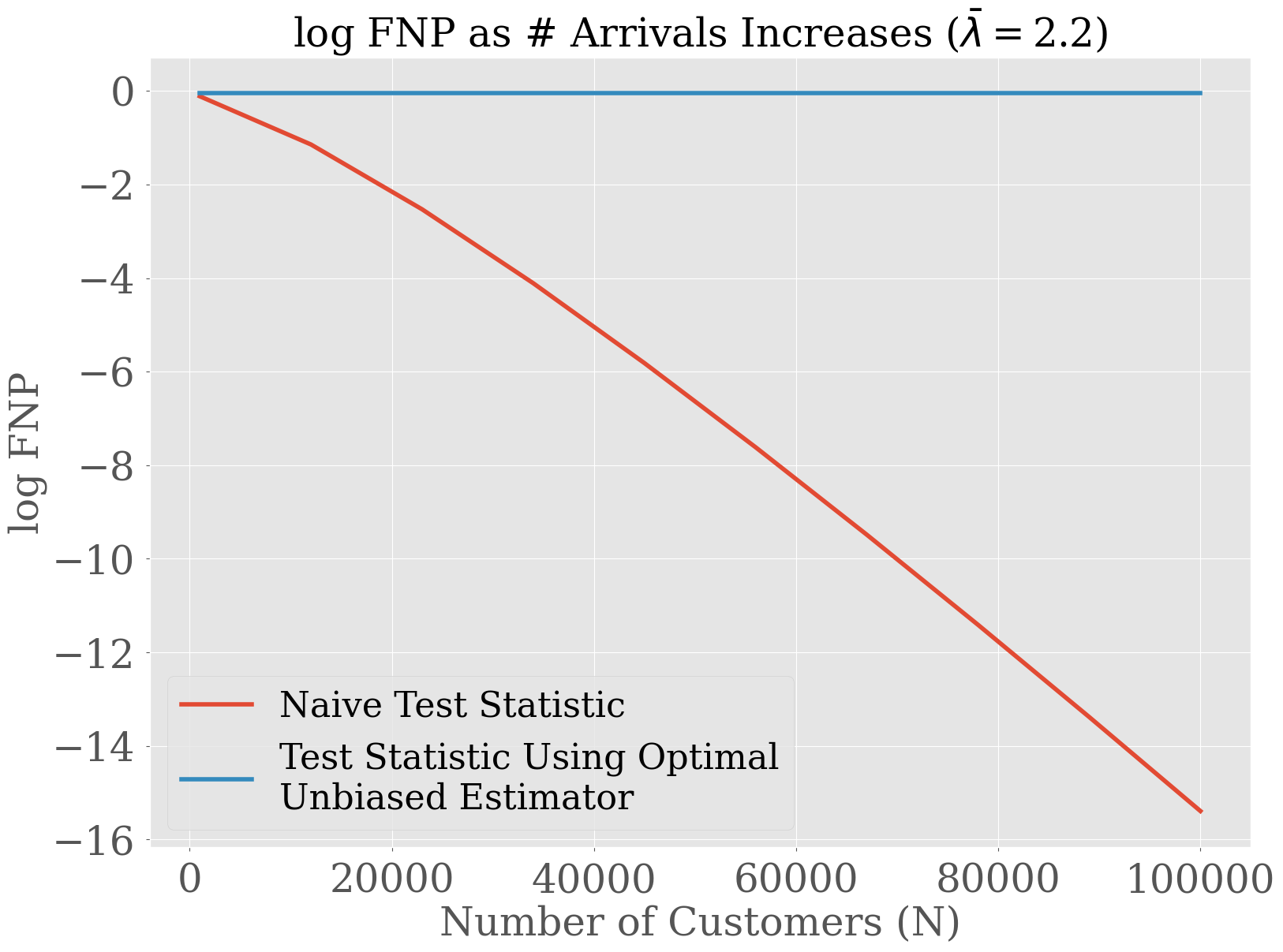}
    \includegraphics[width=0.49\textwidth]{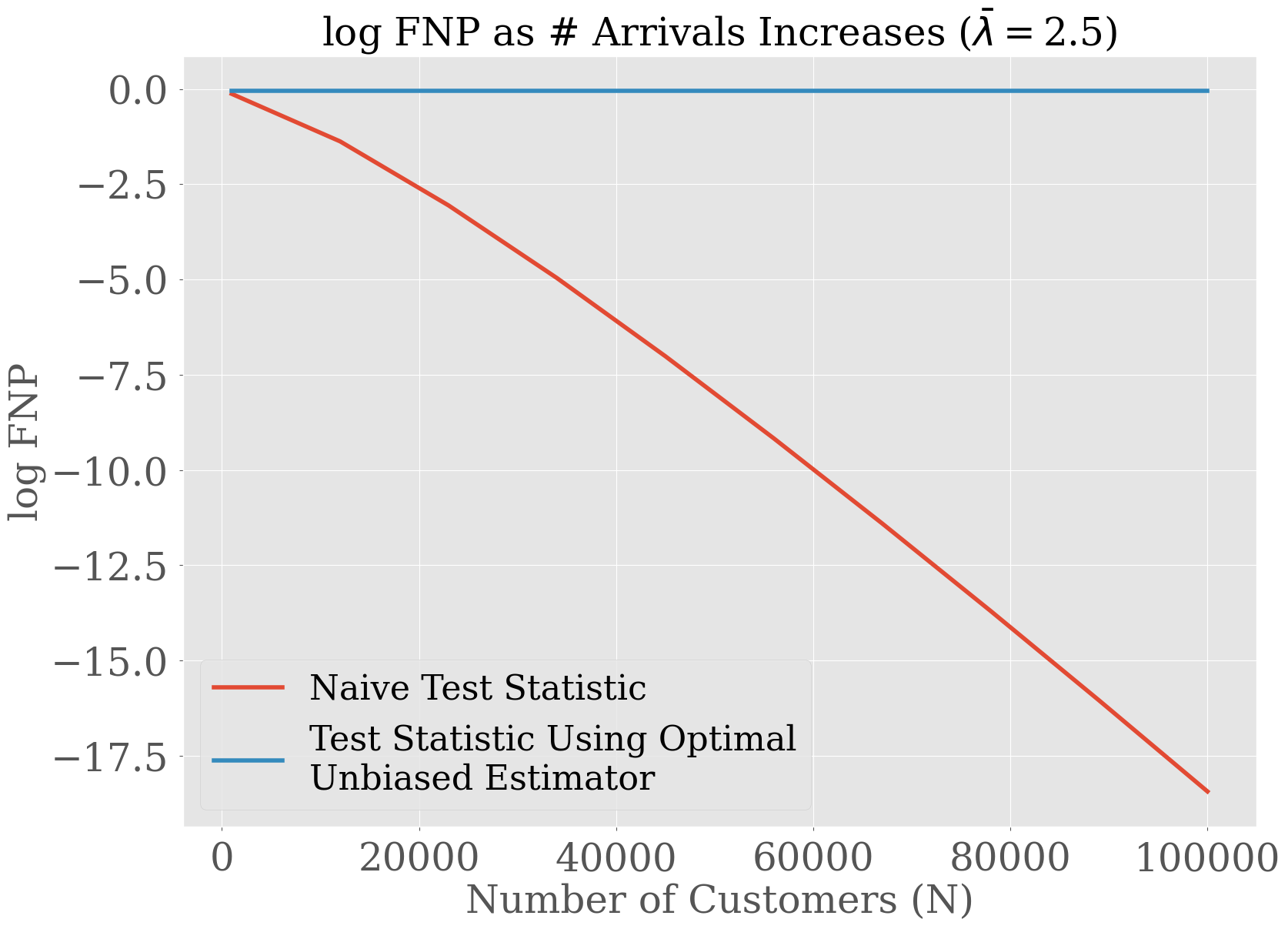}
    \caption{Log false negative probability of decision rule \eqref{eq:decision} using the test statistic $\hat{T}_N$ (cf.~\eqref{eq:tstat}), compared to the log false negative probability  of an unbiased test statistic (cf.~\eqref{eq:ub_power}) using an unbiased estimator that obeys \eqref{eq:unbiased_clt}. Bernoulli randomized user-level experiments with $N\in[10^3,10^5]$, $\bar\lambda\in\{1.5,1.8,2.2,2.5\}$, and other model parameters fixed as in Section \ref{sec:sims setup}..}
    \label{fig:power vary lambda}
\end{figure}

\begin{figure}[ht!]
  \centering
    \includegraphics[width=0.49\textwidth]{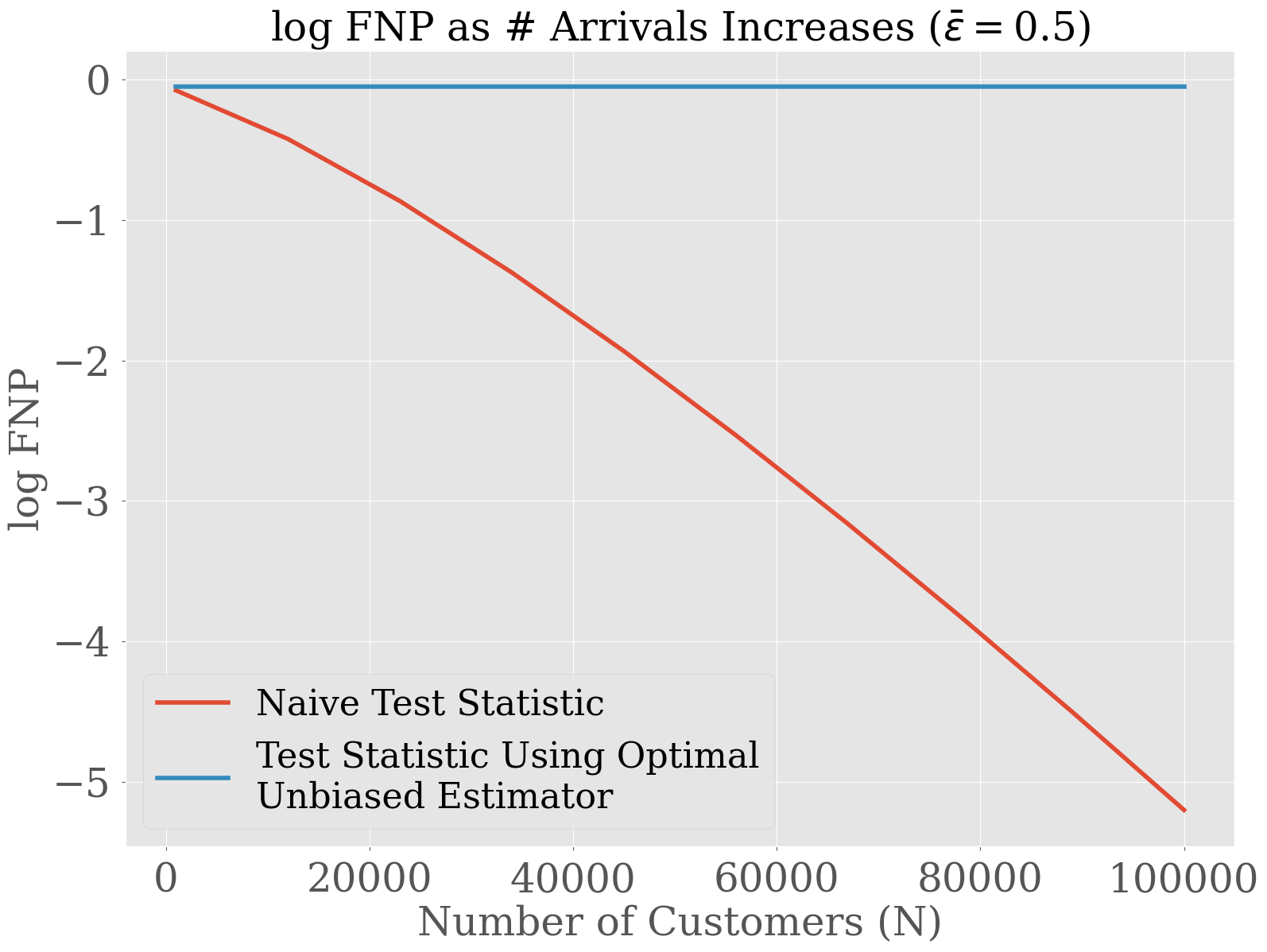}
    \includegraphics[width=0.49\textwidth]{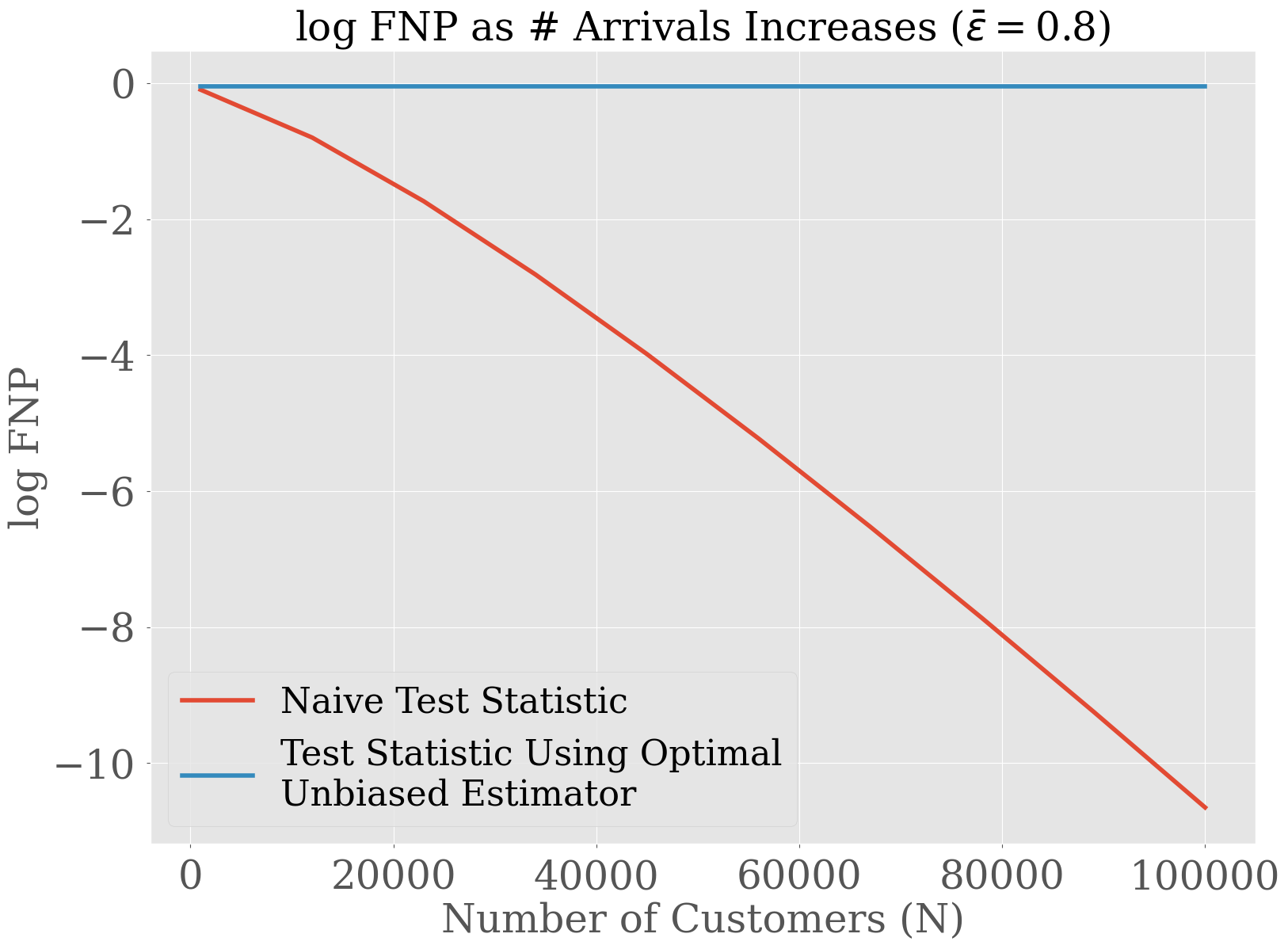}
    \includegraphics[width=0.49\textwidth]{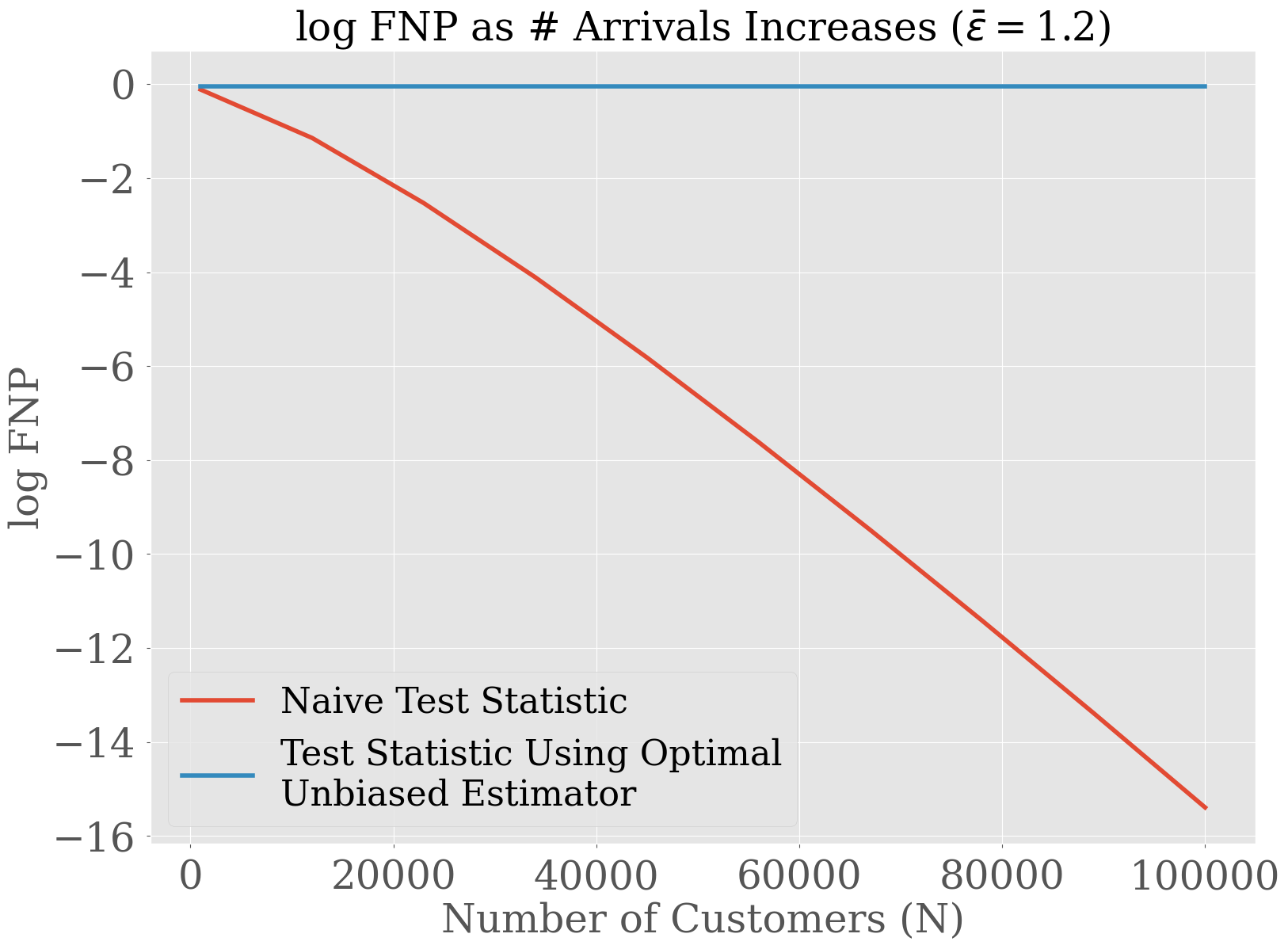}
    \includegraphics[width=0.49\textwidth]{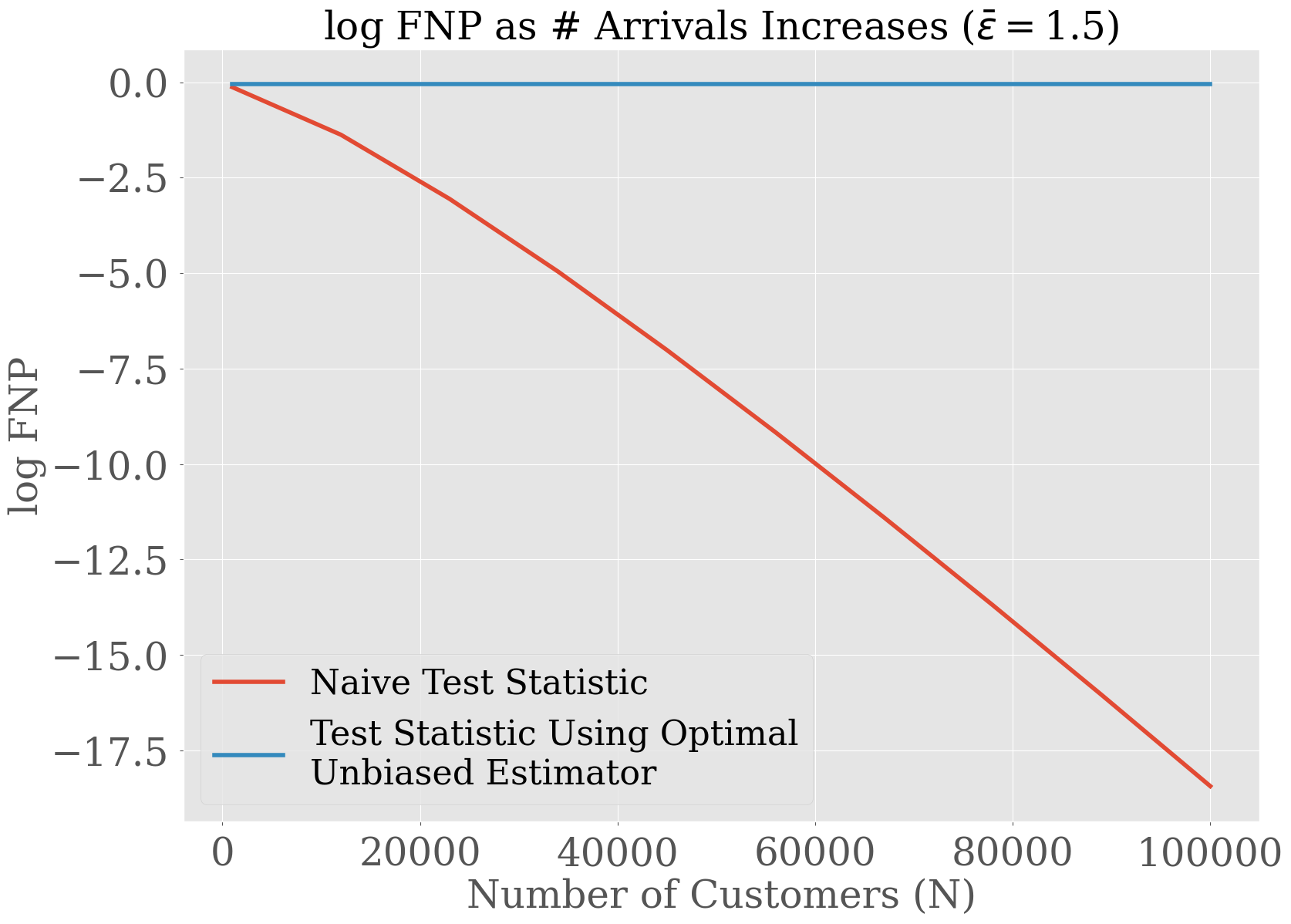}
    \caption{Log false negative probability of decision rule \eqref{eq:decision} using the test statistic $\hat{T}_N$ (cf.~\eqref{eq:tstat}), compared to the log false negative probability  of an unbiased test statistic (cf.~\eqref{eq:ub_power}) using an unbiased estimator that obeys \eqref{eq:unbiased_clt}. Bernoulli randomized user-level experiments with $N\in[10^3,10^5]$, $\bar\epsilon \in\{0.5,0.8,1.2,1.5\}$, and other model parameters fixed as in Section \ref{sec:sims setup}.}
    \label{fig:power vary a}
\end{figure}

\begin{figure}
  \centering
    \includegraphics[width=0.49\textwidth]{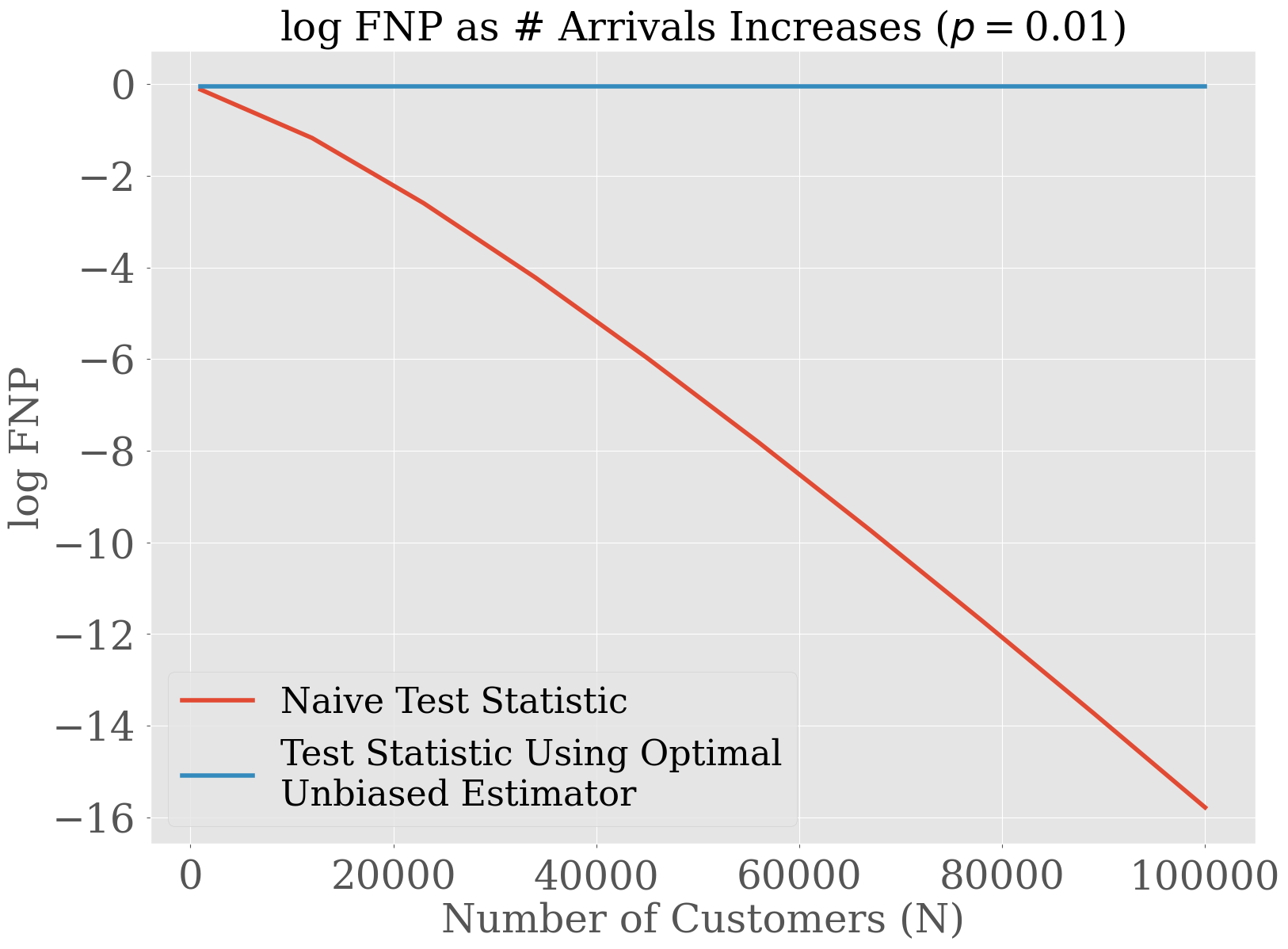}
    \includegraphics[width=0.49\textwidth]{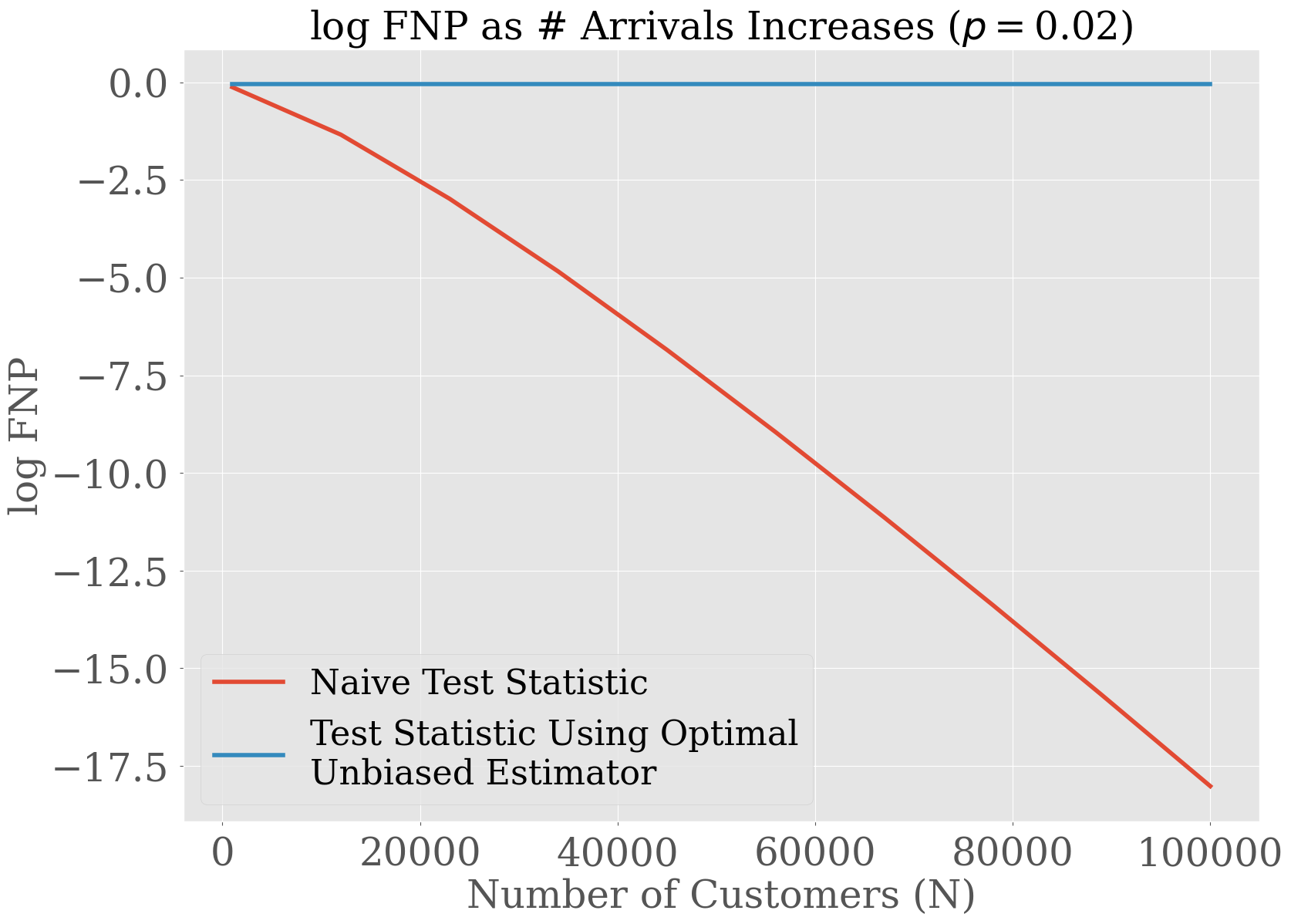}
    \includegraphics[width=0.49\textwidth]{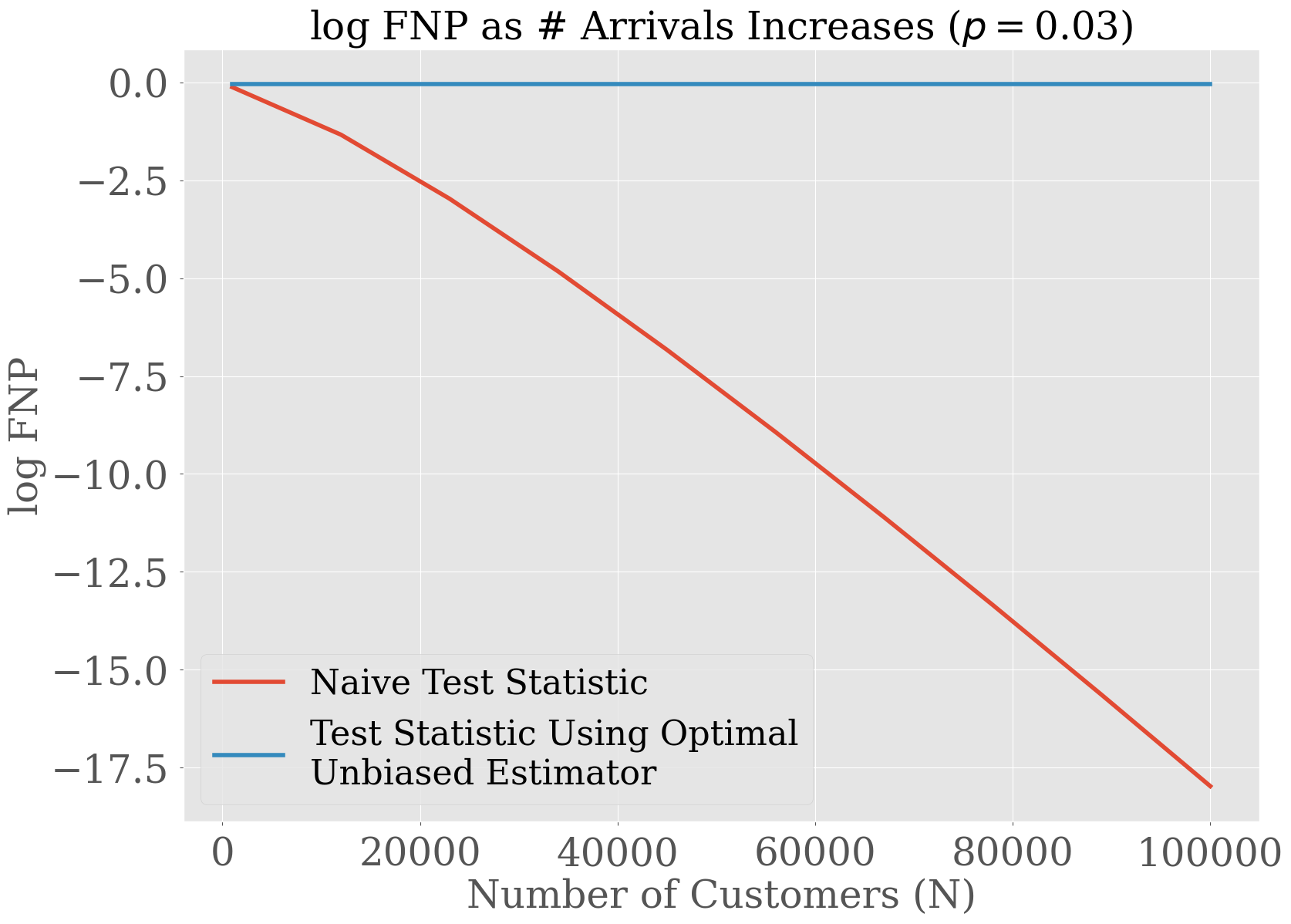}
    \includegraphics[width=0.49\textwidth]{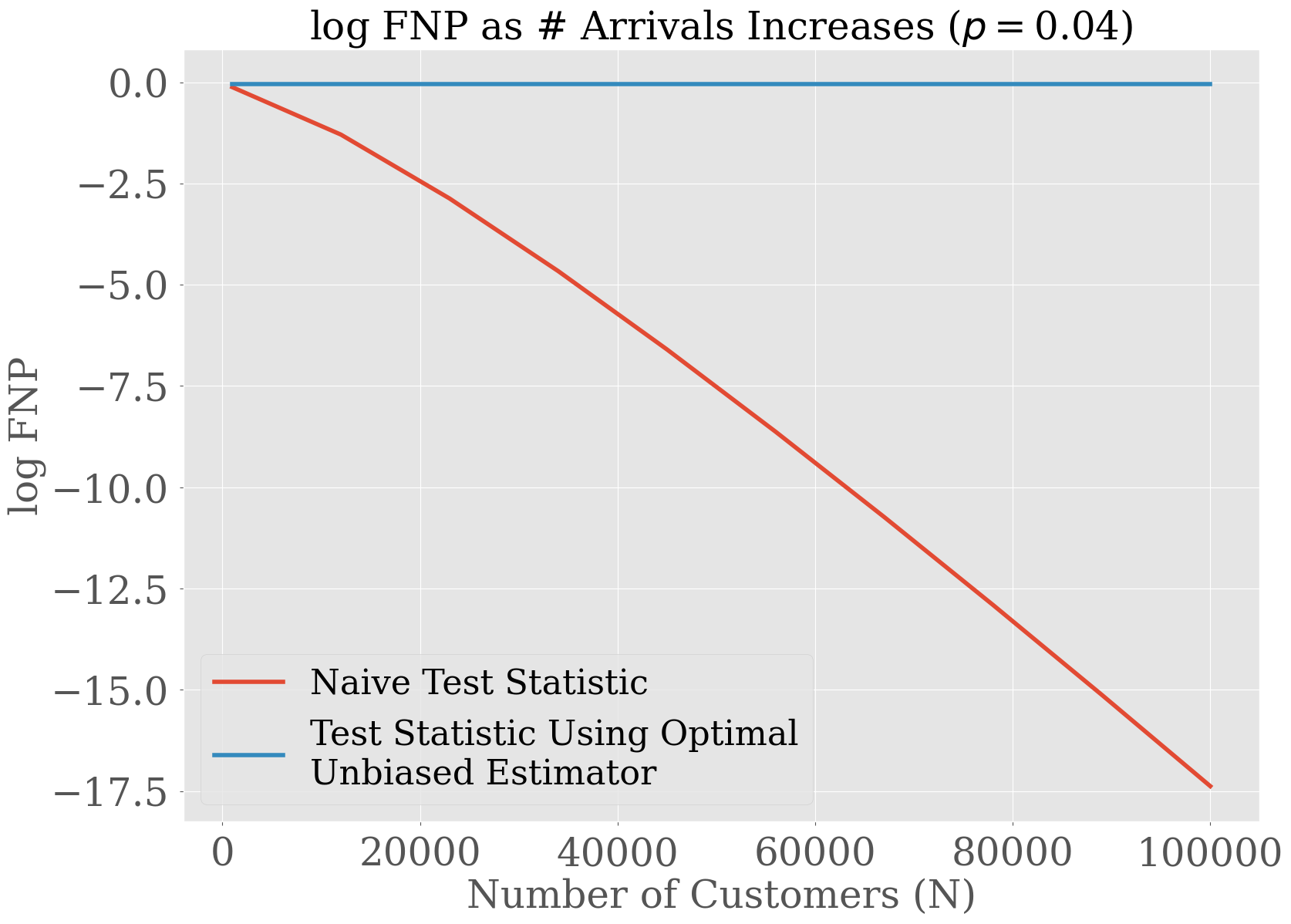}
    \caption{Log false negative probability of decision rule \eqref{eq:decision} using the test statistic $\hat{T}_N$ (cf.~\eqref{eq:tstat}), compared to the log false negative probability  of an unbiased test statistic (cf.~\eqref{eq:ub_power}) using an unbiased estimator that obeys \eqref{eq:unbiased_clt}. Bernoulli randomized user-level experiments with $N\in[10^3,10^5]$, $p \in\{0.01,0.02,0.03,0.04\}$, and other model parameters fixed as in Section \ref{sec:sims setup}.}
    \label{fig:power vary treatment}
\end{figure}